\documentclass[12pt,letterpaper]{article}
\usepackage[T1]{fontenc}
\usepackage[utf8]{inputenc}
\usepackage{lmodern}
\usepackage{amsmath,amssymb,amsthm}
\usepackage{graphicx}
\graphicspath{{figures/}}

\usepackage{natbib}
\usepackage{url}

\usepackage{algorithm}
\usepackage{microtype}
\usepackage{algpseudocode}

\def\spacingset#1{\renewcommand{\baselinestretch}{#1}\small\normalsize}
\spacingset{1}

\theoremstyle{plain}
\newtheorem{theorem}{Theorem}
\newtheorem{lemma}{Lemma}

\newtheorem{proposition}{Proposition}
\theoremstyle{definition}
\newtheorem{definition}{Definition}

\theoremstyle{remark}

\newcommand{\E}{\mathbb{E}}

\newcommand{\anon}{1}
\if1\anon
\fi
\begin{document}
\spacingset{1}
\date{}

\if1\anon
\title{\bfseries Information--Computation Inversion in Pseudo-Marginal MCMC}
\author{Zihan Xu\thanks{Corresponding author: \protect\url{haniizihanxu@gmail.com}.}\\
School of Mathematics and Statistics, Qingdao University\\
Qingdao, Shandong 266071, China}
\maketitle
\vspace{-36pt}
\else
\bigskip\bigskip\bigskip
\begin{center}
{\LARGE\bfseries Information--Computation Inversion in Pseudo-Marginal MCMC}
\end{center}
\medskip
\fi
\spacingset{1}

\begin{abstract}
\normalsize
Observation refinement changes posterior uncertainty and the likelihood
calculation in pseudo-marginal MCMC. We compare their combined effect
through finite-run squared-error risk. A sufficient inversion condition
relates information gain to accepted event flow and coarse-kernel
contraction. A bootstrap construction realizes inversion at every fixed
particle count. We then couple cross-event proposals and refresh
same-event proposals independently. A swap identity establishes invariance; continuation identities
describe subsequent risk. In the same finite
model, a rational certificate proves inversion against optimized constant
mixtures and repair by selective allocation over an initialization class
at a common action-price budget. Reaction-network experiments measure CPU
costs. Under finite-pool initialization, selective allocation reduces
event mean-squared error by 40.5\% against a tuned mixture at 25
post-initialization CPU seconds. Paired transcription observations show
how increased particle effort can raise finite-budget error.
\end{abstract}

\medskip
\spacingset{1}

\section{Introduction}
\label{sec:introduction}

An observation changes both a posterior distribution and the computation
used to estimate it. In a latent stochastic model, refining the data can
increase the variability of a particle-filter likelihood estimate. At a
fixed computational budget, the resulting Monte Carlo error enters the
accuracy of the reported posterior quantity. We study this joint effect
of observation information and computation.

Stochastic transcription provides a concrete setting. Population
snapshots, transcript capture, allele resolution, and new-RNA labeling
reveal different aspects of transcription kinetics
\citep{gomezschiavon2017,tiberi2018,tang2023,ramskold2024,gu2025}.
Sampling and observation models affect Fisher information
\citep{komorowski2011,foxmunsky2019}. Informative observations can also
make forward particle simulation inefficient \citep{golightly2014mjp}.
These effects motivate a comparison that includes both posterior
uncertainty and the error of a finite computation.

Pseudo-marginal MCMC retains a nonnegative unbiased likelihood estimate
in an extended Markov state. Its invariant parameter marginal is the
posterior \citep{andrieu2009pseudo,andrieu2010pmcmc}. Likelihood noise
and particle effort affect efficiency \citep{doucet2015efficient,sherlock2015}.
Weight moments govern convergence regimes \citep{andrieu2015convergence},
and convex-order comparisons order acceptance and asymptotic variance
under a common proposal \citep{andrieu2015convex}. Correlation and
blocking introduce dependence between successive likelihood estimates
\citep{tran2017block,deligiannidis2018cpm,golightly2019cpmskm}.
Weak-Poincar\'e results relate weight tails to convergence and computation
\citep{andrieu2022comparison,andrieu2026wpi}. Function-specific mixing
theory studies recovery of individual expectations \citep{rabinovich2020}.
Recent work develops tuning rules, adaptive particle effort, and
estimator failure control
\citep{sherlock2024variance,abaoubida2025adaptive,sherlock2026frankenfilter}.

Our first question concerns observation refinement. Section~\ref{sec:controlled-realizability}
separates posterior uncertainty from computational mean-squared error
(MSE), with the initial law and stopping rule specified. Section~\ref{sec:information-computation-inversion}
gives a sufficient condition for a finer observation to increase their
sum. It compares the information gained with accepted event flow under
the fine kernel and contraction under the coarse kernel. A finite
bootstrap model satisfies this condition with informative observations.
For every fixed particle count, a finite resolution produces a reversal
under the same parameter proposal in both channels.

Our second question is how to allocate auxiliary work for a specified
posterior event. Section~\ref{sec:route-allocation} couples proposals
that cross the event boundary and refreshes same-event proposals
independently. A residual tilt separately changes the parameter proposal.
We establish invariance through an auxiliary swap identity. Finite-run
and budget continuation identities then describe how within-event
updates change later event estimates. In the same bootstrap model, a
rational certificate connects this allocation to observation refinement:
the optimized fine-channel constant mixture has higher total risk than
the optimized coarse mixture, whereas selective allocation brings the
fine-channel risk below that benchmark. The result covers every constant
mixture probability and a class of initial laws at a common action-price
budget.

Section~\ref{sec:targeted-repair} examines measured computation costs in
two reaction models. A gene-network study compares the component methods
and a tuned constant mixture. New initialization draws provide
confirmation; additional observation records examine variation across
datasets. A paired transcription study evaluates the same posterior
event under allele-specific observations and their total-count
coarsening. Its particle-count comparison measures how likelihood effort
changes the number of completed transitions and the resulting risk.

\section{Posterior functionals and finite-run risk}
\label{sec:controlled-realizability}

\subsection{Posterior functionals and the pseudo-marginal state}
\label{subsec:posterior-functional-pm-state}

Let \(y\) denote the observed data, \(p(\theta)\) a prior density, and
\(L(\theta)=p(y\mid\theta)\) the likelihood. The posterior is
\[
\pi(\theta)
=
\frac{p(\theta)L(\theta)}
{\int p(\vartheta)L(\vartheta)\,d\vartheta}.
\]
The inferential target in this paper is a scalar posterior functional
\[
\Psi
=
\pi(\psi)
=
\int \psi(\theta)\pi(\theta)\,d\theta,
\]
for a measurable function \(\psi\in L^2(\pi)\). Posterior probabilities correspond
to indicator choices of \(\psi\).

Suppose that \(U\sim M_\theta\) generates a nonnegative unbiased
likelihood estimate,
\[
\widehat L(\theta,U)\ge0,
\qquad
\int
\widehat L(\theta,u)M_\theta(du)
=
L(\theta).
\]
Pseudo-marginal Metropolis--Hastings targets the extended distribution
\[
\bar\pi(d\theta,du)
\propto
p(\theta)
\widehat L(\theta,u)
M_\theta(du)\,d\theta,
\]
whose parameter marginal is \(\pi\)
\citep{andrieu2009pseudo,andrieu2010pmcmc}. The auxiliary variable is
part of the Markov state. A rejected proposal therefore retains the
current likelihood estimate as well as the current parameter.

When \(L(\theta)>0\), write
\[
W
=
\frac{\widehat L(\theta,U)}{L(\theta)}
\]
and let \(Q_\theta\) denote its law under \(M_\theta\). Then
\[
\int w\,Q_\theta(dw)=1,
\]
while the conditional law of the retained weight under the extended
target is \(wQ_\theta(dw)\). Section~\ref{sec:information-computation-inversion}
uses this retained law to study fixed-horizon functional risk.

\subsection{Fixed-horizon functional recovery}
\label{subsec:fixed-horizon-functional-recovery}

For a finite run, accuracy depends on the initial state and the number
of transitions. We measure it by the squared error of the reported
posterior-functional estimate.

\begin{definition}[Fixed-horizon functional risk]
\label{def:fixed-horizon-functional-risk}
Let \(K\) be a Markov kernel on an extended state space with parameter
component \(\Theta_t\), and let \(\nu\) be its initial distribution. For
\(B\ge1\), define
\[
\widehat\Psi_B
=
\frac1B
\sum_{t=1}^{B}
\psi(\Theta_t)
\]
and
\[
\mathcal R_B(K,\nu;\psi)
=
\mathbb E_{\nu,K}
\left[
(\widehat\Psi_B-\Psi)^2
\right].
\]
When the chain starts from its pseudo-marginal extended target and the
kernel is clear from context, we write \(\mathcal R_B(\psi)\).
\end{definition}

The transition horizon \(B\) is part of the inferential question.
Particle-filter work and wall-clock time depend on the transition and
are reported separately in the empirical comparisons.

The transcription study evaluates the posterior event probability by
finite-state likelihood integration. The gene-network study uses an
independent importance-integration reference and propagates its Monte
Carlo uncertainty. The experimental sections specify both references and
their numerical checks.

\subsection{Functional risk at a computational budget}
\label{subsec:budget-risk}

For an event $h$, set $\Psi=\pi(h)$ and write $h(X)=h(\theta)$.
Let $C_0$ include the charged initialization and calibration costs and $C_j\geq0$ the cost of
transition $j$. Suppose $T_n=C_0+\sum_{j=1}^n C_j\to\infty$ almost surely.
At budget $t$, define
\[
 N_K(t)=\max\bigl(\{0\}\cup\{n\geq1:T_n\leq t\}\bigr),\qquad
 \widehat\Psi_{K,t}=\frac1{N_K(t)}\sum_{j=1}^{N_K(t)}h(X_j)
\]
when $N_K(t)>0$. If initialization is complete but no transition has
finished, set $\widehat\Psi_{K,t}=h(X_0)$. If initialization is
incomplete or fails its prescribed limit, use a declared measurable
fallback in $[0,1]$, based on the observation and available preparation
only. The initialization, calibration, timing and fallback conventions
are fixed before a comparison.
The budget risk is
\[
 \mathcal R^{\rm time}_t(K,\nu;h)
 =\E_{\nu,K}[\{\widehat\Psi_{K,t}-\Psi\}^2].
\]
Its expectation includes the joint law of the states and their costs.
Any timing history needed to predict that law is included in the state.
A transition that completes after the budget contributes no event value.
The finite model uses declared action prices after initialization;
the experimental sections state which preparation costs enter each
measured budget.
Completion counts can depend on the path \citep{murray2021anytime}.
We evaluate this risk from completed prefixes at fixed checkpoints.
The Supplementary Material gives its joint state--cost recursion and
the timing convention used in the experiment.

\subsection{Observation information and total risk}
\label{subsec:observation-risk}

Let $F$ be a fine observation and $C=T(F)$ its coarsening. For
$H=\psi(\Theta)$ with $\E H^2<\infty$, write
$p_Y=\E(H\mid Y)$ and $v_Y=\operatorname{Var}(H\mid Y)$,
where $Y\in\{C,F\}$. A procedure returns $A_Y=a_Y(Y,U_Y)$.
Its randomization $U_Y$ is independent of $(\Theta,F)$; its specified
initialization, tuning, and stopping rule are included in $a_Y$.
Assume $\E A_Y^2<\infty$. Set
$e_Y(y)=\E\{(A_Y-p_Y(y))^2\mid Y=y\}$.
Conditional squared-error decomposition gives
\begin{equation}
 j_Y(y):=\E\{(A_Y-H)^2\mid Y=y\}=v_Y(y)+e_Y(y).
 \label{eq:obs-own-risk}
\end{equation}
The computational term includes initialization bias and sampling variance.
Under prior predictive averaging, write $V_Y=\E v_Y$ and
$J_Y=\E j_Y(Y)$. Then
\begin{equation}
 J_F-J_C=\E e_F(F)-\E e_C(C)-G,
 \qquad G=\E(p_F-p_C)^2\geq0.
 \label{eq:obs-integrated-risk}
\end{equation}
The tower property gives $G=\E v_C-\E v_F$.
Thus the statistical gain and the computational excess enter one loss.

For a fixed pair $(f,c)$, a common full-observation evaluation conditions
both procedures on $F=f$. Put $b_C(c)=\E(A_C\mid C=c)-p_C(c)$ and
$\delta=p_C(c)-p_F(f)$. The corresponding risk difference is
\begin{equation}
 r_F(f)-r_C(f)=e_F(f)-e_C(c)-\delta^2-2\delta b_C(c).
 \label{eq:obs-common-risk}
\end{equation}
The coarse procedure still receives $c$. Equations
\eqref{eq:obs-own-risk} and \eqref{eq:obs-common-risk} specify the two
conditional comparisons; their prior predictive averages agree.
Complete derivations are in the Supplementary Material.

At a fixed observation, algorithms with the same posterior share $v_Y$.
Their total-risk difference equals their computational-MSE difference.
Section~\ref{sec:route-allocation} uses this equality to compare
auxiliary allocations. The same-model repair theorem then compares
the resulting fine-channel risk with an optimized coarse-channel
benchmark.

\section{Observation refinement and finite-run error}
\label{sec:information-computation-inversion}
\subsection{Information gain and event flow}
\label{subsec:refinement-flow}

Fix an event $h:E\to\{0,1\}$, an observation $y$, and a kernel $K_y$
with extended invariant law
$\Pi_y$. Write $p_y=\Pi_y h$, $v_y=p_y(1-p_y)$, and define the accepted
stationary event flow
\[
 \mathfrak q_y=\int_{\{h=1\}}\Pi_y(dx)K_y(x,\{h=0\}).
\]
This quantity includes the likelihood update and the acceptance decision.
For a stationary run of $B$ transitions, let $e_{B,y}$ be the
computational MSE in Section~\ref{subsec:observation-risk}. Set
\[
 c_B=\frac{B^2-1}{3B},\qquad
 D_B(a)=\frac{B+2\sum_{k=1}^{B-1}(B-k)a^k}{B^2}.
\]
Stationarity gives equal flow in the two directions. The probability that
the event changes in $k$ transitions is at most $2k\mathfrak q_y$.
Consequently,
\begin{equation}
 e_{B,y}\geq(v_y-c_B\mathfrak q_y)_+.
 \label{eq:flow-risk-lower}
\end{equation}
This bound applies to local parameter proposals and correlated auxiliary
updates. It uses stationarity; it requires neither reversibility nor a
positive Markov operator.

\begin{theorem}[A sufficient condition for information--computation inversion]
\label{thm:general-refinement}
Let $C=T(F)$ and fix a coarse outcome $c$ with $v_c>0$. Both algorithms
start from their respective extended posteriors and use $B$ transitions.
Define
\[
 \bar v_F=\E(v_F\mid C=c),\qquad
 G_c=v_c-\bar v_F,\qquad \gamma_c=G_c/v_c.
\]
Suppose $\mathfrak q_f\leq r v_f$ for almost every $f$ conditional on $c$.
Suppose also that, for every centered $g\in L^2(\Pi_c)$,
\[
 \|K_c g\|_2\leq a_c\|g\|_2,\qquad 0\leq a_c<1.
\]
Put $L_B=(1-c_Br)_+$ and $d_c=D_B(a_c)$. Then
\begin{equation}
 \E\{j_F(F)\mid C=c\}-j_C(c)
 \geq v_c\{(1-\gamma_c)(1+L_B)-(1+d_c)\}.
 \label{eq:general-inversion-margin}
\end{equation}
In particular, finer observations give strictly more event information
and strictly larger total risk if
\[
 L_B>d_c,\qquad 0<\gamma_c<\frac{L_B-d_c}{1+L_B}.
\]
\end{theorem}

The result compares three quantities: event information, accepted flow
under the fine kernel, and contraction under the coarse kernel. A
minorization $K_c(x,\cdot)\geq\epsilon_c\Pi_c(\cdot)$ supplies
$a_c=1-\epsilon_c$. Equations~\eqref{eq:flow-risk-lower} and
$e_{B,c}\leq v_cD_B(a_c)$ prove the theorem after conditioning on $C=c$.
Supplementary Section~S2 gives the full proof, a version with
observation-specific flow bounds, and a perturbation bound for initial
laws. A prior predictive conclusion follows by integrating the displayed
risk difference over $C$.

For fresh-estimator pseudo-marginal kernels, retained likelihood weights
provide another way to locate slow event changes. The supplementary
retained-state bound controls acceptance by reverse incoming proposal
mass divided by the retained normalized weight. That result requires a
bound on the incoming mass and the relevant weight tail. The finite
model below supplies direct flow and contraction bounds.
Supplementary Section~S5 also derives moment bounds for fresh guided
particle-filter weights in the network model.

\subsection{A refinement in which every cell carries information}
\label{subsec:observation-exact}

Let $\Theta\in\{-1,1\}$ have equal prior probabilities and let the
latent variable $Z\in[-1,1]$ have conditional density
\begin{equation}
 f_\theta(z)=\frac{1+\theta\varepsilon z}{2},\qquad 0<\varepsilon<1.
 \label{eq:informative-density}
\end{equation}
The coarse observation is the sign of $Z$. The fine observation records
its cell among $2K$ equal intervals, with $K\geq2$. For cells
$j=0,\ldots,2K-1$, write $x_j=-1+(2j+1)/(2K)$. For
$h(\theta)=\mathbf1\{\theta=1\}$,
\begin{equation}
 L_\theta(j)=\frac{1+\theta\varepsilon x_j}{2K},\qquad
 p_j=\frac{1+\varepsilon x_j}{2}.
 \label{eq:informative-likelihood}
\end{equation}
The cells have equal predictive probability. Their posterior event
probabilities are distinct. Any deterministic merger of two cells loses
event information. The average posterior variances and information gain
are
\begin{align}
 V_C&=\frac14-\frac{\varepsilon^2}{16},\qquad
 V_F=\frac14-\frac{\varepsilon^2}{12}
                    +\frac{\varepsilon^2}{48K^2},\nonumber\\
 G&=\frac{\varepsilon^2}{48}(1-K^{-2})>0.
 \label{eq:informative-variance}
\end{align}

Both channels use a prior parameter proposal and the matching estimator
\begin{equation}
 \widehat L_{Y,N}(\theta;y)=\frac1N\sum_{i=1}^N
       \mathbf1\{T_Y(Z_i)=y\},\qquad Z_i\mid\theta\stackrel{\rm iid}{\sim}f_\theta.
 \label{eq:obs-bootstrap}
\end{equation}
The retained state is $(\theta,k)$ with positive match count
$k\in\{1,\ldots,N\}$. A proposed zero count is rejected. This gives a
finite pseudo-marginal chain whose extended stationary probabilities are
\begin{equation}
 \Pi_y(\theta,k)=\frac{k b_{\theta,y}(k)}{2Nm_y},\qquad
 b_{\theta,y}(k)=\Pr\{\operatorname{Bin}(N,L_\theta(y))=k\},
 \label{eq:count-target}
\end{equation}
where $m_y=\{L_{-1}(y)+L_1(y)\}/2$.

For independent proposed counts $J_{-1},J_1$, direct summation gives
\[
 \mathfrak q_y=\frac{\E\min(J_{-1},J_1)}{4Nm_y}
       \leq Nm_yv_y.
\]
The independent Metropolis--Hastings operator is positive
\citep[Proposition 16]{andrieu2015convergence}. Its spectral representation
therefore gives $e_{B,y}\geq v_yD_B(1-\mathfrak q_y/v_y)$.
In the coarse channel, $m_c=1/2$ and $\widehat L\leq1$, so the kernel
minorizes $\Pi_c/2$ and $e_{B,c}\leq v_cD_B(1/2)$.
These are bounds for each channel's own posterior and kernel.

\begin{theorem}[Informative reversal at every fixed particle count]
\label{thm:obs-fixedN}
In \eqref{eq:informative-density}, take $\varepsilon=4/5$, any integer
$N\geq1$, and $K=2N$. Use \eqref{eq:obs-bootstrap}, independent auxiliary
updates, and $B=10$ transitions in both channels. Under extended-posterior
initialization,
\begin{equation}
 J_F-J_C>
 \frac{59}{300}\{1+D_{10}(3/4)\}
 -\frac{21}{100}\{1+D_{10}(1/2)\}
 =\frac{99024089}{3932160000}>0.02518.
 \label{eq:all-N-margin}
\end{equation}
The conclusion also holds if each channel's full initial law, at every
observation, is within total variation $1/40$ of its extended posterior.
In that case $J_F-J_C>458483/1310720000>0.00034979$.
\end{theorem}

Here $N$ is fixed and the resolution is $K=2N$. Each channel uses $10N$
latent draws during the ten transitions. The comparison concerns this
specified likelihood estimator and work count. The likelihoods in
\eqref{eq:informative-likelihood} also make the mechanism fully
calculable: finer cells reduce the chance of a matching proposal, while
their posterior information gain remains explicit. Theorem~\ref{thm:general-refinement}
already gives a positive lower bound $357/25600$ for $N=4$, $K=8$ and
$B=4$. Supplementary Section~S2 verifies both bounds.

The next section changes the auxiliary update within this same model.
It compares selective allocation with the full class of constant
mixtures at a common action-price budget.

\section{Route-conditioned auxiliary allocation}
\label{sec:route-allocation}

\subsection{Event routes and residual tilt}
\label{subsec:route-design}

An event functional changes value only when a transition crosses
its partition. This gives a direct way to allocate auxiliary work.
We retain a coupled update on cross-event proposals and refresh
independently on same-event proposals. The fully coupled comparator uses coupling on both routes.
Constant mixtures use an independent coin to select their auxiliary action.
All these comparisons use the same parameter proposal.

For an event $h$, write $h_0=h(\theta)$ and let
$\kappa_K(X)=\Pr_{X,K}\{h(\Theta_1)\ne h_0\}$ be the probability of
an accepted event crossing from $X=(\theta,u)$.

\begin{lemma}[Event crossing and auxiliary allocation]
\label{lem:event-crossing-risk}
For $\Psi=\pi(h)$ and any transition kernel $K$,
\begin{equation}
\label{eq:event-crossing-risk}
\mathcal R_1(K,\delta_X;h)
=(h_0-\Psi)^2+(1-2h_0)(1-2\Psi)\kappa_K(X).
\end{equation}
A larger $\kappa_K(X)$ strictly reduces this risk when $h_0=0$ and
$\Psi>1/2$, or $h_0=1$ and $\Psi<1/2$.
If two Metropolis--Hastings kernels share the parameter proposal,
cross-event auxiliary conditional law, and cross-event acceptance
rule at $X$, their one-step event distributions and conditional risks
agree. Their same-event auxiliary actions may differ.
\end{lemma}

The identity follows by conditioning on the two possible event values;
the full argument is in the Supplementary Material. It permits
auxiliary computation on same-event proposals to be changed while
preserving the one-step event law. Proposition~\ref{prop:allocation-risk} describes the effect of these
updates on subsequent event values and finite-run risk.

The residual tilt below changes the parameter proposal.
The event route separately assigns the auxiliary action.

Let
\[
h:\Theta\longrightarrow\{0,1\}
\]
be the measurable event function for the reported posterior probability,
and let
\[
\Psi
=
\E_{\pi}\{h(\Theta)\}.
\]
When the proposal is generated on an ambient coordinate space larger
than the target support, we fix a measurable extension of \(h\) to that
space. Proposals outside \(\Theta\) are assigned zero target density and
are rejected before likelihood evaluation. The extension partitions
proposal mass before this support rejection.

The event function induces a pair route
\[
\mathsf r(\theta,\theta')
=
\begin{cases}
\mathrm{cross},
&
h(\theta)\ne h(\theta'),
\\
\mathrm{same},
&
h(\theta)=h(\theta').
\end{cases}
\]
By construction,
\[
\mathsf r(\theta,\theta')
=
\mathsf r(\theta',\theta).
\]

Let \(X=(\theta,u)\) be the current extended pseudo-marginal state.
Choose a fixed centering function \(m:\Theta\to\mathbb R\), calibrated
independently of the transition outcomes to which the method will be
applied, and define the retained residual
\[
r(X)
=
\log\widehat L(\theta,u)-m(\theta),
\qquad
r_-(X)
=
\min\{r(X),0\}.
\]
Any fixed measurable \(m\) defines a valid state-dependent proposal. In
the gene-network experiment, a quadratic centering surface is used to
center the retained log likelihood.
The current state lies on the positive support of the extended target,
so its retained likelihood estimate is positive and the residual is
well defined.

Let \(q_0(\theta'\mid\theta)\) be a baseline parameter proposal on
the proposal space and write its cross-route mass as
\[
p_0(\theta)
=
\int
q_0(\vartheta\mid\theta)
\mathbf 1
\{
\mathsf r(\theta,\vartheta)=\mathrm{cross}
\}
\,d\vartheta,
\]
where the integral is over that proposal space. If the proposal is
support-restricted, the proposal space equals \(\Theta\).
The following proposal tilts the route mass:
\begin{equation}
\label{eq:route-tilted-proposal}
q_R(\theta'\mid X)
=
\frac{
q_0(\theta'\mid\theta)
\exp
\left[
-r_-(X)
\mathbf 1
\{
\mathsf r(\theta,\theta')=\mathrm{cross}
\}
\right]
}{
Z_R(X)
},
\end{equation}
where
\[
Z_R(X)
=
1-p_0(\theta)
+
p_0(\theta)e^{-r_-(X)}.
\]
A negative retained residual increases the probability assigned to
cross-route proposals. When \(r(X)\ge0\), the route probabilities agree
with the baseline proposal. Equation~\eqref{eq:route-tilted-proposal}
is an ordinary state-dependent Metropolis--Hastings proposal. The
reverse density is evaluated at the proposed extended state.

The parameter route also selects the auxiliary action. Same-route
proposals use an independent auxiliary draw,
\[
C^{\mathrm{refresh}}_{\theta,\theta'}(u,du')
=
M_{\theta'}(du'),
\]
where \(M_{\theta'}\) is the ordinary likelihood-estimator law.
Cross-route proposals use a coupled inherited draw,
\[
u'
\sim
C^{\mathrm{inherit}}_{\theta,\theta'}(u,\cdot).
\]
For exactness, the joint law
\(M_\theta(du)C^{\mathrm{inherit}}_{\theta,\theta'}(u,du')\)
must have second marginal \(M_{\theta'}\) and satisfy the swap
identity in Section~\ref{subsec:route-validity-local}. In the
stochastic-network implementation this joint coupling uses split
Markov-jump propagation and coupled particle ancestry
\citep{anderson2018split,jacob2016coupling}.

Refreshment \citep{maire2014refreshment}, correlated pseudo-marginal
updates \citep{deligiannidis2018cpm}, and observable-directed process
couplings \citep{arampatzis2014goal} provide the component operations.
Auxiliary information also enters randomized and asymmetric Metropolis
constructions \citep{nicholls2012randomized,andrieu2021mhaar}. Here the event route assigns the auxiliary
action, and the retained residual changes cross-route proposal mass.

\subsection{The route-conditioned transition}
\label{subsec:route-transition}

Combining the two auxiliary actions gives the state-dependent allocation
\[
C^{\mathsf r(\theta,\theta')}_{\theta,\theta'}
=
\begin{cases}
C^{\mathrm{inherit}}_{\theta,\theta'},
&
\mathsf r(\theta,\theta')=\mathrm{cross},
\\
C^{\mathrm{refresh}}_{\theta,\theta'},
&
\mathsf r(\theta,\theta')=\mathrm{same}.
\end{cases}
\]
Given \(X=(\theta,u)\), draw \(\theta'\) from \(q_R(\cdot\mid X)\)
and \(u'\) from the action selected by \(\mathsf r(\theta,\theta')\). For
\(X'=(\theta',u')\), accept with probability \(1\wedge R_{\rm MH}\),
where
\[
R_{\rm MH}
=
\frac{p(\theta')\widehat L(\theta',u')q_R(\theta\mid X')}
{p(\theta)\widehat L(\theta,u)q_R(\theta'\mid X)}.
\]
The complete transition specification and particle-filter implementation
are given in the Supplementary Material.

\subsection{Validity of the extended-state transition}
\label{subsec:route-validity-local}

Let the auxiliary space be standard Borel. All probability kernels below
are measurable in their conditioning variables. In particular,
$C^s_{\theta,\theta'}(u,du')$ is jointly measurable in
$(\theta,\theta',u)$. The two route actions preserve the same extended
target if their auxiliary laws satisfy a common swap identity. We use the standard
product-measure formulation of Metropolis--Hastings reversibility
\citep{tierney1998}. Let
\[
\bar\pi(d\theta,du)
=
Z^{-1}
p(\theta)
\widehat L(\theta,u)
M_\theta(du)\,d\theta,
\]
where
\[
\widehat L(\theta,u)\ge0,
\qquad
\int
\widehat L(\theta,u)M_\theta(du)
=
L(\theta),
\qquad
0<Z<\infty.
\]
For \(s\in\{\mathrm{same},\mathrm{cross}\}\), suppose
\begin{equation}
\label{eq:route-aux-swap}
M_\theta(du)
C^s_{\theta,\theta'}(u,du')
=
M_{\theta'}(du')
C^s_{\theta',\theta}(u',du)
\end{equation}
as measures on the paired auxiliary space.

\begin{proposition}[Validity of route-conditioned auxiliary allocation]
\label{prop:route-allocation-validity}
Assume that the pair route is symmetric and
\eqref{eq:route-aux-swap} holds for each route. For every parameter pair
relevant to the proposal, suppose that, under the corresponding
auxiliary joint law, almost every \((u,u')\) with positive
extended-target densities satisfies
\[
q_R(\theta'\mid X)>0
\quad\Longleftrightarrow\quad
q_R(\theta\mid X')>0.
\]
Proposal mass outside the target support is assigned to the
self-transition. Then the
Metropolis--Hastings transition is reversible with respect to
\(\bar\pi\). Its parameter marginal is the posterior
\[
\pi(d\theta\mid y)
=
Z^{-1}
p(\theta)L(\theta)\,d\theta.
\]
\end{proposition}

\noindent\textit{Proof sketch.}\quad
Pair symmetry selects the same auxiliary action in both directions.
The swap identity cancels the auxiliary proposal measure from the
forward--reverse product-measure ratio. The remaining ratio is the
Metropolis--Hastings ratio above. The full product-measure proof,
including the coupled-particle disintegration argument, is given in
the Supplementary Material.

Independent refresh satisfies \eqref{eq:route-aux-swap} through the
symmetry of \(M_\theta\otimes M_{\theta'}\). For the inherited action,
the network coupling is defined through a swap-symmetric joint law of
two particle-filter outputs. Split reaction
propagation preserves the two correct Markov-jump marginals
\citep{anderson2018split}, and the coupled-resampling matrix preserves
the two particle-weight marginals and transposes when the two sides are
exchanged \citep{jacob2016coupling}. Disintegration of the resulting
joint law gives the conditional inherited kernel required by
\eqref{eq:route-aux-swap}. The Supplementary Material states the numerical
conventions used to implement this mathematical kernel.

\subsection{How within-event updates affect later risk}
\label{subsec:allocation-risk}

Write $P$ for selective allocation and $Q$ for coupling on both routes,
with the same parameter proposal. Their accepted cross-event
subkernels agree. Thus $D=P-Q$ is supported within the current event
class and $Dh=0$. The next event value has the same distribution.
Within-event updates change the full state from which later crossings
are attempted.

\begin{proposition}[The first stationary effect of within-event allocation]
\label{prop:allocation-risk}
Suppose $P$ and $Q$ are reversible for $\Pi$ and have the same
cross-event subkernel. Let $g=Ph=Qh$ and start both chains from $\Pi$.
Their event MSEs agree for $B=1,2,3$. At $B=4$,
\begin{equation}
 \mathcal R_4(P)-\mathcal R_4(Q)
 =-\frac1{16}\int\Pi(dx)D(x,dz)\{g(z)-g(x)\}^2.
 \label{eq:delayed-effect}
\end{equation}
\end{proposition}

The function $g(x)$ is the probability of being in the event after one
transition from the full state $x$. Equation~\eqref{eq:delayed-effect}
shows how redistribution within an event class affects later estimation.
The first effect can occur at the second transition under a
nonstationary initial law. Supplementary Section~S4 gives the proof and
the corresponding finite-horizon continuation formula.

A constant mixture also changes the cross-event update. Let $I$ denote
the fully independent kernel and $M_\alpha=(1-\alpha)I+\alpha Q$.
Decompose $P-M_\alpha=\alpha D_s+(1-\alpha)D_c$, where the
off-diagonal parts are $I-Q$ within an event class and $Q-I$ between
classes, respectively; each signed kernel has row mass zero.
For stationary runs of any fixed length, the exact risk difference is
\begin{equation}
 \mathcal R_B(P)-\mathcal R_B(M_\alpha)
 =-\frac{\alpha T_s+(1-\alpha)T_c}{B^2}.
 \label{eq:mixture-route-effect}
\end{equation}
The two terms integrate changes in future event probabilities along
their respective routes. Supplementary Section~S4 defines them and
proves the identity without assuming the kernels commute. The comparison
with a constant mixture therefore accounts for both routes.

At a computation budget, the continuation state also records elapsed
cost, the completed event sum and its count. The supplementary budget
recursion compares the joint next-state and cost laws. It retains the
cost of preparation and the rule for a transition that crosses the
budget. We use a fully specified price model next; Section~\ref{sec:targeted-repair}
uses measured computation times.

\subsection{Repair within the informative refinement model}
\label{subsec:certified-repair}

Return to \eqref{eq:informative-density} and the retained count in
\eqref{eq:count-target}. The independent action draws a fresh binomial
count. For the coupled action, let $J_{a,b}$ be the joint law of two
counts formed from the same $N$ uniform draws, with success probabilities
$a$ and $b$. Given current count $k$, draw the new count from the
conditional joint law
\[
 (1-\rho)\{b_a\otimes b_b\}+\rho J_{a,b},\qquad 0\leq\rho\leq1.
\]
Both components have the binomial marginals and transpose when the two
parameters are exchanged. They satisfy the swap identity.
Use $q_R$ centered at the known $m(\theta)=\log L_\theta(y)$.
Here $q_0$ assigns probability $1/2$ to each parameter. In the fine
channel of the result below, $NL_\theta(y)<1$, so the residual tilt is
inactive and the parameter proposal remains $q_0$.

An independent action costs one unit and a coupled action costs two.
The budget is twenty units after preparation of the initial state.
An action is selected first. If its price exceeds the remaining budget,
the procedure stops and reports the completed event average. An empty
average returns the initial event value. These are declared action
prices. They specify the joint state--cost comparison independently of
machine timings.

For each channel, allow a constant coupling probability
$\alpha\in[0,1]$. Let $J_{Y,\alpha}$ be its total risk and define
\[
 B_Y^{\rm common}=\inf_{\alpha\in[0,1]}\E_Y j_{Y,\alpha}(Y),
 \qquad
 B_Y^{\rm oracle}=\E_Y\inf_{\alpha\in[0,1]}j_{Y,\alpha}(Y).
\]
The second comparison allows a different constant for every observed
outcome. For either benchmark, write
\begin{equation}
 \Delta=B_F-B_C,\qquad
 \Gamma=B_F-J_{F,P},\qquad
 R=J_{F,P}-B_C=\Delta-\Gamma.
 \label{eq:repair-bridge}
\end{equation}
Thus $\Delta>0$ establishes inversion, and $\Gamma>\Delta$ establishes
repair below the optimized coarse benchmark.

Consider three full initial laws at each observation: the extended
posterior $\nu_{0,y}=\Pi_y$; equal prior parameter probabilities with
count fixed at one, $\nu_{1,y}(\theta,k)=\mathbf1\{k=1\}/2$; and
equal prior parameter probabilities with a positive conditional count,
\[
 \nu_{2,y}(\theta,k)=
 \frac{b_{\theta,y}(k)}{2\{1-(1-L_\theta(y))^N\}},\qquad k\geq1.
\]
The last law first draws the parameter, then conditions its count to be
positive. Let $\nu_{\lambda,y}=\sum_{i=0}^2\lambda_i\nu_{i,y}$,
where $\lambda_i\geq0$, $\sum_i\lambda_i=1$, and the same weights
are used at every observation in both channels.

\begin{theorem}[Certified inversion and repair over an initialization class]
\label{thm:certified-repair}
Take $N=4$, $K=8$, $\varepsilon=4/5$, $\rho=3/4$ and the price budget
defined above. For every $\lambda$ in the closed initialization simplex,
the common and outcome-specific mixture benchmarks coincide. Their
minimizers are $\alpha=0$ in the coarse channel and $\alpha=1$ in
the fine channel. Selective allocation satisfies
\begin{equation}
 \Delta\geq0.00779460,\qquad
 \Gamma\geq0.00961868,\qquad
 R\leq-0.00059812.
 \label{eq:certificate-margins}
\end{equation}
\end{theorem}

The theorem uses a finite calculation with rigorous rational interval
bounds. At each observation and initial-law vertex, the mixture risk is
a Bernstein polynomial of degree at most 21. Its coefficient bounds
certify the two endpoint minimizers over the entire interval $[0,1]$.
Risk is affine in the initial law for each fixed policy. The shared
minimizers then extend the vertex bounds to the whole simplex.
Supplementary Section~S3 gives the transition probabilities, recursion,
coefficient inequalities and proof.

This comparison puts both questions on the same latent model. Refinement
increases total risk even after optimizing a constant mixture separately
in each channel. Event-based allocation on the fine channel recovers
enough computational accuracy to cross the optimized coarse benchmark.
The result specifies the likelihood estimator, price budget and initial
laws under which this occurs.

Figure~\ref{fig:certificate} shows the three initial-law vertices.
The fine mixture stays above the optimized coarse risk; selective
allocation falls below it at each vertex.

\begin{figure}[htbp]
\centering\includegraphics[width=\textwidth]{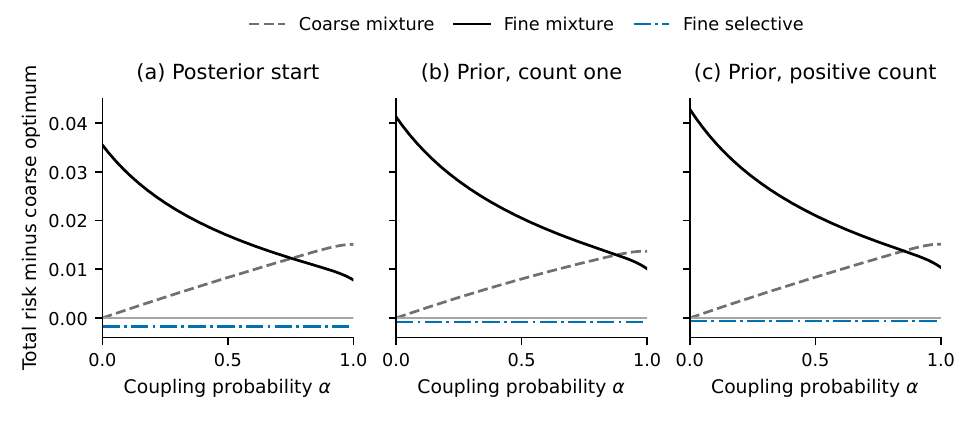}
\caption{Inversion and repair in the informative count model.
Each panel subtracts its optimized coarse total risk. Curves evaluate
the saved Bernstein risk polynomials at the three initial laws in
Theorem~\ref{thm:certified-repair}, with $N=4$, $K=8$, $\rho=3/4$
and budget 20 after initialization. The horizontal blue line is the
fine-channel selective risk. The rational coefficient certificate
establishes the inequalities over the full mixture interval.}
\label{fig:certificate}
\end{figure}

The same count model separates the state and cost effects. With the
untilted parameter proposal, selective and full coupling have stationary
four-step risk difference $-\rho(1-\rho)^2\mathcal C$, where
$\mathcal C>0$ for $N\geq2$ and strictly positive binomial probabilities.
At $\rho=1$ their fixed-step event paths have the same law, while their
action prices differ. The supplementary proof gives both statements.
They explain why finite-step mixing and budget performance require
separate comparisons.

\section{Finite-budget comparisons}
\label{sec:targeted-repair}

The experiments examine two consequences of the risk decomposition.
The network comparison measures the effect of auxiliary allocation at
a fixed observation. The paired transcription comparison measures
posterior uncertainty and computational error under coarse and fine
observations. Each comparison specifies its initial law and charged clock.

\subsection{Network model and comparison methods}
\label{subsec:gene-route-instantiation}

We use a feed-forward gene network with seven components and fourteen
reactions. A single trajectory provides allele-specific counts
$(M_A,M_{B1},M_{B2})$ at times $1,\ldots,20$, with capture probability
one. The unknown activation rates are $k_{1,Bj}=0.4e^{z_j}$,
$j=1,2$. The prior is uniform on $[-2,2]^2$; the generating rates
are $(0.2,0.8)$. The event
\[
 h(z)=\mathbf1\{|z_2-z_1|\geq\log2\}
\]
records at least a two-fold rate contrast. The guided particle filter
uses $N=256$. A fixed quadratic surface from 512 separate likelihood
evaluations centers the residual in $q_R$. The supplement specifies
the reaction hazards, proposal densities and auxiliary coupling.

Table~\ref{tab:network-methods} separates the parameter proposal from
the auxiliary action. A five-method pilot uses 96 retained states and
two paired repetitions per state. A separate tuning comparison evaluates
$\alpha\in\{0,0.25,0.5,0.75,1\}$ on 48 stratified states.
It selects $\alpha=0.25$ at the prespecified 25-second CPU checkpoint.
This mixture and the selective rule are then frozen for confirmation.
All mixture probabilities use the same $q_R$.

\begin{table}[H]
\centering\spacingset{1}
\caption{Parameter proposals and auxiliary actions in the network comparison.}
\label{tab:network-methods}
\begin{tabular}{llll}
\hline
Method & Proposal & Same event & Cross event\\
\hline
Baseline & $q_0$ & Independent & Independent\\
Tilt with refresh & $q_R$ & Independent & Independent\\
Tilt with full coupling & $q_R$ & Coupled & Coupled\\
Constant mixture & $q_R$ & Mixture $\alpha$ & Mixture $\alpha$\\
Selective allocation & $q_R$ & Independent & Coupled\\
\hline
\end{tabular}
\par\smallskip\begin{minipage}{\textwidth}\small
A mixture selects the coupled action with probability $\alpha$ before
each transition, independently of the proposed route. The five-method
pilot uses $\alpha=0.5$; confirmation uses the separately tuned
$\alpha=0.25$. Both directions of every proposal enter the acceptance ratio.
\end{minipage}
\end{table}

\subsection{Confirmation under new initialization draws}
\label{subsec:cost-aware-confirmation}

The confirmation generates 64 initialization draws in each of two
panels. The first samples with replacement from a new, likelihood-weighted
pool of 2,048 candidate extended states. The second draws directly from
a separate frozen initialization proposal and uses the first positive
likelihood estimate. Each draw supplies two parameter representations and two
paired repetitions for both methods, giving 1,024 trajectories. The analysis averages
trajectory squared errors within each initialization draw. Its sampling
unit is that draw, with 64 clusters per panel. The pool analysis conditions
on the generated finite pool.

An independent importance reference gives
$\widehat\Psi=0.6524107$ with Monte Carlo standard error $0.0091054$.
The estimator averages event values from transitions completed within
6.25, 12.5 or 25 process-CPU seconds after initialization. It uses the
initial event value if the prefix is empty. The transition crossing a
checkpoint is excluded from the average. Figure~\ref{fig:network-confirmation}
shows the resulting risk curves.

\begin{figure}[htbp]
\centering
\includegraphics[width=\textwidth]{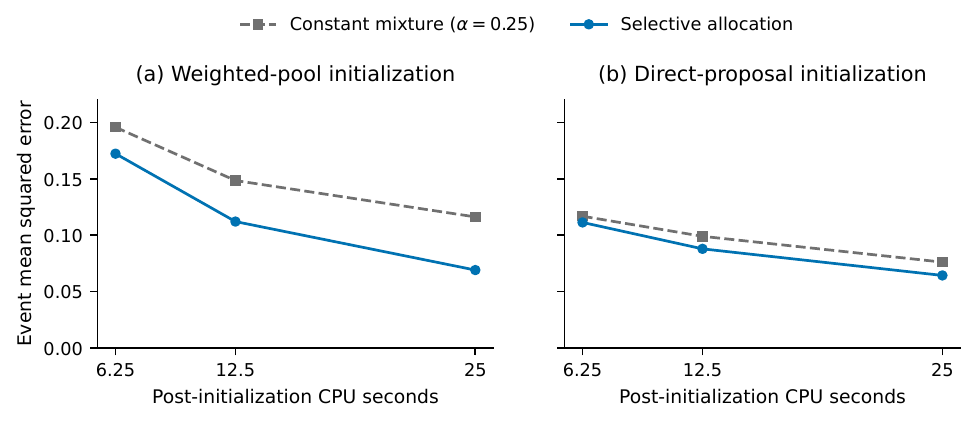}
\caption{Network confirmation at three post-initialization CPU budgets.
Each point averages trajectory squared errors over 64 initialization
draws, two parameter representations and two repetitions. Lines connect
the recorded checkpoints. The 25-second primary difference intervals,
reported in the text, account for initialization-cluster and reference
uncertainty. Preparation costs are reported separately in the supplement.}
\label{fig:network-confirmation}
\end{figure}

At 25 seconds, selective allocation lowers weighted-pool MSE from
$0.11631$ to $0.06915$, a relative reduction of $40.5\%$. The paired
difference is $-0.04715$, with approximate joint-coverage interval
$[-0.07925,-0.01482]$. For direct-proposal initialization, the MSEs are
$0.07619$ and $0.06439$; the difference is $-0.01180$, with interval
$[-0.02718,0.00538]$. These intervals use a reference-endpoint envelope
and initialization-cluster Student intervals, as specified in the
supplement. The weighted-pool interval supports a gain under that initialization
law. The direct-proposal interval spans zero.

The selective trajectories complete means of $20.27$ and $19.59$
transitions in the two panels, compared with $21.28$ and $21.67$ for
the mixture. At a fixed ten transitions, the risk differences are also
negative. At one transition they are positive. The fixed-step comparison
in the supplement isolates this dependence on the trajectory horizon.
Pool generation and initialization require $1{,}985.48$ and $134.47$
CPU seconds across this confirmation. Calibration and method development
have separate recorded costs. The reported risk gain concerns the
post-initialization budget.

\subsection{Variation across observation records}

Three additional simulated records from the same network examine the
scope of the frozen comparison. Each has eight paired initialization
draws per panel. Particle counts are $512$, $1024$ and $256$ after
the corresponding pool and reference checks. Figure~\ref{fig:cross-dataset}
retains the records and initializations separately.

\begin{figure}[H]
\centering
\includegraphics[width=\textwidth]{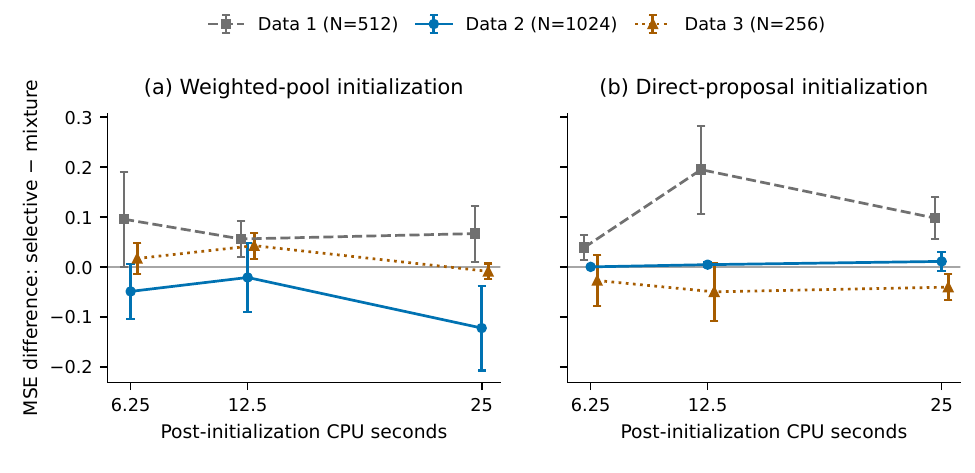}
\caption{Selective-minus-mixture MSE on three additional network records.
Each point uses eight paired starts. Bars show one start Monte Carlo
standard error with the reference held fixed. Small horizontal offsets separate the records at
each checkpoint. The supplement gives reference sensitivity separately.
Particle counts differ across records.}
\label{fig:cross-dataset}
\end{figure}

At 25 seconds, weighted-pool differences are $0.0667$, $-0.1225$ and
$-0.0090$; direct-proposal differences are $0.0977$, $0.0109$ and
$-0.0406$. The changes across records and starts make the initial law
and observation part of the empirical comparison. Theorem~\ref{thm:certified-repair}
gives uniform conclusions within its stated count model and initialization
class. The network results describe these particular records, particle
counts and frozen rules.

\subsection{Paired observation and particle-effort comparison}
\label{subsec:controlled-model}
\label{subsec:reference-pseudomarginal}
\label{subsec:paired-posterior-risk}

The transcription model has two alleles, each with promoter activation,
deactivation, transcript synthesis and degradation. One latent trajectory
and shared capture draws generate 20 allele-specific observations with
capture probability $0.6$. Their sums give the coarse record. The event
is $h(\theta)=\mathbf1\{\theta_1<\theta_2\}$, where the coordinates
are the log activation and synthesis multipliers for allele two.
Finite-state likelihood quadrature gives posterior probabilities
$p_C=0.4877789$ and $p_F=0.7911916$, with variances $0.2498506$ and
$0.1652074$. An alternative prior-CDF quadrature changes the probabilities
by $1.61\times10^{-5}$ and $1.40\times10^{-6}$. The supplement gives
the likelihood truncation checks and reference sensitivity.

Both channels use the same fixed Gaussian proposal and four equally
weighted initial parameter values. The bootstrap filter resamples
systematically when effective sample size falls below $N/2$.
At 120 CPU seconds, the estimate averages completed post-transition
event values, including rejections. This budget includes initialization.
The particle-count confirmation has 310 paired groups, each containing
both channels at $N=600$ and $N=1200$. An earlier variance study fixes their
allocation across starts at $(40,110,70,90)$; the analysis weights
remain $1/4$. Table~\ref{tab:paired-risk} reports the contrasts.

\begin{table}[H]
\centering\spacingset{1}
\caption{Risk contrasts for the paired transcription record at 120 CPU seconds.}
\label{tab:paired-risk}
\begin{tabular}{lrr}
\hline
Contrast & Estimate & 95\% interval\\
\hline
Total risk: fine minus coarse, $N=600$ & $-0.1041$ & $[-0.1171,-0.0911]$\\
Total risk: fine minus coarse, $N=1200$ & $-0.0925$ & $[-0.1099,-0.0751]$\\
Computational MSE: $1200-600$, coarse & $0.0356$ & $[0.0235,0.0478]$\\
Computational MSE: $1200-600$, fine & $0.0472$ & $[0.0324,0.0621]$\\
Particle-count interaction & $0.0116$ & $[-0.0054,0.0286]$\\
\hline
\end{tabular}
\par\smallskip\begin{minipage}{\textwidth}\small Intervals are approximate pointwise stratified Monte Carlo intervals at the fixed observations and four equally weighted starts. The particle-count interaction is the fine-channel MSE increase minus the coarse-channel increase. CPU budgets include initialization. Numerical references are held fixed.\end{minipage}
\end{table}

The fine observation has lower conditional total risk at both particle
counts. Increasing $N$ raises computational MSE from $0.07981$ to
$0.11545$ in the coarse channel and from $0.06036$ to $0.10761$ in
the fine channel. Mean completed transitions fall from $24.44$ to
$11.36$ and from $23.69$ to $11.08$. Fresh-estimator relative-variance
point estimates fall in all eight start--channel comparisons. The larger particle calculation leaves fewer transitions within the budget. The interaction
interval in Table~\ref{tab:paired-risk} includes zero.

A separate seven-cell allocation experiment pairs proposals and auxiliary
actions on this transcription record. Its 112-trajectory extension gives
120-second selective-$q_R$-minus-baseline MSE differences of $-0.0148$
and $-0.0005$ in the coarse and fine channels, with intervals
$[-0.0758,0.0463]$ and $[-0.0784,0.0773]$. The supplement reports all
seven component contrasts and shorter-budget results. Together, the
two reaction models distinguish the observation-risk comparison from
the allocation comparison at a fixed posterior.

\section{Discussion}
\label{sec:discussion}

Finite-run risk connects observation design with Monte Carlo computation.
The information gain from a finer observation must be compared with the
change in computational error. Accepted event flow and coarse-kernel
contraction give a sufficient condition for inversion. The bootstrap
construction realizes this condition with informative refinement at every
fixed particle count. Its likelihood estimator and work count are
explicit, so the statistical and computational contributions can be
evaluated separately.

At a fixed observation, posterior uncertainty is common to the competing
algorithms. Auxiliary allocation therefore acts directly on their
total-risk difference. The continuation identities show how a same-event
update changes the states from which later crossings occur. The
same-model certificate makes this connection quantitative: selective
allocation recovers enough accuracy in the fine channel to pass the
optimized coarse benchmark at the declared price budget. The comparison
holds over the full interval of constant mixture probabilities and the
specified initialization class.

The reaction-network studies examine this criterion with measured CPU
costs. The gene-network confirmation gives a gain over the tuned mixture
under finite-pool initialization; the direct-proposal and additional-record
results describe its dependence on the initial law and data. In the
paired transcription study, the finer record has lower conditional total
risk. Increasing the particle count lowers fresh likelihood-variance
point estimates. It also leaves fewer completed transitions and increases
computational MSE. These comparisons identify the observation,
initialization and charged budget as parts of the computational task.

The identities also suggest extensions. A data-dependent allocation rule
could use estimates of continuation loss together with its calibration
cost. Quantitative guarantees for general sequential filters would need
bounds on their retained-state laws and joint state--cost transitions.
For functionals beyond events, the allocation must account for the size
of each functional change.

\section*{Disclosure statement}
No potential conflict of interest was reported.

\section*{Funding}
This research received no external funding.

\begin{samepage}
\section*{Data availability statement}
The simulated observations, saved numerical results, and code are
deposited in Figshare \citep{xu2026materials}, with reserved DOI
\url{https://doi.org/10.6084/m9.figshare.33985732}.
The materials will be made public upon publication of the journal article.
\par
\end{samepage}

\section*{Supplementary material}
The Supplementary Material contains proofs, model and algorithm specifications,
computational costs, and additional comparisons. The figure-reproduction
materials contain frozen inputs, plotting code, and numerical checks.

\spacingset{1}
\bibliographystyle{agsm}
\bibliography{references_main}
\end{document}


\date{}
\bigskip\bigskip\bigskip
\begin{center}
{\LARGE\bf Supplementary Material for ``Information--Computation Inversion in Pseudo-Marginal MCMC''}
\end{center}
\medskip
\spacingset{1}

\setcounter{section}{0}
\renewcommand{\thesection}{S\arabic{section}}
\renewcommand{\thesubsection}{S\arabic{section}.\arabic{subsection}}
\setcounter{table}{0}
\renewcommand{\thetable}{S\arabic{table}}
\setcounter{figure}{0}
\renewcommand{\thefigure}{S\arabic{figure}}
\setcounter{equation}{0}
\renewcommand{\theequation}{S\arabic{equation}}
\setcounter{algorithm}{0}
\renewcommand{\thealgorithm}{S\arabic{algorithm}}
\section{Risk identities and observation refinement}\label{supp:observation-theory}\label{supp:risk-conditioning}
\subsection{A common risk for statistical and computational error}
Let $\Theta$ have prior $\nu$, let $F$ be an observation, and let $C=T(F)$ be its coarsening. Fix $h$ with $\E h(\Theta)^2<\infty$. For $Y\in\{F,C\}$, define
\[
 p_Y=\E\{h(\Theta)\mid Y\},\qquad
 v_Y=\Var\{h(\Theta)\mid Y\}.
\]
An algorithm returns $A_Y$, using $Y$ and randomization that is conditionally independent of $\Theta$ given $Y$. This requirement includes the initialization and the stopping rule. Assume $\E A_Y^2<\infty$. Write
\[
 V_Y=\E v_Y,\quad E_Y=\E(A_Y-p_Y)^2,\quad
 J_Y=\E\{A_Y-h(\Theta)\}^2.
\]
Here $E_Y$ includes both bias and variance of the computational estimate.

\begin{lemma}[Risk decomposition]\label{lem:risk}
The following identities hold:
\begin{align}
 J_Y&=V_Y+E_Y,\label{eq:risk}\\
 J_F-J_C&=E_F-E_C-G,\qquad
 G=\E(p_F-p_C)^2=V_C-V_F\geq0.\label{eq:contrast}
\end{align}
Conditional on a fixed observation $y$, the corresponding risk is
$j_Y(y)=v_Y(y)+e_Y(y)$, where $e_Y(y)=\E\{(A_Y-p_Y(y))^2\mid Y=y\}$.
\end{lemma}
\begin{proof}
Expand $A_Y-h(\Theta)=(A_Y-p_Y)+(p_Y-h(\Theta))$. Conditional on $(Y,A_Y)$, the second term has mean zero by conditional independence. The cross term therefore vanishes. Next $p_C=\E(p_F\mid C)$ by the tower property. Conditional variance decomposition, followed by expectation, gives $V_C=V_F+\E\Var(p_F\mid C)=V_F+G$. The conditional identity follows from the same expansion at fixed $y$.
\end{proof}

Equation~\eqref{eq:contrast} compares procedures on a common inferential task. Fine-data procedures include every procedure that first computes $C=T(F)$. The results below fix the procedure used for each observation and compare its finite-run risk. They retain the statistical gain $G$ explicitly.

\subsection{Evaluation under a common observation}
Let $H=h(\Theta)$ be square integrable, let $F$ be a fine observation,
and let $C=T(F)$. Put $p_Y=\E(H\mid Y)$ and $v_Y=\Var(H\mid Y)$.
A coarse procedure returns $A_C=a(C,U_C)$, where $U_C$ is independent
of $(\Theta,F)$; a fine procedure returns $A_F=b(F,U_F)$ under the same
independence requirement. The algorithmic randomizations may be coupled
with each other. Assume both outputs are square integrable. Tuning constants
are part of the specified procedures; data-dependent tuning must be included
in their information sets.

\begin{proposition}[Three risk comparisons]
Define $e_Y(y)=\E[(A_Y-p_Y(y))^2\mid Y=y]$ and
$b_C(c)=\E[A_C\mid C=c]-p_C(c)$. Each procedure's own conditional risk is
\[
 j_Y(y)=v_Y(y)+e_Y(y).
\]
For a common fine observation $F=f$, with $c=T(f)$ and
$\delta=p_C(c)-p_F(f)$, the risks of both procedures under this conditioning are
\begin{align}
 r_F(f)&=v_F(f)+e_F(f),\\
 r_C(f)&=v_F(f)+e_C(c)+\delta^2+2\delta b_C(c).
\end{align}
Consequently,
\[
 r_F(f)-r_C(f)=e_F(f)-e_C(c)-\delta^2-2\delta b_C(c).
\]
Under prior predictive averaging,
\[
 \E\{j_F(F)-j_C(C)\}
 =\E\{r_F(F)-r_C(F)\}
 =\E e_F(F)-\E e_C(C)-\E(p_F-p_C)^2.
\]
\end{proposition}
\begin{proof}
Conditional on an observation and the procedure's randomization, the
conditional expectation of the posterior residual is zero. Expanding the
squared error gives $j_Y=v_Y+e_Y$ and
$r_Y(f)=v_F(f)+\E[(A_Y-p_F(f))^2\mid F=f]$.
The law of $A_C$ conditional on $F=f$ is its law at input $c$.
Expand $A_C-p_F=(A_C-p_C)+\delta$ to obtain the formula for $r_C$.
The tower property gives $\E(p_F\mid C)=p_C$, so
$\E[\delta b_C(C)]=0$. The same property gives
$\E v_C-\E v_F=\E(p_F-p_C)^2$. Substitution proves the averaged identities.
\end{proof}

At a fixed paired observation, $j_F(f)-j_C(c)$ and $r_F(f)-r_C(f)$
evaluate different conditional risk comparisons. Their difference is
\[
 [r_F-r_C]-[j_F-j_C]=v_C-v_F-\delta^2-2\delta b_C.
\]
This correction is available from saved estimates of the posterior mean.
For an event functional, the empirical common-observation risk contrast is
\[
 \overline{(A_F-p_F)^2-(A_C-p_F)^2}.
\]
Evaluation uses the fine posterior reference; the coarse procedure
continues to receive its coarse input. It is evaluated separately from the prespecified own-observation contrast.

Within a fixed observation channel, two algorithms share $p_Y$ and $v_Y$.
The difference of their own conditional risks is exactly their computational
MSE difference, for arbitrary prescribed initializations and stopping rules
that satisfy the conditional independence requirement.

\subsection{Observation refinement under a common particle proposal}
Fix a parameter $\theta$ and a coarse observation $c$. Let $\mu_\theta$ be the latent-state law. For each fine outcome $f$ with $T(f)=c$, let $g_f(z)\geq0$ be its observation probability, and put $g_c=\sum_{f:T(f)=c}g_f$. The index set is finite or countable. Write
\[
 L_f=\int g_f\,d\mu_\theta,\quad L_c=\sum_f L_f,
 \qquad \omega_f=L_f/L_c.
\]
Assume $0<L_c<\infty$ and omit indices with $L_f=0$. Use a single proposal probability $\xi$ for $c$ and every compatible $f$, with $g_c\,\mu_\theta\ll\xi$. Define
\[
 r_f=\frac{d(g_f\mu_\theta)}{L_f\,d\xi},\qquad
 r_c=\frac{d(g_c\mu_\theta)}{L_c\,d\xi}.
\]
Use the same iid draws $Z_1,\ldots,Z_N\sim\xi$. The normalized likelihood estimates are
\[
 W_f=\frac1N\sum_i r_f(Z_i),\qquad W_c=\frac1N\sum_i r_c(Z_i).
\]

\begin{proposition}[Coarsening identity]\label{prop:weights}
We have $r_c=\sum_f\omega_f r_f$, $W_c=\sum_f\omega_f W_f$, and
\begin{equation}\label{eq:convex}
 W_c\cx W_I,\qquad \Prb(I=f)=\omega_f,
\end{equation}
where $I$ is independent of the particle draws. If $\sum_f\omega_f\int r_f^2\,d\xi<\infty$, then
\begin{equation}\label{eq:variance}
 \sum_f\omega_f\Var(W_f)-\Var(W_c)
 =\frac1N\sum_f\omega_f\int(r_f-r_c)^2\,d\xi.
\end{equation}
The remainder vanishes exactly when $r_f=r_c$ $\xi$-almost everywhere for every $f$ of positive $\omega_f$.
\end{proposition}
\begin{proof}
The Radon--Nikodym identity follows from $g_c=\sum_f g_f$ and Tonelli's theorem. Thus $W_c=\E(W_I\mid Z_1,\ldots,Z_N)$, which gives~\eqref{eq:convex} by conditional Jensen's inequality. Since $\int r_f\,d\xi=\int r_c\,d\xi=1$, independence of the draws gives $N\Var(W_f)=\int r_f^2\,d\xi-1$. Pointwise expansion gives
\[
 \sum_f\omega_f(r_f-r_c)^2
 =\sum_f\omega_f r_f^2-r_c^2.
\]
Integration proves~\eqref{eq:variance}. Every summand of the remainder is nonnegative, which gives the equality condition.
\end{proof}

The densities $r_f$ describe the latent posterior relative to the particle proposal. The remainder in~\eqref{eq:variance} measures separation among these latent posteriors. The quantity $G$ measures changes in the posterior mean of $h(\Theta)$. The identity uses classical conditional-expectation and importance-sampling arguments; the role of the $\chi^2$ divergence in importance sampling is developed by \citet{sanzalonso2021}. The convex order in~\eqref{eq:convex} compares the coarse estimator with a mixture over fine outcomes. The fixed-target ordering results of \citet{andrieu2015convex} require a common posterior target and a pointwise order of weight laws at each parameter.

\subsection{Connection to sequential particle likelihoods}
For independent capture of allele counts $m_1,m_2$ with a common capture probability $\rho$, write
\[
 g_F(y_1,y_2\mid m_1,m_2)
 =\binom{m_1}{y_1}\binom{m_2}{y_2}
 \rho^{y_1+y_2}(1-\rho)^{m_1+m_2-y_1-y_2}.
\]
For total count $c=y_1+y_2$, Vandermonde's identity gives
\[
 \sum_{y_1+y_2=c}g_F(y_1,y_2\mid m_1,m_2)
 =\binom{m_1+m_2}{c}\rho^c(1-\rho)^{m_1+m_2-c}=g_C(c\mid m_1,m_2).
\]
At multiple observation times, conditional independence multiplies the observation probabilities along a latent path. Summation over compatible fine sequences and integration over the common latent-path law commute by Tonelli's theorem. Thus the full likelihoods obey the same coarsening relation. Proposition~\ref{prop:weights} applies when the latent-path proposal is shared, including the one-observation bootstrap case.

\begin{proposition}[Likelihood unbiasedness with adaptive resampling]\label{prop:PF}
Consider a state-space model with initial law $\mu$, Markov transitions $M_t$, and nonnegative bounded observation potentials $g_t$. A bootstrap filter propagates particles with $M_t$. It either retains normalized weights or resamples with conditionally unbiased offspring counts. The resampling decision is measurable with respect to the current particle system. The product of its weighted mean observation increments is an unbiased estimate of the likelihood.
\end{proposition}
\begin{proof}
Let $Z_t^N$ be the product of increments through time $t$, and let $\eta_t^N$ be the weighted particle measure before optional resampling. Define $\Gamma_t^N(\varphi)=Z_t^N\eta_t^N(\varphi)$. At time zero, iid initial particles give $\E\Gamma_0^N(\varphi)=\mu(\varphi)$, with $Z_0^N=1$. Let $\bar\eta_t^N$ denote the weighted measure after the optional resampling. On a resampling step it has equal weights; the conditional offspring property gives
\[
 \E\{\bar\eta_t^N(\varphi)\mid\mathcal F_t\}=\eta_t^N(\varphi).
\]
The same identity holds when no resampling is selected. Here $\mathcal F_t$ contains the current weighted system and the resampling decision, before the fresh resampling randomness.

After propagation and multiplication by $g_{t+1}$, normalization cancels the new likelihood increment, giving
\[
 \E\{\Gamma_{t+1}^N(\varphi)\mid\mathcal F_t\}
 =Z_t^N\eta_t^N\{M_{t+1}(g_{t+1}\varphi)\}
 =\Gamma_t^N\{M_{t+1}(g_{t+1}\varphi)\}.
\]
If an increment is zero, define the unnormalized measure to remain zero thereafter; the same recursion holds. Induction proves that $\E\Gamma_t^N$ equals the model's unnormalized filtering measure. Taking $\varphi=1$ proves the claim. Bounded nonnegative potentials ensure all likelihood expectations exist.
\end{proof}

For systematic resampling, let $U$ be uniform on $[0,1/N)$ and place the $N$ points $U+j/N$, $j=0,\ldots,N-1$, on the cumulative-weight intervals. The expected number in an interval of length $w_i$ is $Nw_i$, by integrating the point indicators over $U$. This verifies the offspring condition of Proposition~\ref{prop:PF}. It gives the likelihood property needed for marginal parameter PMMH. A selected-path extended-target construction may require further conditions on labelled ancestor indices; these are explicit in \citet[Assumption 2]{andrieu2010pmcmc}.

\paragraph{Retained-weight control.}
For a posterior density $\pi(\theta)>0$, let $W$ be the normalized fresh likelihood estimator, with law $Q_\theta$ and mean one. Consider an independent fresh-weight proposal with parameter density $q(\theta,\theta')$. At a retained state $(\theta,w)$, the usual pseudo-marginal ratio gives
\[
 \Prb(\mathrm{accept}\mid\theta,w)\leq
 \min\left\{1,\frac{H(\theta)}w\right\},\qquad
 H(\theta)=\frac{\int\pi(\theta')q(\theta',\theta)\,d\theta'}{\pi(\theta)}.
\]
Indeed, bounding $\min(1,r)$ by $r$ cancels the forward proposal on its support, and integration over the proposed weight uses $\int w'Q_{\theta'}(dw')=1$; extending the remaining nonnegative integral over all reverse proposals gives the bound. If $q$ is a fixed $d$-dimensional Gaussian random walk with covariance $\Sigma\succ0$, then
\[
 H(\theta)\leq\frac{(2\pi)^{-d/2}|\Sigma|^{-1/2}}{\pi(\theta)}<\infty
\]
at every point with finite positive posterior density. Conditional on such a retained state, the no-acceptance event gives the finite-horizon bound
\[
 \E_{\theta,w}(A_B-\pi h)^2
 \geq\{h(\theta)-\pi h\}^2(1-H(\theta)/w)_+^B.
\]
The expectation starts before the first of the $B$ transitions. This bound links retained weight to the chosen posterior functional. It applies to fresh auxiliary proposals; a conditional auxiliary kernel requires its own transition ratio and acceptance calculation.

\paragraph{The common inferential target.}
Observation refinement changes both terms in~\eqref{eq:risk}. After the observation is fixed, algorithms targeting the same posterior share $v_Y(y)$. Their total-risk difference equals their computational-MSE difference. This gives a common criterion for studying observation choice and the allocation of auxiliary computation. An allocation method can be evaluated by the event-functional error it removes at a specified cost.

\paragraph{Sequential and time-budget comparisons.}
Adaptive particle filters can generate different particle populations and resampling times under the two channels. Proposition~\ref{prop:PF} establishes likelihood unbiasedness for each channel; Proposition~\ref{prop:weights} specifies the shared-proposal condition for its weight comparison. Equation~\eqref{eq:risk} also covers prescribed nonstationary starts and time-based stopping, subject to its conditional-independence and integrability assumptions. Such comparisons estimate $e_Y(y)$ for the actual algorithm and budget.

\section{Event flow and informative refinement}
\label{supp:informative-refinement}

This section proves the two refinement theorems in the main paper.
The event is $H=h(\Theta)\in\{0,1\}$. Conditional on the observation
$Y=y$, the algorithm and its initial law use randomization independent
of $\Theta$. Write $p_y=\E(H\mid Y=y)$ and
$v_y=p_y(1-p_y)$. For any square-integrable reported value $A_y$,
conditioning on the algorithm's randomization gives
\[
 \E\{(A_y-H)^2\mid y\}=v_y+\E\{(A_y-p_y)^2\mid y\}.
\]
The cross term vanishes by conditional independence. For $C=T(F)$,
the tower property gives
\[
 \E(v_F\mid C=c)=v_c-\E\{(p_F-p_c)^2\mid C=c\}.
\]
The initial law and any calibration are part of the algorithm. A coarse
algorithm in this comparison has access to $c$ and its own randomization.

\subsection{General event-flow bounds}

Fix $y$ and start $K_y$ from its invariant law $\Pi_y$. Put
$g=h-p_y$ and $\gamma_k=\langle g,K_y^kg\rangle_{\Pi_y}$.
Invariance implies equal one-step flow from event 1 to event 0 and
from event 0 to event 1. Thus the probability of a one-step event
change is $2\mathfrak q_y$. Binary event values give
\[
 \gamma_k=v_y-\tfrac12\Pr\{h(X_k)\ne h(X_0)\}.
\]
The event on the right requires a change during at least one of the
$k$ intervening transitions. The union bound and stationarity yield
$\gamma_k\geq v_y-k\mathfrak q_y$. Therefore
\begin{align*}
 e_{B,y}
 &=B^{-2}\left\{Bv_y+2\sum_{k=1}^{B-1}(B-k)\gamma_k\right\}\\
 &\geq v_y-\frac{2\mathfrak q_y}{B^2}
                  \sum_{k=1}^{B-1}k(B-k)
 =v_y-\frac{B^2-1}{3B}\mathfrak q_y.
\end{align*}
Taking the maximum with zero proves the claimed bound. All quantities
are finite because $h$ is bounded. Reversibility is not required.

For the coarse kernel, the centered $L^2$ contraction assumption gives
$|\gamma_k|\leq v_ca_c^k$ by Cauchy--Schwarz and iteration. This gives
$e_{B,c}\leq v_cD_B(a_c)$. Under a minorization
$K_c\geq\epsilon_c\Pi_c$ with $0<\epsilon_c<1$, the residual kernel
$R_c=(K_c-\epsilon_c\Pi_c)/(1-\epsilon_c)$ is Markov and invariant
for $\Pi_c$. Jensen's inequality makes it an $L^2$ contraction.
On centered functions $K_c=(1-\epsilon_c)R_c$, proving the asserted
value of $a_c$. If $\epsilon_c=1$, $K_c=\Pi_c$ and $a_c=0$.

Under $\mathfrak q_f\leq rv_f$, the fine-channel MSE is at least
$v_f(1-c_Br)_+$. Average this inequality conditionally on $C=c$,
add the posterior-variance term, and subtract the coarse upper bound.
This gives the first theorem of the main paper. More generally, inversion
holds whenever
\[
 \E\{v_F+(v_F-c_B\mathfrak q_F)_+\mid C=c\}
 >v_c\{1+D_B(a_c)\}.
\]
This version permits different flow ratios across compatible fine outcomes.

\subsection{The informative cell model}

Let $\Theta$ be uniform on $\{-1,1\}$ and let
$f_\theta(z)=(1+\theta\varepsilon z)/2$ on $[-1,1]$.
Integrating this density over the $j$th interval gives
$L_\theta(j)=(1+\theta\varepsilon x_j)/(2K)$.
The predictive mass is $m_j=1/(2K)$ and the posterior event probability
is $p_j=(1+\varepsilon x_j)/2$.
For the sign observation, $L_\theta(c)=(1+\theta\varepsilon c/2)/2$
and $p_c=(1+\varepsilon c/2)/2$, where $c\in\{-1,1\}$.
In particular, the coarse observation is informative.

The cell centers satisfy
$(2K)^{-1}\sum_j x_j^2=1/3-1/(12K^2)$. Substitution into
$v_j=(1-\varepsilon^2x_j^2)/4$ proves the three variance formulas in
the main paper. If a deterministic compression merges two distinct
cells, their posterior probabilities differ and both have positive mass.
The conditional variance of $p_F$ within that merged outcome is positive.
The total-variance identity therefore proves strict event-information loss.

\subsection{Flow and spectral bounds for the independent count chain}

Fix an observed outcome and write $b_\theta(k)$ for its binomial count
probability. The fresh joint proposal is
$\eta(\theta,k)=b_\theta(k)/2$, for $0\leq k\leq N$.
Let $\ell(\theta,k)=k/N$ and $m=(L_{-1}+L_1)/2$.
The extended target is $\Pi=\ell\eta/m$ on positive counts.
From positive count $k$, the independent MH acceptance probability of
count $l$ is $1\wedge(l/k)$, including zero acceptance at $l=0$.
Multiplying by the extended-target mass and summing over cross-event
transitions gives
\[
 \mathfrak q
 =\sum_{k,l}\frac{b_{-1}(k)b_1(l)\min(k,l)}{4Nm}
 =\frac{\E\min(J_{-1},J_1)}{4Nm}.
\]
For nonnegative integers, $\min(k,l)\leq kl$.
The two fresh counts are independent, so
$\mathfrak q\leq NL_{-1}L_1/(4m)=Nmv$.

For completeness, positivity of the independent MH operator follows
from the integral representation of the minimum. With $r(x)$ the
rejection probability, for any real $f$ on the finite retained space,
\[
 \langle f,Pf\rangle_\Pi
 =\int r(x)f(x)^2\Pi(dx)
 +\frac1m\int_0^\infty
    \left\{\int f(x)\mathbf1\{\ell(x)>t\}\eta(dx)\right\}^2dt
 \geq0.
\]
This is the independent-proposal case of
\citet[Proposition 16]{andrieu2015convergence}. Reversibility and
positivity place its spectrum in $[0,1]$. For $g=h-p$,
$\langle g,(I-P)g\rangle_\Pi=\mathfrak q$.
Its normalized spectral measure consequently has mean
$1-\mathfrak q/v$. Since $D_B$ is increasing and convex on $[0,1]$,
the spectral representation and Jensen's inequality give
\[
 e_B=v\int D_B(a)\,d\sigma_g(a)
       \geq vD_B(1-\mathfrak q/v).
\]
In the coarse channel, $\ell\leq1$ and $m=1/2$.
The accepted kernel is at least $\ell(z)\eta(dz)=\Pi(dz)/2$.
Hence $e_{B,c}\leq v_cD_B(1/2)$ by the preceding minorization argument.

For $\varepsilon=4/5$, $K=2N$ implies $Nm_j=1/4$, and thus
$e_{10,j}\geq v_jD_{10}(3/4)$.
Also $V_F=59/300+1/(75K^2)>59/300$ and $V_C=21/100$.
Subtracting the coarse total-risk upper bound from the fine lower bound
proves the second main-paper theorem. Direct rational arithmetic gives
\[
 \frac{59}{300}\{1+D_{10}(3/4)\}
 -\frac{21}{100}\{1+D_{10}(1/2)\}
 =\frac{99024089}{3932160000}.
\]
For the first theorem's nonempty example, $N=4$, $K=8$ and $B=4$
give $r=1/4$, $a_c=1/2$, $\gamma_c=1/16$,
$L_B=11/16$ and $d_c=33/64$. The conditional bound is
$357/25600>0$ in either sign class.

\subsection{Initial-law stability}

Define $\operatorname{TV}(\mu,\nu)=\sup_A|\mu(A)-\nu(A)|$.
For a fixed algorithm, including its stopping and fallback rules,
let $\ell_y(x)=\E_x(A_y-p_y)^2$. Event averages and their declared
fallbacks lie in $[0,1]$, so
$0\leq\ell_y(x)\leq M_y:=\max\{p_y^2,(1-p_y)^2\}$.
It follows that
\[
 |e_y(\mu)-e_y(\nu)|\leq M_y\operatorname{TV}(\mu,\nu).
\]
After predictive averaging, the bound is
$\eta_Y=\E_Y[M_Y\operatorname{TV}(\mu_Y,\nu_Y)]$.
It is uniform over any fixed family of algorithms with these output
bounds, and hence also bounds the change in its infimum risk.
The refinement difference changes by at most $\eta_F+\eta_C$.
This statement concerns the complete extended initial law.

In the model above,
$\E M_C=49/100$ and
$\E M_F=151/300-1/(75K^2)\leq151/300$.
Uniform initial-law distance at most $1/40$ therefore changes the
refinement difference by at most $149/6000$.
Subtracting this from the stationary bound gives
$458483/1310720000>0.00034979$, as claimed.

For a general application, the same argument subtracts
$\eta_F+\eta_C$ from an established stationary refinement margin.
The flow bound, coarse contraction bound and initial-law distances
are separate conditions to be verified for that application.

\section{A certificate for inversion and repair in the same model}
\label{supp:refinement-repair}

We use the cell model of Section~\ref{supp:informative-refinement}.
All quantities in this section are conditional on the observed outcome
until predictive averaging is stated. Let $b_a(k)=\binom Nk a^k(1-a)^{N-k}$,
where $a=L_\theta(y)$, and let $p_y=L_1(y)/(L_{-1}(y)+L_1(y))$.
The retained space has $2N$ states $(\theta,k)$, $k=1,\ldots,N$.

\subsection{Complete count kernels and their costs}

For $a\leq b$, the shared-uniform joint count law is
\[
 J_{a,b}(k,l)=
 \begin{cases}
 \displaystyle\binom Nk\binom{N-k}{l-k}
       a^k(b-a)^{l-k}(1-b)^{N-l},&0\leq k\leq l\leq N,\\
 0,&\text{otherwise}.
 \end{cases}
\]
For $a>b$ set $J_{a,b}(k,l)=J_{b,a}(l,k)$. This is the multinomial
law of uniforms in $[0,a]$, $(a,b]$ and $(b,1]$.
Its marginals are $b_a,b_b$, and exchanging $a,b$ transposes it.
The conditional coupled count kernel is
\[
 C^\rho_{a,b}(k,l)=(1-\rho)b_b(l)+\rho J_{a,b}(k,l)/b_a(k).
\]
Every denominator is positive in the retained space since $0<a<1$.

Set $r_-(\theta,k)=\min\{\log(k/(NL_\theta)),0\}$.
The residual parameter proposal has cross-event probability
$q_R(-\theta\mid\theta,k)=1/(1+\exp r_-(\theta,k))$.
The complementary probability proposes the same parameter.
For a proposed positive count $l$, acceptance is
\[
 a\{(\theta,k),(\theta',l)\}
 =1\wedge\frac{l\,q_R(\theta\mid\theta',l)}
                    {k\,q_R(\theta'\mid\theta,k)}.
\]
Zero count proposals are rejected. For independent, full and selective
allocation, multiply the parameter proposal, its selected count
conditional law and this acceptance probability. Add all rejected mass
to the current state. This completely specifies the matrices $I,Q,P$.
The joint count identity and symmetric event routing prove reversibility
for $\Pi_y(\theta,k)=k b_{L_\theta}(k)/(2Nm_y)$.

For each policy $a$, let $K_{a,1}$ and $K_{a,2}$ be its subprobability
matrices for the next retained state and selected price 1 or 2.
They include rejection and sum to the transition matrix.
For $I$, only $K_{I,1}$ is nonzero; for $Q$, only $K_{Q,2}$ is
nonzero. For $P$, the proposed parameter route determines the price,
including when the proposal is rejected. For the constant mixture,
$K_{\alpha,1}=(1-\alpha)I$ and $K_{\alpha,2}=\alpha Q$.
The mixture coin is drawn before the update and independently of its
outcome.

\subsection{Finite recursion and interval certificate}

Let $t$ be the remaining integer budget, $s$ the accumulated event sum,
$n$ the completed count and $h_0$ the initial event value. Define
\[
 \ell(s,n,h_0)=
 \begin{cases}(s/n-p_y)^2,&n>0,\\(h_0-p_y)^2,&n=0.\end{cases}
\]
For a fixed policy, its conditional computational risk $V_a$ satisfies
\begin{align*}
 V_a(x,t,s,n,h_0)
 =\sum_{c\in\{1,2\}}\sum_z K_{a,c}(x,z)
 \begin{cases}
 V_a(z,t-c,s+h(z),n+1,h_0),&c\leq t,\\
 \ell(s,n,h_0),&c>t.
 \end{cases}
\end{align*}
At $t=0$ this expression is the terminal loss. Each nonterminal action
reduces the remaining budget by at least one, so the recursion is finite.
The total conditional risk under initial law $\nu$ is
$v_y+\sum_x\nu(x)V_a(x,20,0,0,h(x))$.
This also proves that the recursion keeps the dependence of costs on
states and rejections.

For a constant mixture, a complete path contains at most 20 affordable
actions and one unaffordable action label. Its probability contains a
factor $\alpha^r(1-\alpha)^{m-r}$ with $m\leq21$.
Summing the terminal losses gives a polynomial of degree at most 21.
Degree elevation puts it in the form
\[
 j_{y,\alpha}(\nu_i)=\sum_{r=0}^{21}b_{i,y,r}
     \binom{21}{r}\alpha^r(1-\alpha)^{21-r}.
\]
All model parameters, kernel probabilities, initial masses and terminal
losses are rational in the certified setting. Indeed,
$\exp r_-=\min\{k/(NL_\theta),1\}$, so the logarithm in the
proposal definition does not require numerical evaluation here.
The certificate propagates rational enclosures with endpoints on the
$10^{-40}$ grid, rounding lower bounds down and upper bounds up after
each operation. Kernel nonnegativity, unit row mass, invariance and
reflection are checked using exact rational arithmetic before risk
propagation. Elementary interval inclusion then proves by induction
that each reported interval encloses its exact coefficient or risk.

For the coarse observations, the saved intervals satisfy
$\underline b_{i,y,r}>\overline b_{i,y,0}$ for every $r>0$.
For fine observations they satisfy
$\underline b_{i,y,r}>\overline b_{i,y,21}$ for every $r<21$.
The Bernstein basis is nonnegative and sums to one. Hence the unique
minimizers are 0 and 1, respectively, over the entire closed interval.
The minimum certified coefficient gaps over the observations, for the
three initial laws in order, exceed $0.00083976$, $0.00084338$ and
$0.00084785$.
Reflection $(\theta,j,k)\mapsto(-\theta,2K-1-j,k)$ complements the
event and its posterior mean while preserving squared loss, proposal
probabilities and costs. The same reflection preserves each initial-law
family. Thus one representative from each reflected observation pair
suffices; all outcomes enter the uniform predictive average.

\subsection{Extension to the complete initialization simplex}

For each fixed policy, conditioning on its initial state gives
\[
 j_{y,a}\left(\sum_i\lambda_i\nu_{i,y}\right)
       =\sum_i\lambda_i j_{y,a}(\nu_{i,y}).
\]
The algorithm does not read the index of the selected initial-law
component. The same weights apply at all observations and in both
channels. At each vertex, for every $\alpha$,
$j_{c,\alpha}(\nu_i)\geq j_{c,0}(\nu_i)$ and
$j_{f,\alpha}(\nu_i)\geq j_{f,1}(\nu_i)$.
Multiplication by $\lambda_i$ and summation preserve these inequalities.
The minimizer is common to all observations within each channel, so
both the common and outcome-specific optimized risks use that same
endpoint throughout the simplex.

Consequently $\Delta,\Gamma,R$ are affine in $\lambda$.
Table~\ref{tab:certificate-vertices} gives the vertex values used in
the simplex bound.
\begin{table}[H]
\centering\spacingset{1}
\caption{Certified inversion and repair at the three initial-law vertices.}
\label{tab:certificate-vertices}
\begin{tabular}{lrrr}
\hline
Initial law & $\Delta$ & $\Gamma$ & $R$\\
\hline
$\nu_0$ & $0.007794603972$ & $0.009618680322$ & $-0.001824076349$\\
$\nu_1$ & $0.010089809178$ & $0.010921227696$ & $-0.000831418518$\\
$\nu_2$ & $0.010372172001$ & $0.010970294684$ & $-0.000598122682$\\
\hline
\end{tabular}
\par\smallskip\begin{minipage}{\textwidth}\small Each displayed center has enclosure radius $2\times10^{-12}$. The exact rational endpoints are supplied in the certificate records. These values use the optimized coarse and fine constant-mixture benchmarks.\end{minipage}
\end{table}

The exact rational endpoints and all 27 representative observation and
initial-law coefficient records are supplied with the certificate.
Taking the smallest lower endpoint for $\Delta$ and $\Gamma$, and the
largest upper endpoint for $R$, proves the conservative main-paper bounds.
The optimization step matters here: an infimum of affine functions is
concave, and the difference of two such infima need not be affine.
The common-minimizer inequalities are what permit this extension.

The stability bound in Section~\ref{supp:informative-refinement} also
applies to these optimized risks. Uniform full-state TV distance at most
$\tau$ from a common-simplex initial law changes $\Delta$ and $R$ by
at most $(49/100+161/320)\tau=0.993125\tau$.
Thus $\tau\leq0.0006$ preserves both signs. The primary initialization
guarantee is the simplex itself; this additional radius controls arbitrary
perturbations away from it.

\subsection{State effects at a fixed number of transitions}

Use the uniform parameter proposal, as occurs when the residual tilt
is inactive. Denote the independent kernel by $K_I$, the full coupled
kernel by $Q_\rho$, and selective allocation by $P_\rho$.
Since the acceptance ratio is $l/k$, independent of $\rho$,
\[
 Q_\rho=(1-\rho)K_I+\rho Q_1,\qquad P_\rho-Q_\rho=\rho A.
\]
The off-diagonal entries of $A$ are zero between parameter classes and
\[
 A((\theta,k),(\theta,l))=\tfrac12 b_{L_\theta}(l)\min(1,l/k),
 \qquad k\ne l.
\]
Its diagonal entries make the row sums zero. Thus $Ah=0$ and
$\Pi(x)A(x,z)=\Pi(z)A(z,x)$.
Let $f=h-p_y$. Shared uniforms make the accepted cross-parameter
probability under $Q_1$ equal to
$\tfrac12\min(1,L_{-\theta}/L_\theta)$, independently of $k$.
It follows that $AQ_1f=0$.
The stationary four-step identity in the next section now yields
\[
 e_4(P_\rho)-e_4(Q_\rho)
 =\frac{\rho(1-\rho)^2}{8}\langle K_If,A K_If\rangle_\Pi
 =-\rho(1-\rho)^2\mathcal C,
\]
where
\[
 \mathcal C=\frac1{32}\sum_{\theta,k,l}
 \Pi(\theta,k)b_{L_\theta}(l)\min(1,l/k)
 \{(K_If)(\theta,l)-(K_If)(\theta,k)\}^2.
\]
For $N\geq2$ and both likelihoods in $(0,1)$ this is strictly positive.
To see this, the independent crossing probability is
$a_\theta(k)=\tfrac12\sum_l b_{L_{-\theta}}(l)\min(1,l/k)$.
Then $a_\theta(1)-a_\theta(2)=b_{L_{-\theta}}(1)/4>0$.
The $(k,l)=(1,2)$ term has positive weight and nonzero squared difference.
For $N=1$ no within-parameter count change is available and
$\mathcal C=0$. The formula compares selective with full coupling at
four stationary transitions; the budget certificate makes the separate
comparison with the optimized mixture class.

\subsection{The event and cost process at perfect count coupling}

Under the same uniform parameter proposal, set $\rho=1$ and let
$a_\theta=\min(1,L_{-\theta}/L_\theta)$.
If the proposed likelihood exceeds the current likelihood, the shared
count cannot decrease and acceptance is one. Otherwise its conditional
law is $\operatorname{Bin}(k,L_{-\theta}/L_\theta)$, and the mean
acceptance probability is $L_{-\theta}/L_\theta$.
Selective allocation therefore has the following next-parameter and
price probabilities, independently of the retained count:
\[
 \Pr(\theta,1)=1/2,\qquad
 \Pr(-\theta,2)=a_\theta/2,\qquad
 \Pr(\theta,2)=(1-a_\theta)/2.
\]
Full coupling has the same parameter transition probabilities and price
two on every attempt. Conditional iteration proves equality of the
fixed-step event path laws for any common initial parameter marginal.
For either policy, its budget output depends only on that marginal,
because the joint parameter and cost process no longer depends on the
initial count. Its risk is obtained from the preceding budget recursion.
This gives a cost effect even when the fixed-step event laws coincide.

\section{Finite-run effects of event allocation}
This section proves the event-allocation risk identity in the main text
and gives its two-step and stationary consequences. The argument uses
finite operator telescoping and conditional second moments; these tools
also underlie earlier pseudo-marginal comparisons
\citep{andrieu2009pseudo,andrieu2015convex}.

\subsection{Finite-horizon effects of within-event updates}
\label{sec:allocation-finite-horizon-draft}

Let $P$ and $Q$ be Markov kernels on the full extended state space $E$,
assumed to be standard Borel.
Let $h:E\to\{0,1\}$, $E_j=\{x:h(x)=j\}$, and $\Psi\in[0,1]$.
Assume that, for every $x\in E_j$ and every measurable $A\subset E_{1-j}$,
\begin{equation}
 P(x,A)=Q(x,A).\label{eq:allocation-common-cross-draft}
\end{equation}
Write $D=P-Q$ and $f=h-\Psi$. For a chain with kernel $K$, define
\[
 R_{B,K}(x)=\mathbb E_{x,K}\left[
 \left\{B^{-1}\sum_{t=1}^B h(X_t)-\Psi\right\}^2\right].
\]
The initial value $h(X_0)$ is excluded from the average. In the application,
$\Psi=\bar\pi(h)$ and $\bar\pi$ is the common extended target.

For $n\geq0$, let
\[
 A_{n,K}(x)=\mathbb E_{x,K}\sum_{j=1}^n f(X_j),\qquad
 C_{n,K}(x)=\mathbb E_{x,K}\left\{\sum_{j=1}^n f(X_j)\right\}^2.
\]
Both functions are zero at $n=0$. Products of functions below are pointwise.
Conditioning on the first transition gives
\begin{align}
 A_{n,K}&=K(f+A_{n-1,K}),\\
 C_{n,K}&=K\{f^2+2fA_{n-1,K}+C_{n-1,K}\}.
\label{eq:allocation-moments-draft}
\end{align}

\begin{proposition}[Finite-horizon allocation identity]
\label{prop:allocation-finite-horizon-draft}
Assume \eqref{eq:allocation-common-cross-draft}. Let
$S_t=\sum_{j=1}^t f(X_j)$, with $S_0=0$. For each integer $B\geq1$,
\begin{align}
 R_{B,P}(x)-R_{B,Q}(x)
 =\frac1{B^2}\sum_{t=0}^{B-1}\mathbb E_{x,P}\bigl[
 &2\{S_t+f(X_t)\}D A_{B-t-1,Q}(X_t)\nonumber\\
 &+D C_{B-t-1,Q}(X_t)\bigr].
\label{eq:allocation-finite-horizon-draft}
\end{align}
In particular, the last term is zero. If $g=Ph=Qh$, then
\begin{equation}
 R_{2,P}(x)-R_{2,Q}(x)
 =\{1/4-\Psi+h(x)/2\}Dg(x).
\label{eq:allocation-two-step-draft}
\end{equation}
\end{proposition}

\begin{proof}
Condition \eqref{eq:allocation-common-cross-draft} implies that
$D(x,\cdot)$ is supported on $E_{h(x)}$, has total mass zero, and
annihilates $f$ and $f^2$. Define the continuation loss
\[
 V_t(x,s)=B^{-2}\{s^2+2sA_{B-t,Q}(x)+C_{B-t,Q}(x)\},
 \qquad 0\leq t\leq B.
\]
Thus $V_B(x,s)=s^2/B^2$ and
$V_t(x,s)=\int Q(x,dy)V_{t+1}(y,s+f(y))$.
Telescoping along a $P$ path yields
\[
 R_{B,P}(x)-R_{B,Q}(x)
 =\sum_{t=0}^{B-1}\mathbb E_{x,P}
 \int D(X_t,dy)V_{t+1}(y,S_t+f(y)).
\]
On the support of this signed measure, $f(y)=f(X_t)$. The squared
constant term integrates to zero. Expanding the remaining two terms gives
\eqref{eq:allocation-finite-horizon-draft}. All functions are bounded for
fixed $B$, so these conditional expectations and finite sums are valid.
For $B=2$, only $t=0$ remains. Use
$A_{1,Q}=g-\Psi$ and $C_{1,Q}=(1-2\Psi)g+\Psi^2$ to obtain
\eqref{eq:allocation-two-step-draft}.
\end{proof}

The identity places the effect of an auxiliary update in two continuation
quantities: the mean of the remaining centered event sum and its second
moment. Updates within an event class can change both quantities while
preserving the next event value in distribution.

For a bounded measurable $v$, define
\[
 \omega_j(v)=\sup_{y\in E_j}v(y)-\inf_{y\in E_j}v(y),\qquad
 \tau(x)=\tfrac12\lvert P(x,\cdot)-Q(x,\cdot)\rvert(E).
\]
Here $|\cdot|$ is the variation measure of a finite signed measure.

\begin{proposition}[Control by variation within an event class]
\label{prop:allocation-oscillation-draft}
Under the assumptions of Proposition~\ref{prop:allocation-finite-horizon-draft},
\begin{align}
 |R_{B,P}(x)-R_{B,Q}(x)|
 \leq\frac1{B^2}\sum_{t=0}^{B-1}\mathbb E_{x,P}\Bigl[
 \tau(X_t)\bigl\{
 &2|S_t+f(X_t)|\,\omega_{h(X_t)}(A_{B-t-1,Q})\nonumber\\
 &+\omega_{h(X_t)}(C_{B-t-1,Q})\bigr\}\Bigr].
\label{eq:allocation-oscillation-draft}
\end{align}
If $\kappa(x)$ is the common probability of leaving $E_{h(x)}$, then
$\tau(x)\leq1-\kappa(x)$.
\end{proposition}

\begin{proof}
The positive and negative parts of $D(x,\cdot)$ each have mass $\tau(x)$
and are supported on $E_{h(x)}$. Subtracting any constant from $v$ leaves
$Dv(x)$ unchanged. The Jordan decomposition therefore gives
$|Dv(x)|\leq\tau(x)\omega_{h(x)}(v)$.
Apply this inequality to the two terms in
\eqref{eq:allocation-finite-horizon-draft}, followed by the triangle inequality.
The restrictions of $P(x,\cdot)$ and $Q(x,\cdot)$ to $E_{h(x)}$ each have
mass $1-\kappa(x)$, which bounds $\tau(x)$.
\end{proof}

If the common crossing probability is constant on each event class,
the event sequence has the same two-state Markov law under both kernels.
Then $A_{n,Q}$ and $C_{n,Q}$ are constant on each event class and the
finite-horizon risks agree for every $B$. In applications, the continuation
quantities can vary within a class. Equation~\eqref{eq:allocation-oscillation-draft}
identifies the variation that must be controlled to obtain a quantitative
risk guarantee.

\subsection{The crossing profile after a within-event update}
Let $P$ and $Q$ be Markov kernels on a standard Borel space $E$.
Fix a measurable $h:E\to\{0,1\}$ and write $E_j=\{x:h(x)=j\}$.
Assume that
\[
 P(x,A)=Q(x,A),\qquad x\in E_j,\quad A\subseteq E_{1-j}
\]
for every measurable $A$. Define $D=P-Q$ and
\[
 \kappa(x)=P(x,E_{1-h(x)})=Q(x,E_{1-h(x)}),\qquad g=Ph=Qh.
\]
For $\Psi\in[0,1]$, set
\[
 R_{B,K}(x)=\mathbb E_{x,K}\left[
 \left\{\frac1B\sum_{t=1}^B h(X_t)-\Psi\right\}^2\right].
\]
The average starts after the first transition.

\begin{proposition}[Two-step crossing profile]
\label{prop:crossing-profile-two-step}
For every $x\in E$,
\begin{align}
 R_{2,P}(x)-R_{2,Q}(x)
 &=\{1/4-\Psi+h(x)/2\}Dg(x)\\
 &=\{1/4-\Psi+h(x)/2\}\{1-2h(x)\}D\kappa(x).
 \label{eq:crossing-profile-two-step}
\end{align}
Also, $Dg=P^2h-Q^2h$.
\end{proposition}
\begin{proof}
The signed measure $D(x,\cdot)$ has mass zero and is supported on
$E_{h(x)}$. Thus $Dh=0$, $Dh^2=0$, and
$D(hg)(x)=h(x)Dg(x)$. Expanding the square and conditioning on the
first transition gives
\[
 4\{R_{2,P}(x)-R_{2,Q}(x)\}
 =\{1-4\Psi+2h(x)\}Dg(x).
\]
Since $g=h+(1-2h)\kappa$, restriction to $E_{h(x)}$ yields
$Dg(x)=(1-2h(x))D\kappa(x)$. Finally,
$P^2h-Q^2h=Pg-Qg=Dg$.
\end{proof}

The quantity $D\kappa(x)$ measures how the two updates change the
expected crossing probability at the resulting full state. For
$1/4<\Psi<3/4$, the coefficient of $D\kappa(x)$ in
\eqref{eq:crossing-profile-two-step} is negative in both classes.
In this range, a larger expected next crossing probability lowers
two-step risk. Outside this range the coefficient in one class changes
sign. The target event probability therefore enters the allocation
criterion directly.

\begin{proposition}[A paired-path diagnostic]
\label{prop:paired-two-step-diagnostic}
Suppose two chains start at the same full state $x$. Couple their first
transitions so that $H_1^P=H_1^Q$, where $H_t^K=h(X_t^K)$, and so that
their full states agree whenever the first transition crosses the event
boundary. Couple their second event values to agree whenever their first
full states agree. Then, almost surely,
\[
 \left\{\frac{H_1^P+H_2^P}{2}-\Psi\right\}^2
 -\left\{\frac{H_1^Q+H_2^Q}{2}-\Psi\right\}^2
 =\{1/4-\Psi+h(x)/2\}(H_2^P-H_2^Q).
\]
\end{proposition}
\begin{proof}
When $H_2^P=H_2^Q$, both sides vanish. Otherwise, the first transition
could not have crossed, so $H_1^P=H_1^Q=h(x)$. Expanding the two
squares and using $(H_2^K)^2=H_2^K$ proves the identity.
\end{proof}

\subsection{Stationary risk and within-class variation of the crossing profile}
Assume additionally that $P$ and $Q$ are reversible with respect to the
same probability measure $\bar\pi$, and let $\Psi=\bar\pi(h)$.
Define $R_{B,K}(\bar\pi)=\int R_{B,K}(x)\bar\pi(dx)$ and
\[
 \mathcal E_K(v)=\frac12\int\bar\pi(dx)K(x,dy)\{v(y)-v(x)\}^2.
\]

\begin{proposition}[The first stationary risk difference]
\label{prop:stationary-four-step-profile}
Under these assumptions,
\[
 R_{B,P}(\bar\pi)=R_{B,Q}(\bar\pi),\qquad B=1,2,3,
\]
and
\[
 R_{4,P}(\bar\pi)-R_{4,Q}(\bar\pi)
 =\frac18\langle g,Dg\rangle_{\bar\pi}
 =\frac18\{\mathcal E_Q(g)-\mathcal E_P(g)\}.
\]
\end{proposition}
\begin{proof}
Write $f=h-\Psi$. Stationarity gives
\[
 B^2 R_{B,K}(\bar\pi)
 =B\langle f,f\rangle_{\bar\pi}
   +2\sum_{j=1}^{B-1}(B-j)\langle f,K^jf\rangle_{\bar\pi}.
\]
Since $Pf=Qf=g-\Psi$, the lag-one terms agree. Reversibility implies
$\langle f,K^2f\rangle=\langle Kf,Kf\rangle$, so the lag-two terms
also agree. At lag three,
\[
 \langle f,P^3f\rangle-\langle f,Q^3f\rangle
 =\langle g-\Psi,D(g-\Psi)\rangle
 =\langle g,Dg\rangle.
\]
The last equality uses $D1=0$ and $\bar\pi D=0$. Its coefficient at
$B=4$ is $2/16$. The Dirichlet-form expression follows from
$\mathcal E_K(g)=\langle g,(I-K)g\rangle$.
\end{proof}

Cross-event flow contributes equally to both Dirichlet forms.
Their difference concerns variation of $g$ within an event class.
This is distinct from the Dirichlet form of $h$, which is already
fixed by the common crossing subkernel. The result applies to a
stationary initial distribution; a fixed starting-state bank defines
a different risk.

\subsection{Comparison with a constant mixture at every fixed horizon}
\label{supp:mixture-route-proof}

Let $P,Q$ be reversible for a common probability measure $\Pi$, set
$D=P-Q$ and $f=h-\Pi h$. Both chains start from $\Pi$.
For bounded functions $u,v$, row mass zero, invariance and the symmetry
of the signed measure $\Pi(dx)D(x,dz)$ imply
\[
 \langle u,Dv\rangle_\Pi
 =-\tfrac12\int\Pi(dx)D(x,dz)
                  \{u(z)-u(x)\}\{v(z)-v(x)\}.
\]
All integrals exist because the variation of $D(x,\cdot)$ is at most two.
The finite telescoping identity is
\[
 P^k-Q^k=\sum_{j=0}^{k-1}P^jD Q^{k-1-j}.
\]
No commutation assumption is used. Self-adjointness of $P$ then gives
$\langle f,P^jDQ^lf\rangle=\langle P^jf,DQ^lf\rangle$.
Substitute the bilinear identity into each stationary covariance term.
With
\begin{align*}
 \mathcal F_B^{P,Q}(x,z)=
 \sum_{\substack{j,l\geq0\\j+l\leq B-2}}(B-1-j-l)
 \{P^jh(z)-P^jh(x)\}\{Q^lh(z)-Q^lh(x)\},
\end{align*}
the resulting identity is
\[
 \mathcal R_B(P)-\mathcal R_B(Q)
 =-B^{-2}\int\Pi(dx)D(x,dz)\mathcal F_B^{P,Q}(x,z).
\]
The sum is empty for $B=1$. The constant centering cancels in the
differences of functions.

Now let $I,Q$ be fully independent and fully coupled updates with the
same parameter proposal, and let $P$ use $I$ within an event and $Q$
across events. For $M_\alpha=(1-\alpha)I+\alpha Q$, define $D_s,D_c$
by their off-diagonal restrictions as in the main paper and complete
each diagonal to give row mass zero. The symmetric routing preserves
the symmetry of $\Pi D_s$ and $\Pi D_c$, and
$P-M_\alpha=\alpha D_s+(1-\alpha)D_c$.
Set
\[
 T_s=\int\Pi(dx)D_s(x,dz)\mathcal F_B^{P,M_\alpha}(x,z),\qquad
 T_c=\int\Pi(dx)D_c(x,dz)\mathcal F_B^{P,M_\alpha}(x,z).
\]
The preceding formula proves the main-paper mixture identity.
Both terms depend on $\alpha$ through the continuation kernel
$M_\alpha$. Optimizing the mixture requires this dependence as well
as the displayed route weights.

On a same-event edge, $h(z)-h(x)=0$. Thus terms with $j=0$ or $l=0$
vanish in $T_s$, and the first possible direct same-event term occurs
at $B=4$. A cross-event term can occur at $B=2$.
For the particular comparison of selective and full coupling, $Ph=Qh$.
The stationary proof above then gives equal risks for $B\leq3$ and
\[
 \mathcal R_4(P)-\mathcal R_4(Q)
 =\tfrac18\langle Ph,(P-Q)Ph\rangle_\Pi
 =-\tfrac1{16}\int\Pi(dx)(P-Q)(x,dz)
                         \{Ph(z)-Ph(x)\}^2.
\]
This is the main-paper delayed-effect proposition. Its nonstationary
two-step counterpart was proved earlier in this section.

The identities here concern a fixed number of transitions. At a random
budget, the joint-law recursion in the following subsection applies to
selective and mixture policies as well. Its local difference must
include both routes and the conditional cost distribution.

\subsection{Event allocation at a computational budget}\label{supp:joint-state-cost}

Fix an event $h:E\to\{0,1\}$ on a standard Borel extended state
space $E$. Let $\Psi$ denote its posterior probability. The state includes
the retained likelihood estimate and its auxiliary record.
State-dependent computing times affect the completed prefix
\citep{murray2021anytime}. Selective
allocation and full coupling use the same parameter proposal. Their
accepted transitions across the event boundary also agree. We first
identify an additional relation between their state kernels, then compare
the event averages completed within a fixed computational budget.

\subsubsection{The refresh component in the actual state kernels}

Let $M_{\theta'}$ be the fresh auxiliary law and
$C_{\theta,\theta'}$ an auxiliary conditional kernel satisfying
\[
 M_\theta(du)C_{\theta,\theta'}(u,du')
 =M_{\theta'}(du')C_{\theta',\theta}(u',du).
\]
Fix $0<\epsilon<1$ and write
$C_\epsilon=\epsilon M+(1-\epsilon)C$.
The parameter proposal $q(\theta'\mid x)$ may depend on the current
retained state. It is fixed across the following three constructions.
Each uses the same acceptance function
\[
 \alpha(x,y)=1\wedge
 \frac{p(\theta')\widehat L(\theta',u')q(\theta\mid y)}
 {p(\theta)\widehat L(\theta,u)q(\theta'\mid x)},
 \qquad x=(\theta,u),\quad y=(\theta',u'),
\]
where $p$ is the prior density. A zero numerator causes rejection.
All constructions use $C_\epsilon$ on cross-event proposals. On
same-event proposals, $P$ uses $M$, $Q$ uses $C_\epsilon$, and $S$ uses $C$.
Their rejection probabilities complete the three Markov kernels.

\begin{proposition}[Mixture relation for event allocation]
\label{prop:bridge-mixture}
Suppose the model evidence is finite and strictly positive,
the fresh likelihood estimate is nonnegative and unbiased,
and the displayed acceptance rule is well defined on positive target
states. Then $P,Q,S$ are reversible for the same extended posterior, and
\begin{equation}
 Q=\epsilon P+(1-\epsilon)S. \label{eq:bridge-mixture}
\end{equation}
In particular $P(x,\cdot)\ll Q(x,\cdot)$ and
$P(x,A)\leq \epsilon^{-1}Q(x,A)$ for every measurable $A$.
If $D=P-Q$ and $\kappa(x)$ is their common crossing probability, then
\begin{equation}
 \tau(x):=\tfrac12|D(x,\cdot)|(E)
 \leq (1-\epsilon)\{1-\kappa(x)\}. \label{eq:bridge-tv}
\end{equation}
\end{proposition}

\begin{proof}
Fresh auxiliary pairs have joint law $M_\theta\otimes M_{\theta'}$,
which obeys the same exchange relation as $C$. The mixtures also obey
it. The same-event and cross-event route indicators are symmetric
under exchanging the two states. The accepted flux for each
construction is the minimum of the forward and reverse fluxes, which
proves detailed balance, including the diagonal rejection measure.
Unbiasedness gives the posterior parameter marginal.

On cross-event proposals the accepted measures in $P,Q,S$ coincide.
On same-event proposals the candidate measure of $Q$ is the indicated
mixture of those of $P$ and $S$. Multiplication by the common
$\alpha(x,y)$ preserves this relation. The holding mass is one minus
the accepted mass away from the current state, so it satisfies the
same relation. This includes rejection outside prior support and
rejection of zero likelihood estimates. Equation~\eqref{eq:bridge-mixture}
follows. Nonnegativity gives the measure domination.
Finally $D=(1-\epsilon)(P-S)$. The restrictions of $P$ and $S$ to the
current event class both have mass $1-\kappa(x)$; their crossing
restrictions agree. Their total variation distance is at most
$1-\kappa(x)$, proving \eqref{eq:bridge-tv}.
\end{proof}

The paired transcription implementation in Section~\ref{supp:paired-rccr} has $\epsilon=0.1$.
Proposition~\ref{prop:bridge-mixture} applies separately within the
original and residual proposal families. The mixture relation concerns
state distributions. Execution times depend on the implementation of
each component.

\subsubsection{A difference identity for the completed-prefix risk}

For method $K$, let $\mathsf H_K(x,dy,dc)$ be the joint conditional
law of the next state and its nonnegative charged cost, with state
marginal $K(x,dy)$. This assumption requires the state to contain
the information needed to predict the joint law. If timing depends
on additional history, that history must be included in the state.
Assume cumulative costs diverge almost surely from the states considered.
The completed-prefix convention includes zero-cost transitions under this
nonexplosion condition. If an implementation stops on reaching a checkpoint,
the same convention applies when transition costs are strictly positive.
For a rule that stops immediately when the remaining budget reaches zero,
declare that state terminal and set $V_K=\ell$ there. The formulas below
apply at active states; set $a_Q=\delta_Q=0$ at terminal states.

After initialization, write $z=(x,b,n,s)$, where $b\geq0$ is the
remaining budget, $n$ is the number of completed transitions, and
$s\in\{0,\ldots,n\}$ is their event sum. At $n=0$, set $s=0$. Set
\[
 \ell(z)=
 \begin{cases}(s/n-\Psi)^2,&n>0,\\(h(x)-\Psi)^2,&n=0.\end{cases}
 \qquad L_\Psi=\max\{\Psi^2,(1-\Psi)^2\}.
\]
At $n=0$, $x$ is the initialized state. Let $V_K(z)$ be the expected
loss at budget exhaustion. For $c\leq b$, put
$z_{y,c}=(y,b-c,n+1,s+h(y))$. Define
\begin{align}
 a_Q(z;y,c)&={\bf1}_{\{c\leq b\}}\{V_Q(z_{y,c})-\ell(z)\},
 \label{eq:bridge-gain}\\
 \delta_Q(z)&=\int a_Q(z;y,c)
       \{\mathsf H_P-\mathsf H_Q\}(x,dy,dc).
 \label{eq:bridge-residual}
\end{align}
The value of $V_Q(z_{y,c})$ is used only when $c\leq b$.

Run $P$ from $z$, and let $J$ be the index of the first transition
whose cost exceeds the remaining budget. Thus $J=N+1$, where $N$
is the number of additional completed transitions. Under the immediate
stopping rule, define $J=N+1$ as well and give the state reached at
the last completion residual zero. Let $Z_j$ be the augmented state
after $j$ completed transitions, with $Z_0=z$. After termination,
extend the path by an absorbing cemetery state with residual and
$V_P-V_Q$ both zero. This defines all stopped expressions below.

\begin{proposition}[Budget risk difference]
\label{prop:bridge-budget}
Under the preceding assumptions, for each integer $m\geq0$,
\begin{align}
 V_P(z)-V_Q(z)
 &=\mathbb E_{z,P}\sum_{j=0}^{\min(J,m)-1}\delta_Q(Z_j)
       +r_m(z),\label{eq:bridge-stopped}\\
 |r_m(z)|&\leq L_\Psi\Pr_{z,P}(J>m).\label{eq:bridge-tail}
\end{align}
Consequently,
\begin{equation}
 V_P(z)-V_Q(z)=\lim_{m\to\infty}
 \mathbb E_{z,P}\sum_{j=0}^{\min(J,m)-1}\delta_Q(Z_j).
 \label{eq:bridge-limit}
\end{equation}
If $\delta_Q\leq0$ at all reachable pre-transition states, then
$V_P(z)\leq V_Q(z)$.
\end{proposition}

\begin{proof}
The event average lies in $[0,1]$, so $0\leq V_K,\ell\leq L_\Psi$.
Conditioning on the next pair gives a bounded solution for the
continuation value. It is unique: after $m$ completed transitions,
the difference of two bounded solutions is bounded by their uniform
difference times the probability of still being active, which tends to zero.
The first-step recursion is
\[
 V_K(z)=\ell(z)+\int_{c\leq b}
       \{V_K(z_{y,c})-\ell(z)\}\mathsf H_K(x,dy,dc).
\]
Subtract the two recursions and set $W=V_P-V_Q$. This gives
\[
 W(z)=\delta_Q(z)+\int_{c\leq b}W(z_{y,c})\mathsf H_P(x,dy,dc).
\]
Iterating $m$ times yields \eqref{eq:bridge-stopped} with
\[
 r_m(z)=\mathbb E_{z,P}[{\bf1}_{\{J>m\}}W(Z_m)].
\]
Since $|W|\leq L_\Psi$, the stated tail bound follows. Nonexplosion
gives $J<\infty$ almost surely and hence the limit. Each finite
partial expectation is nonpositive under the sign condition, so its
limit is nonpositive. No interchange of an infinite series with an
expectation is used. If $\mathbb E J<\infty$, boundedness of
$\delta_Q$ also permits the usual expected-sum notation.
\end{proof}

The budget-crossing transition is represented by the terminal value
$\ell(z)$ and contributes no event value. An exact equality $c=b$
counts as completion. Zero-cost transitions are allowed under
nonexplosion. For a deterministic common cost, the recursion reduces
to the corresponding fixed-count risk. The finite-count continuation
identity is then recovered by expanding the first and second moments
of the remaining event sum.

\subsubsection{Separating state allocation from conditional execution cost}

Disintegrate the joint laws as
$\mathsf H_K(x,dy,dc)=K(x,dy)F_K(x,y,dc)$.
Suppose $P(x,\cdot)\ll Q(x,\cdot)$, as in
Proposition~\ref{prop:bridge-mixture}. Thus a version of $F_Q$
is defined for $P(x,\cdot)$-almost every next state. Write
\[
 G_Q(z;y)=\int a_Q(z;y,c)F_Q(x,y,dc).
\]
Adding and subtracting the measure $P(x,dy)F_Q(x,y,dc)$ gives
\begin{align}
 \delta_Q(z)&=\delta_Q^{\rm state}(z)+\delta_Q^{\rm cost}(z),
 \label{eq:bridge-split}\\
 \delta_Q^{\rm state}(z)
 &=\int G_Q(z;y)\{P-Q\}(x,dy),\\
 \delta_Q^{\rm cost}(z)
 &=\int P(x,dy)\int a_Q(z;y,c)
             \{F_P-F_Q\}(x,y,dc).
\end{align}
This is a decomposition with $Q$ as the specified reference method.
It keeps the dependence of cost on the next full state.

\begin{corollary}[Within-event continuation and cost variation]
\label{cor:bridge-bound}
Suppose $P$ and $Q$ have the same cross-event subkernel. Then
$\delta_Q^{\rm state}$ involves only next states in $E_{h(x)}$.
With total variation distance defined as half the variation norm,
\begin{align}
 |\delta_Q(z)|\leq
 &\tau(x)\,\operatorname{osc}_{E_{h(x)}}G_Q(z;\cdot)\nonumber\\
 &+L_\Psi\int P(x,dy)\,
      \operatorname{TV}\{F_P(x,y,\cdot),F_Q(x,y,\cdot)\}.
 \label{eq:bridge-local-bound}
\end{align}
The oscillation can be taken essentially with respect to $Q(x,dy)$
restricted to the event class; set it to zero if that restriction has
zero mass. The expected sum of the right side of
\eqref{eq:bridge-local-bound} along the first
$\min(J,m)$ attempted transitions under $P$, together with
$L_\Psi\Pr(J>m)$, bounds the absolute budget risk difference.
\end{corollary}

\begin{proof}
The signed measure $D(x,\cdot)$ has mass zero and is supported on
the current event class. Its positive and negative parts each have
mass $\tau(x)$. Subtracting a constant from $G_Q$ therefore gives
the first term. For fixed $z,y$, the function $a_Q(z;y,c)$ ranges
inside $[-\ell(z),L_\Psi-\ell(z)]$, including its value zero on
$c>b$. Its oscillation in $c$ is at most $L_\Psi$. Applying the same
signed-measure bound to $F_P-F_Q$ gives the second term. The final
statement follows from Proposition~\ref{prop:bridge-budget} and
the triangle inequality for its finite sum.
\end{proof}

The mixture relation supplies the explicit factor
$(1-\epsilon)(1-\kappa(x))$ in the first term. The cost term remains
even when $P=Q$. A joint equality on cross-event states and costs
would localize the entire residual to same-event outcomes. Equality
of the scientific state subkernels alone localizes its state term.
If an enlarged timing state is needed, both domination and the
common-crossing condition must be checked on that enlarged space
before applying this corollary. Proposition~\ref{prop:bridge-budget}
still applies to the augmented joint law without these extra conditions.

\subsubsection{Initialization and the observation-risk connection}

Let $\eta_K$ be the law of the initialized augmented state, including
its charged cost. Include a terminal state when initialization
exhausts the budget or reaches its prescribed failure limit, with
the specified fallback event loss. Extend $V_K$ to equal that loss
on terminal states. Then
\begin{align}
 \mathcal R_t(P)-\mathcal R_t(Q)
 &=\int V_Q\,d(\eta_P-\eta_Q)
   +\int(V_P-V_Q)\,d\eta_P.\label{eq:bridge-init}
\end{align}
The first term records initialization differences. The second is
given by Proposition~\ref{prop:bridge-budget}. The state and cost
terms in \eqref{eq:bridge-split} give a further decomposition under
its stated assumptions. Equal initialized parameter values or equal
initial likelihoods do not alone remove the first term.

Condition on the observed data and let $H^\star$ be a posterior event
draw, independent of the algorithmic randomization. For its completed
event average $A_K$,
\[
 \mathbb E\{(A_K-H^\star)^2\mid y\}
 =\Psi(1-\Psi)+\mathbb E\{(A_K-\Psi)^2\mid y\}.
\]
Thus \eqref{eq:bridge-init} also compares the conditional total
risks of algorithms using the same observation. Observation refinement
changes the posterior term and the computing problem. Event allocation
changes the second term within a chosen computing problem.

\subsubsection{A finite-budget check on cost reasoning}

Consider the chain on $\{0,1\}$ with transition matrix
$\left(\begin{smallmatrix}1/10&9/10\\9/10&1/10\end{smallmatrix}\right)$,
target event probability $\Psi=1/2$, and initial state $X_0=1$.
It is irreducible, aperiodic, and reversible for the uniform target.
For $f(x)=x-1/2$, the conditional mean is $Kf=(-4/5)f$ and
$f^2=1/4$. Thus $\mathbb E[f(X_i)f(X_j)]=(-4/5)^{j-i}/4$ for
$j\geq i$, including this fixed initialization. Direct expansion gives
fixed-count risks $R_2=1/40$ and $R_3=3/100$.
Initialization costs zero. At budget three, a transition cost of
$3/2$ completes two values, giving risk $1/40$. A cost of one
completes three values, giving risk $3/100$.
The state kernel and every fixed-count risk agree across implementations,
and every transition in the second implementation is cheaper.
The completed-prefix risk increases by $1/200$ at this budget. This example
identifies the extra condition needed for a speed argument: continuation
loss must have an appropriate ordering as the remaining budget changes.
The budget identity retains that condition through $V_Q$.

\paragraph{Identical marginals can give different budget risks.}

A two-state example shows why the joint law matters. Under both schemes,
the successive event values $Y_j$ are independent Bernoulli$(1/2)$
variables, $h(x)=x$, and $\Psi=1/2$. Initialization has zero cost.
In scheme A, transition $j$ costs $1+Y_j$. In scheme B, its cost is
independently drawn from $\{1,2\}$ with equal probabilities. The two
schemes have the same transition kernel, invariant target, and sequence
law of the costs. Consequently they also have the same completion-count
distribution at every budget.

At budget two, scheme A always yields a squared error of $1/4$.
A first simulated event value of one consumes the entire budget; a first zero
can be followed by another completed value only if it too is zero.
Under scheme B, the completion count is one with probability $3/4$
and two with probability $1/4$, independently of the event values.
Its risk is therefore $3/16+1/32=7/32$.
At budget three, direct enumeration gives risks $1/8$ and $29/192$
for A and B, respectively. For A, prefixes with nonzero loss have
probabilities $1/8$, $1/8$, and $1/4$: the event strings $000$, $001$
(the last value is uncompleted), and $11$ (the last value is uncompleted).
Each has loss $1/4$. For B, the completion-count probabilities for
one, two, and three values are $1/4$, $5/8$, and $1/8$; their risks
are $1/4$, $1/8$, and $1/12$.
Thus the risk ordering changes with the budget even though the
transition and cost marginals agree. Equation~\eqref{eq:bridge-residual}
retains the dependence that distinguishes the two schemes.

\section{Retained-state bounds and auxiliary validity}
\label{supp:proofs}

\subsection{A retained-state risk bound and its proof}
\label{subsec:functional-risk-obstruction}

Consider a broad class of fresh-estimator pseudo-marginal transitions.
Let \((\Theta,\mathcal B,\mu)\) be a \(\sigma\)-finite parameter space,
and let \(\pi\) be a probability density with respect to \(\mu\). Write
\[
\Theta_+
=
\{\theta\in\Theta:0<\pi(\theta)<\infty\}.
\]
For each \(\theta\), let \(Q_\theta\) be a measurable probability
kernel on \(\mathbb R_+\) satisfying
\[
\int w\,Q_\theta(dw)=1.
\]
Let \((\theta,w,\theta')\mapsto q_w(\theta,\theta')\) be jointly
measurable and, for every \((\theta,w)\), let
\(q_w(\theta,\cdot)\) be a probability density with respect to
\(\mu\). The pseudo-marginal extended target is
\[
\widetilde\pi(d\theta,dw)
=
\pi(\theta)wQ_\theta(dw)\,\mu(d\theta).
\]

From a retained state \((\theta,w)\), with
\(\theta\in\Theta_+\) and \(w>0\), propose
\[
\theta'
\sim
q_w(\theta,\theta')\,\mu(d\theta'),
\qquad
w'
\sim
Q_{\theta'}.
\]
The parameter proposal may depend on the retained weight. The proposed
weight is drawn afresh at \(\theta'\). Define the measurable
extended-valued function
\[
H(\theta)
=
\frac{1}{\pi(\theta)}
\int
\pi(\theta')
\left\{
\int
w'q_{w'}(\theta',\theta)Q_{\theta'}(dw')
\right\}
\mu(d\theta'),
\]
from \(\Theta_+\) to \([0,\infty]\).

\begin{theorem}[Retained-state finite-horizon functional obstruction]
\label{thm:functional-risk-obstruction}
Fix \(\theta\in\Theta_+\) and \(w>0\) with \(H(\theta)<\infty\), and
let \(A(\theta,w)\) be the total one-step acceptance probability from
\((\theta,w)\). Then
\[
A(\theta,w)
\le
\min
\left\{
1,\frac{H(\theta)}{w}
\right\}.
\]

Let \(\psi\in L^2(\pi)\),
\[
\Psi
=
\int
\psi(\theta)\pi(\theta)\,\mu(d\theta),
\qquad
\widehat\Psi_B
=
\frac1B
\sum_{t=1}^{B}\psi(\Theta_t),
\qquad
B\ge1.
\]
For the chain initialized at any state satisfying the preceding
conditions and \(|\psi(\theta)|<\infty\),
\begin{equation}
\label{eq:functional-risk-conditional}
\mathbb E_{(\theta,w)}
\left[
(\widehat\Psi_B-\Psi)^2
\right]
\ge
\{\psi(\theta)-\Psi\}^2
\left(
1-\frac{H(\theta)}{w}
\right)_+^B.
\end{equation}

If the chain starts from \(\widetilde\pi\), for \(\delta>0\) and
\(0\le C<\infty\), define
\[
D_{\delta,C}
=
\left\{
\theta\in\Theta_+:
|\psi(\theta)-\Psi|\ge\delta,\,
H(\theta)\le C
\right\}.
\]
Then, for every \(M>C\),
\begin{equation}
\label{eq:functional-risk-set}
\mathcal R_B(\psi)
\ge
\delta^2
\left(
1-\frac{C}{M}
\right)^B
\widetilde\pi
\left(
D_{\delta,C}\times[M,\infty)
\right).
\end{equation}
\end{theorem}

\paragraph{Proof of the retained-state bound.}
\label{supp:proof-functional-risk-obstruction}

Fix $\theta\in\Theta_+$ and $w>0$ with $H(\theta)<\infty$.
Let
\[
S_{\theta,w}
=
\{\theta'\in\Theta:q_w(\theta,\theta')>0\}.
\]
From $x=(\theta,w)$, the proposal draws
\[
\theta'\sim q_w(\theta,\theta')\,\mu(d\theta'),
\qquad
w'\sim Q_{\theta'}.
\]
On $S_{\theta,w}$, the Metropolis--Hastings ratio is
\[
R(x,x')
=
\frac{
\pi(\theta')w'q_{w'}(\theta',\theta)
}{
\pi(\theta)wq_w(\theta,\theta')
}.
\]
Therefore
\begin{align*}
A(\theta,w)
&=
\int_{S_{\theta,w}}
q_w(\theta,\theta')
\left\{
\int
\{1\wedge R(x,x')\}
Q_{\theta'}(dw')
\right\}
\mu(d\theta')
\\
&\le
\frac{1}{\pi(\theta)w}
\int_{S_{\theta,w}}
\pi(\theta')
\left\{
\int
w'q_{w'}(\theta',\theta)
Q_{\theta'}(dw')
\right\}
\mu(d\theta')
\\
&\le
\frac{1}{\pi(\theta)w}
\int_{\Theta}
\pi(\theta')
\left\{
\int
w'q_{w'}(\theta',\theta)
Q_{\theta'}(dw')
\right\}
\mu(d\theta')
\\
&=
\frac{H(\theta)}{w}.
\end{align*}
The first inequality uses $1\wedge r\le r$ on the forward support. The
second extends the integral from $S_{\theta,w}$ to $\Theta$ and uses
nonnegativity. Thus no mutual-support assumption is needed for this
upper bound. Since $A(\theta,w)\le1$,
\[
A(\theta,w)
\le
\min\left\{1,\frac{H(\theta)}{w}\right\}.
\tag{S1}
\]

The density calculation is shorthand for the same
Radon--Nikodym argument when the proposal kernels are expressed relative
to a common dominating measure.

A rejected pseudo-marginal Metropolis--Hastings proposal leaves the
complete extended state unchanged. Conditional on no previous
acceptance, every subsequent transition therefore starts again from
$(\theta,w)$ and has the same total acceptance probability
$A(\theta,w)$. Hence
\[
\Pr(
\text{no acceptance in transitions }1,\ldots,B
\mid
\Theta_0=\theta,W_0=w
)
=
\{1-A(\theta,w)\}^{B}.
\]
Using (S1),
\[
\Pr(
\text{no acceptance in }1{:}B
\mid\theta,w
)
\ge
\left(
1-\frac{H(\theta)}{w}
\right)_+^B.
\tag{S2}
\]

On the event in (S2),
\[
\Theta_1=\cdots=\Theta_B=\theta,
\qquad
\widehat\Psi_B=\psi(\theta).
\]
Since squared error is nonnegative,
\begin{align*}
\mathbb E[
(\widehat\Psi_B-\Psi)^2
\mid\theta,w]
&\ge
\{\psi(\theta)-\Psi\}^2
\Pr(\text{no acceptance in }1{:}B\mid\theta,w)
\\
&\ge
\{\psi(\theta)-\Psi\}^2
\left(
1-\frac{H(\theta)}{w}
\right)_+^B.
\end{align*}
This proves the conditional bound.

Now start the chain from
\[
\widetilde\pi(d\theta,dw)
=
\pi(\theta)wQ_\theta(dw)\,\mu(d\theta).
\]
Because $\psi\in L^2(\pi)$, the stationary risk
$\mathcal R_B(\psi)$ is well defined. On
$D_{\delta,C}\times[M,\infty)$, with $M>C$,
\[
\{\psi(\theta)-\Psi\}^2\ge\delta^2
\]
and
\[
\left(
1-\frac{H(\theta)}{w}
\right)_+^B
\ge
\left(
1-\frac{C}{M}
\right)^B.
\]
Integrating the conditional inequality over the stationary extended
target and restricting the integral to this set gives
\[
\mathcal R_B(\psi)
\ge
\delta^2
\left(
1-\frac{C}{M}
\right)^B
\widetilde\pi
\left(
D_{\delta,C}\times[M,\infty)
\right).
\]
This proves Theorem~\ref{thm:functional-risk-obstruction}.

\subsection{Event crossings and one-step conditional risk}
\label{supp:proof-event-crossing-risk}

Let $h$ take values in $\{0,1\}$, fix $X=(\theta,u)$, and write
$h_0=h(\theta)$ and $\Psi=\pi(h)$. For a transition kernel $K$, set
\[
\kappa_K(X)
=\Pr_{X,K}\{h(\Theta_1)\ne h_0\}.
\]
This is the probability of an accepted event crossing for a
Metropolis--Hastings transition. Conditional on $X$, the next event
value equals $h_0$ with probability $1-\kappa_K(X)$ and $1-h_0$ with
probability $\kappa_K(X)$. Hence
\begin{align*}
\mathcal R_1(K,\delta_X;h)
&=(1-\kappa_K(X))(h_0-\Psi)^2
  +\kappa_K(X)(1-h_0-\Psi)^2\\
&=(h_0-\Psi)^2
  +(1-2h_0)(1-2\Psi)\kappa_K(X).
\end{align*}
The second equality follows by subtracting the two squared losses.
For two kernels $K_a,K_b$ scored about the same $\Psi$ from the same
$X$, their conditional risk difference is therefore
\[
\mathcal R_1(K_a,\delta_X;h)-\mathcal R_1(K_b,\delta_X;h)
=(1-2h_0)(1-2\Psi)\{\kappa_{K_a}(X)-\kappa_{K_b}(X)\}.
\]
The coefficient is strictly negative when $h_0=0$ and $\Psi>1/2$,
or when $h_0=1$ and $\Psi<1/2$. These are precisely the event values
farther from $\Psi$. In either case, a strictly larger crossing
probability gives strictly smaller one-step conditional risk.
When $\Psi=1/2$, both event values have squared loss $1/4$.

An event crossing requires acceptance. For the fresh-estimator kernel
of Theorem~\ref{thm:functional-risk-obstruction}, this gives
\[
\kappa_K(\theta,w)\le A(\theta,w)
\le\min\{1,H(\theta)/w\}.
\]
Thus the retained-weight bound also bounds the probability of a
one-step change in the posterior event. For the route-conditioned
Metropolis--Hastings transition, the same probability is
\[
\kappa_K(X)
=\int_{\{\theta'\in\Theta:h(\theta')\ne h_0\}}
q_R(\theta'\mid X)
\left\{\int\alpha(X,X')
C^{\mathrm{inherit}}_{\theta,\theta'}(u,du')\right\}
\,d\theta'.
\]
This expression separates cross-event proposal mass and conditional
acceptance. Both contribute to the one-step mechanism endpoint.
Suppose two kernels use the same parameter proposal at $X$ and, for
each cross-event parameter proposal, the same auxiliary conditional
law and acceptance rule. Their displayed integrands agree, so their
$\kappa_K(X)$ values agree. Since $h$ is binary, this also gives the
same complete distribution of $h(\Theta_1)$ and the same one-step
conditional risk. The same-event auxiliary action may change the
accepted state within the event class and the work spent on the
transition. These changes can affect later transitions; the stated
identity concerns the first transition from the common state $X$.
At stationary initialization, any kernel preserving the same extended
target has $\mathcal R_1=\Psi(1-\Psi)$. The conditional comparison
above identifies the effect within each initial event class.

\subsection{Proof of Proposition 1}
\label{supp:r4-proof-validity}

Let
\[
E=\Theta\times\mathcal U,
\qquad
X=(\theta,u),
\qquad
X'=(\theta',u'),
\qquad
s=\mathsf r(\theta,\theta').
\]
For the off-diagonal balance calculation, let the in-support proposal
subkernel be
\[
G_{\mathrm{in}}(X,dX')
=
q_R(\theta'\mid X)
C^s_{\theta,\theta'}(u,du')
\,\mu(d\theta'),
\qquad \theta'\in\Theta.
\]
If the parameter proposal is generated on a larger ambient space, its
remaining mass
\[
\kappa(X)=1-G_{\mathrm{in}}(X,E)
\]
is assigned directly to the self-transition. The off-diagonal
detailed-balance calculation excludes this diagonal mass.

Define the forward in-support proposal-product measure on $E\times E$ by
\[
\nu(dX,dX')
=
\bar\pi(dX)G_{\mathrm{in}}(X,dX'),
\]
and let $\nu^\leftarrow$ be its image under the swap
$(X,X')\mapsto(X',X)$. Pair symmetry gives
\[
\mathsf r(\theta,\theta')
=
\mathsf r(\theta',\theta)
=
s.
\]
Using the auxiliary swap identity,
\begin{align*}
\nu(dX,dX')
&=
Z^{-1}
p(\theta)\widehat L(\theta,u)
q_R(\theta'\mid X)
M_\theta(du)
C^s_{\theta,\theta'}(u,du')
\mu(d\theta)\mu(d\theta')
\\
&=
Z^{-1}
p(\theta)\widehat L(\theta,u)
q_R(\theta'\mid X)
M_{\theta'}(du')
C^s_{\theta',\theta}(u',du)
\mu(d\theta)\mu(d\theta').
\end{align*}
On the common off-diagonal support of $\nu$ and
$\nu^\leftarrow$,
\begin{equation}
\label{eq:supp-r4-rn}
\frac{d\nu^\leftarrow}{d\nu}(X,X')
=
\frac{
p(\theta')
\widehat L(\theta',u')
q_R(\theta\mid X')
}{
p(\theta)
\widehat L(\theta,u)
q_R(\theta'\mid X)
}
=:R(X,X').
\end{equation}

For completeness, take the common dominating measure
\[
\xi=\nu+\nu^\leftarrow
\]
and write
\[
f=\frac{d\nu}{d\xi},
\qquad
g=\frac{d\nu^\leftarrow}{d\xi}.
\]
On $\{f>0\}$, define $R=g/f$, including $R=0$ when $g=0$; its value
on $\{f=0\}$ is immaterial to the forward proposal. With
\[
\alpha(X,X')=1\wedge R(X,X'),
\]
the accepted off-diagonal flow has density
\[
\alpha f
=
\min(f,g)
\]
with respect to $\xi$. This density is unchanged by interchanging
$X$ and $X'$. The remaining probability, including any out-of-support
proposal mass $\kappa(X)$, is assigned to the diagonal. The full Markov
kernel therefore satisfies detailed balance with respect to $\bar\pi$.

The extended target itself is well defined because
\[
\int_E
p(\theta)\widehat L(\theta,u)
M_\theta(du)\mu(d\theta)
=
\int_\Theta
p(\theta)L(\theta)\mu(d\theta)
=
Z.
\]
Integrating over $u$ gives
\[
\bar\pi(d\theta)
=
Z^{-1}p(\theta)L(\theta)\mu(d\theta),
\]
which is the desired posterior marginal.

The reversibility calculation above requires the extended target and
swap-balanced proposal structure. Unbiasedness identifies the
$\theta$-marginal with the exact posterior. States with
$\widehat L(\theta,u)=0$ have zero $\bar\pi$ mass; a proposed
zero estimator gives a zero Metropolis ratio and hence a
self-transition. Operational numerical exceptions terminate computation
and remain distinct from estimator values.

\subsection{Swap-symmetric couplings imply auxiliary balance}
\label{supp:r4-coupling-disintegration}

The swap identity used in Proposition~1 follows by disintegrating a
joint coupling. Fix a route $s$. Suppose there is a probability kernel,
measurable in the parameter pair $(\theta,\theta')$,
\[
J^s_{\theta,\theta'}(du,du')
\]
on $\mathcal U\times\mathcal U$ with first marginal $M_\theta$,
second marginal $M_{\theta'}$, and exchange property
\begin{equation}
\label{eq:supp-r4-Jswap}
J^s_{\theta,\theta'}(du,du')
=
J^s_{\theta',\theta}(du',du).
\end{equation}
For each parameter pair, the standard Borel property of $\mathcal U$
gives a regular conditional distribution
\[
C^s_{\theta,\theta'}(u,du')
\]
given its first coordinate. Choose a version that is jointly measurable
in $(\theta,\theta',u)$, as required for the complete proposal to be a
Markov kernel. The stagewise conditional simulations below provide such
a version: their rates, probabilities, and record maps are measurable
functions of the parameters and retained record. Hence
\[
J^s_{\theta,\theta'}(du,du')
=
M_\theta(du)
C^s_{\theta,\theta'}(u,du').
\]
Disintegrating the exchanged measure in
\eqref{eq:supp-r4-Jswap} gives
\[
M_\theta(du)
C^s_{\theta,\theta'}(u,du')
=
M_{\theta'}(du')
C^s_{\theta',\theta}(u',du),
\]
which is the auxiliary swap identity.

For independent refresh,
\[
J^{\mathrm{refresh}}_{\theta,\theta'}
=
M_\theta\otimes M_{\theta'},
\]
so the result is immediate. The inherited action below constructs the
required $J^{\mathrm{inherit}}_{\theta,\theta'}$ at the complete
particle-filter-output level.

\subsection{Split marked-jump exchange law}
\label{supp:r4-split-exchange}

Consider one propagation interval of the inherited cross-route
particle update. Let the two predictable guided reaction hazards be
$q=(q_j)_j$ and $q'=(q'_j)_j$. Each side computes its hazards at the
start of the observation interval and immediately after each of its
own reaction events. The guide uses that side's state and the time
remaining until the next observation. The resulting hazards remain
constant until that side's next event or the end of the interval.
An event occurring only on the other side leaves them unchanged.
This is the piecewise-constant proposal law used by the implementation
and by the path importance correction below. Define
\[
a_j=\min\{q_j,q'_j\},
\qquad
b_j=q_j-a_j,
\qquad
c_j=q'_j-a_j.
\]
The paired proposal can be represented by shared, first-only, and
second-only channel-$j$ point-process components with intensities
$a_j$, $b_j$, and $c_j$. With respect to the usual marked-jump reference
measure, a finite realization over the interval has density
\[
\mathcal D
=
\exp
\left[
-\int
\sum_j(a_j+b_j+c_j)\,dt
\right]
\prod_{k\in S}a_{j_k}(t_k)
\prod_{k\in F}b_{j_k}(t_k)
\prod_{k\in G}c_{j_k}(t_k),
\]
with the predictable, pre-event hazards evaluated along the paired
path. Exchanging the two sides leaves $a_j$ unchanged, interchanges
$b_j$ and $c_j$, and interchanges the first-only and second-only
marks. The state updates are exchanged by the same operation. Thus
the paired marked-path law is invariant under the simultaneous swap
of parameter labels and paths. Its first marginal has channel hazards
$a_j+b_j=q_j$, and its second marginal has channel hazards
$a_j+c_j=q'_j$.

The first-side path density is
$\exp\{-\int\sum_jq_j\,dt\}\prod_kq_{j_k}(t_k)$.
Dividing the paired density by this marginal density leaves the
proposed-only survival factor and the conditional shared-event marks.
Thus, conditional on the complete first-side path, proposed-only
channel-$j$ events occur with intensity
\[
q'_j-\min\{q_j,q'_j\},
\]
and, at a retained first-side channel-$j$ event with $q_j>0$, the
conditional probability that the event is shared is
\[
\frac{\min\{q_j,q'_j\}}{q_j}.
\]
Exact retained-path replay together with these conditional draws is
therefore a regular conditional simulation from the same
swap-symmetric paired path law. This construction is the split
coupling used for stochastic population processes
\citep{anderson2018split}. It assumes non-explosion of the guided
processes on the finite observation horizon.

\subsection{Particle weighting and marginal preservation}
\label{supp:r4-weight-exchange}

Let $\lambda_{\theta,j}$ denote the target reaction hazards and
$q_{\theta,j}$ the guided proposal hazards. For a proposed marked path
$\omega$, the continuous-time importance correction is a deterministic
measurable functional of the path and the side-specific parameter. In
the notation used here it has the standard form
\[
W_\theta(\omega)
=
G_\theta(\omega)
\exp
\left[
-\int
\{\lambda_{\theta,0}(t)-q_{\theta,0}(t)\}\,dt
\right]
\prod_k
\frac{
\lambda_{\theta,j_k}(t_k)
}{
q_{\theta,j_k}(t_k)
},
\]
where $G_\theta$ contains the observation contribution and
$\lambda_{\theta,0}$ and $q_{\theta,0}$ are the total hazards. This is
the usual importance-correction structure for conditioned-hazard
particle filters for Markov jump processes
\citep{golightly2014mjp}.

For Proposition~1, marginal validity and nonnegative unbiasedness of
the ordinary particle filter are requirements on $M_\theta$ and
$\widehat L$. The guides used here satisfy, at every hazard update,
\[
0.1\lambda_{\theta,j}\le q_{\theta,j}\le c\lambda_{\theta,j},
\qquad
c=5\ \text{for the gene network},\quad c=2\ \text{for Lotka--Volterra}.
\]
Unmodified channels have $q_{\theta,j}=\lambda_{\theta,j}$.
Thus $q_{\theta,j}=0$ exactly when $\lambda_{\theta,j}=0$.
Between a side's events, its state and target hazards are constant,
and the guide is frozen as specified above. The displayed weight is
therefore the target-to-proposal path density ratio multiplied by the
observation contribution. The common observation-time multinomial
resampling has the normalized particle weights as its categorical
marginal. Conditional expectation at each propagation and resampling
step then gives the usual nonnegative, unbiased particle-filter
normalizing-constant estimate.

The finite-horizon path laws are non-explosive in both models. In the
gene network, promoter indicators are bounded and the Hill response
lies in $[0,1]$. If $V$ is the sum of the four molecule counts, the
total target hazard is bounded by $C_\theta(1+V)$ for finite
$C_\theta$. Each reaction increases $V$ by at most one. The guide cap
therefore bounds its event count by a linear-rate pure-birth process.
For Lotka--Volterra, let $V=X_1+X_2$. Only prey reproduction increases
$V$, at guided rate at most $2c_1V$; predation preserves $V$ and
predator death decreases it. A linear-rate pure-birth process bounds
$V$ on every finite horizon. On each bounded set of populations all
three reaction rates are bounded, which rules out infinitely many
events while $V$ remains bounded. These arguments apply to the
ordinary and coupled marginals.

The coupling changes the joint law of two valid particle-filter
outputs while preserving both marginals.
Applying the two side-specific deterministic propagation and weight
maps equivariantly to the swap-symmetric paired path law preserves
exchange symmetry of the paired output. Each marginal remains the
ordinary guided particle-filter law at its own parameter.

The balance argument uses the mathematical path and ancestry laws.
The frozen network code evaluates them in floating-point arithmetic.
Its propagation loops use an absolute time tolerance of $10^{-15}$.
Raised consistency errors and event safeguards terminate execution.
In ordinary propagation, its categorical sampler can select a zero-rate
channel at a finite-precision endpoint; that branch returns a zero path
weight. Split propagation raises a support error on such a selection.
The ancestry tolerance is specified below. These numerical conventions
are part of the archived implementation.

\subsection{Transpose-symmetric ancestry coupling}
\label{supp:r4-index-exchange}

At a resampling step let $w=(w_i)$ and $v=(v_i)$ be the normalized
particle weights on the current and proposed sides. Define
\[
m_i=\min\{w_i,v_i\},
\qquad
\gamma=\sum_i m_i.
\]
For $\gamma<1$, set
\[
P_{ij}
=
m_i\mathbf 1\{i=j\}
+
\frac{
(w_i-m_i)(v_j-m_j)
}{
1-\gamma
}.
\]
Because
\[
\sum_j(v_j-m_j)=1-\gamma,
\qquad
\sum_i(w_i-m_i)=1-\gamma,
\]
the row and column marginals satisfy
\[
\sum_jP_{ij}=w_i,
\qquad
\sum_iP_{ij}=v_j.
\]
Moreover,
\[
P(v,w)=P(w,v)^\mathsf T.
\]
When $\gamma=1$, necessarily $w=v$, and the diagonal coupling
\[
P_{ij}=w_i\mathbf1\{i=j\}
\]
has the same marginal and transpose properties.

If the retained ancestor is $A=a$, then $w_a>0$ almost surely under
the retained categorical draw and
\[
\Pr(B=j\mid A=a)
=
\frac{P_{aj}}{w_a}
\]
is a regular conditional law of the coupled ancestry pair. Thus the
resampling operation has the correct multinomial marginal on each side
and reverses by transposition. This is the coupling property required
by coupled particle-filter constructions
\citep{jacob2016coupling}.

The frozen network implementation uses the residual formula when
$1-\gamma>10^{-15}$. At $1-\gamma\leq10^{-15}$, it returns the
current ancestor after the diagonal test. Conditional on a retained
ancestor $a$, the discrepancy from the ideal residual draw can be as
large as $(w_a-m_a)/w_a$. Averaging this bound under the fresh categorical
law $w$ gives $1-\gamma$. This bound concerns the fresh auxiliary law;
the retained extended-target law also weights the filter likelihood.
The transpose identity above is for the ideal matrix. The numerical
implementation follows the stated tolerance.

\subsection{Composition to the inherited particle-filter coupling}
\label{supp:r4-coupled-pf-composition}

Assume that the ordinary guided particle filter has the valid marginal
law $M_\theta$ and a nonnegative unbiased likelihood estimator,
retained replay is exact, the guided processes are non-explosive on the
finite observation horizon, and operational numerical exceptions
terminate computation.

Between successive observations, apply the split marked-jump coupling
of Supplementary Section~\ref{supp:r4-split-exchange}; apply each
side's deterministic weighting rule; and, at each common deterministic
resampling time in the frozen particle-filter schedule, apply the
transpose-symmetric ancestry coupling of
Supplementary Section~\ref{supp:r4-index-exchange}. The implementation
resamples at every nonfinal observation, using the same deterministic
schedule on both sides. Induction over observation times gives two
properties simultaneously:

\begin{enumerate}
\item each particle population has exactly the ordinary
particle-filter marginal at its parameter;
\item the complete paired particle-filter output is invariant under
$(\theta,u)\leftrightarrow(\theta',u')$.
\end{enumerate}

Consequently the completed inherited update induces a joint measure
\[
J^{\mathrm{inherit}}_{\theta,\theta'}(du,du')
\]
with marginals $M_\theta$ and $M_{\theta'}$ and with
\[
J^{\mathrm{inherit}}_{\theta,\theta'}(du,du')
=
J^{\mathrm{inherit}}_{\theta',\theta}(du',du).
\]
Supplementary Section~\ref{supp:r4-coupling-disintegration} then gives
a conditional inherited kernel satisfying the auxiliary swap identity
required in Proposition~1.

Zero likelihoods can be included by continuing a zero-weight side with
zero likelihood and a fixed legal particle rule, then mapping it to
the recorded zero-likelihood outcome. Using the same continuation
and recording maps on both sides preserves exchange symmetry. A
positive retained record is complete, and a zero-weight candidate has
zero extended-target mass. The validity statement concerns this
mathematical kernel in exact arithmetic; the preceding subsections
specify the frozen code's numerical conventions.

\subsection{Local cross-route residual ordering}
\label{supp:r4-proof-local-ordering}

Fix the current extended state \(X\) with a positive finite likelihood
estimate and a particular in-support cross-route parameter proposal
\(\theta'\). Suppose the prior densities and both baseline proposal
densities for this pair are positive and finite, and the centering
values are finite. Let
\[
p=p_0(\theta')\in(0,1)
\]
be the baseline probability of proposing back across the partition from
\(\theta'\), and define
\[
\phi_p(r')
=
r'
-
\min\{r',0\}
-
\log
\left[
(1-p)
+
p\exp\{-\min\{r',0\}\}
\right].
\]
This formula defines \(\phi_p\) for finite \(r'\). Extend it by
\(\phi_p(-\infty)=-\infty\). A zero proposed likelihood estimate has
residual \(r'=-\infty\) and is assigned acceptance probability zero.
The full proposed-residual law is thus a probability law on
\([-\infty,\infty)\).

\begin{lemma}[Local cross-route residual ordering]
\label{lem:one-step-action-ordering}
For fixed \(X\) and fixed cross-route \(\theta'\), the part of the log
Metropolis--Hastings ratio that varies with the proposed auxiliary state
equals \(\phi_p(r')\) up to a finite additive constant, with the
extended value \(-\infty\) for a zero estimate. The function
\(\phi_p\) is continuous on \(\mathbb R\) and strictly increasing on
\([-\infty,\infty)\).

Hence \(r'_a>r'_b\) implies
\[
\log R_{\mathrm{MH},a}
>
\log R_{\mathrm{MH},b},
\qquad
\alpha_a\ge\alpha_b.
\]
Conditional on the same \(X\) and \(\theta'\), first-order stochastic
dominance of one complete proposed-residual law over another,
including any mass at \(-\infty\), implies weakly larger expected
one-step acceptance.
\end{lemma}

\noindent\textit{Proof.}\quad

Fix the current extended state $X$ and a particular cross-route
parameter proposal $\theta'$. Terms depending only on $X$,
$\theta'$, the baseline parameter proposal, and the prior are then
constant across the auxiliary actions being compared.

Let
\[
p=p_0(\theta')\in(0,1)
\]
be the baseline reverse cross-route mass and let $r'$ be the residual
computed from the proposed auxiliary state. First take a positive
proposed likelihood estimate, so $r'$ is finite. For a reverse
cross-route proposal,
\[
\log q_R(\theta\mid X')
=
\log q_0(\theta\mid\theta')
-
\min\{r',0\}
-
\log
\left[
(1-p)+p e^{-\min\{r',0\}}
\right].
\]
Also,
\[
\log\widehat L(\theta',u')
=
m(\theta')+r'.
\]
The terms that vary with the proposed auxiliary state therefore reduce
to
\[
\phi_p(r')
=
r'
-
\min\{r',0\}
-
\log
\left[
(1-p)+p e^{-\min\{r',0\}}
\right].
\]

For $r'<0$,
\[
\phi_p'(r')
=
\frac{
p e^{-r'}
}{
(1-p)+p e^{-r'}
}
\in(0,1),
\]
while for $r'>0$,
\[
\phi_p'(r')=1.
\]
At $r'=0$, both branches equal zero, so $\phi_p$ is continuous. It is
strictly increasing on each open half-line and across zero, hence
strictly increasing on all of $\mathbb R$.

Thus $r'_a>r'_b$ implies
\[
\log R_{\mathrm{MH},a}
>
\log R_{\mathrm{MH},b}.
\]
The map
\[
x\longmapsto1\wedge e^x
\]
is nondecreasing, so
\[
\alpha_a\ge\alpha_b.
\]
Strict ordering of log ratios need not give strict ordering of
acceptance probabilities because both may be truncated at one.
For $r<0$, $\phi_p(r)=-\log\{(1-p)+pe^{-r}\}$, which tends to
$-\infty$ as $r\to-\infty$. The extension therefore preserves the
ordering, and the acceptance at $r=-\infty$ is zero.

For the stochastic-dominance statement, conditional on the same
$X$ and $\theta'$, write the one-step acceptance as
\[
g(r)
=
1\wedge
\exp\{K+\phi_p(r)\},
\]
where $K$ is constant across the compared auxiliary actions. The
function $g$, with $g(-\infty)=0$, is bounded and nondecreasing on
$[-\infty,\infty)$. If the complete law of residual $R_a$
first-order stochastically dominates that of $R_b$, then
\[
\mathbb E\{g(R_a)\}
\ge
\mathbb E\{g(R_b)\}.
\]
The lemma gives a conditional one-step ordering under its stated
stochastic-dominance premise. The prospective experiments compare the
resulting chains over multiple transitions.

\subsection{Moment bounds for the guided network likelihood estimate}
\label{supp:network-moments}

Consider the ideal-arithmetic, fresh particle filter for the network
with exact RNA-count observations at $L$ times, the last at time $T$.
All initial species counts are zero. Let $s$ bound the total RNA birth
intensity and let $d$ bound the RNA death coefficient. Only RNA birth
and death channels are guided. On their positive support, assume
$0<f\leq q_j/h_j\leq c$, with $f\leq1\leq c$; all other channels
have $q_j=h_j$. These bounds apply on each simulated holding interval.
The reaction intensities have a linear bound in the species counts.
The same bound holds for the proposal intensities, so the finite-time
processes are nonexplosive. For example, stopping at total count $m$
and applying the linear drift bound gives a bound on expected count
uniform in $m$; the probability of reaching $m$ before $T$ then tends
to zero.

Let $W$ be the nonnegative likelihood estimate with $N$ particles and
multinomial resampling. Define
\[
 a=\max\{1-f,c-1\},\qquad b=\max\{-\log f,\log c\},
\]
\[
 A=asT,\qquad B_*=adT+(L+1)b,\qquad \Lambda=NcsT.
\]
Let $S$ count RNA births across every simulated particle and interval
in a fresh filter record. It is stochastically dominated by a Poisson
variable with mean $\Lambda$. On $\{W>0\}$,
\begin{equation}
 N^{-L}e^{-A-B_*S}\leq W\leq e^{A+B_*S}.
 \label{eq:network-pathwise-weight}
\end{equation}
Consequently, for every $p>0$,
\begin{align}
 \E W^p&\leq \exp\{pA+\Lambda(e^{pB_*}-1)\},\nonumber\\
 \E[W^{-p}\mathbf1_{\{W>0\}}]
 &\leq N^{Lp}\exp\{pA+\Lambda(e^{pB_*}-1)\}.
 \label{eq:network-weight-moments}
\end{align}
Bounds that hold uniformly over the parameter support give the same
bounds after integration against the prior.

\begin{proof}
The total proposed RNA birth intensity per particle is at most $cs$.
The sum of simulated time across all particle intervals is $NT$.
Poisson thinning therefore gives the asserted domination of $S$,
including when an all-zero-weight filter stops early.

Write $B_i$ for births across the particles in interval $i$ and
$S_{i-1}$ for births in all preceding intervals. Resampling copies
existing states. Every copied particle has RNA count at most
$S_{i-1}$. On interval $i$, any one path has at most
$S_{i-1}+2B_i$ RNA jumps: each death removes either an initial molecule
or one born on that path. Its RNA count is at most $S$ throughout.
The absolute log importance weight on that interval is bounded by
\[
 a\Delta t_i(s+dS)+b(S_{i-1}+2B_i).
\]
This follows by bounding the integrated intensity difference and each
jump log ratio separately. Unguided channels contribute zero to both.
Summing the interval bounds, using
$\sum_i S_{i-1}\leq(L-1)S$, gives $A+B_*S$.

Exact-count observation weights are zero or one. If $W>0$, each
interval has at least one positive particle weight. Its average lies
between the smallest positive weight divided by $N$ and the largest
weight. Multiplying these bounds over intervals proves
\eqref{eq:network-pathwise-weight}. Finally use Poisson domination
and $\E e^{tZ}=\exp\{\Lambda(e^t-1)\}$ to obtain
\eqref{eq:network-weight-moments}.
\end{proof}

Let $\nu$ be the prior and fresh-record law, let $m=\E_\nu W>0$,
and let $\Pi=W\nu/m$. Exponential tilting of the Poisson mass function
also gives, for $u\geq0$,
\[
 \Pi(S>u)\leq
 \min\left\{1,
 \frac{e^{A+\Lambda(e^{B_*}-1)}}{m}
 \Pr\{\operatorname{Poisson}(\Lambda e^{B_*})>u\}\right\}.
\]
The moment bounds are structural statements about the fresh likelihood
estimator and its posterior-weighted law. A risk comparison further
requires quantitative control of the accepted event flow, the comparison
kernel and the initial law. Conditional proposals given a retained
record have their own laws. The count-model certificate in Section~S3
uses its explicitly specified binomial kernels.

The implementation's execution safeguards and floating-point arithmetic
are separate from this ideal-arithmetic argument. A safeguard abort
ends a run and must remain recorded as such; it is not a zero-likelihood
proposal.

\section{Route-conditioned implementation specification}
\label{supp:algorithm}

\subsection{Extended state and transition order}
The implemented transition stores the parameter \(\theta\), the retained
nonnegative likelihood estimate, and the compact particle information
needed by the inherited cross-route update. One transition proceeds in
the following order: evaluate the retained residual and route
normalizer; draw the parameter proposal from \(q_R\); determine the
pair route; generate the route-specific auxiliary proposal; evaluate
the proposed likelihood; evaluate the reverse route probability from
the proposed state; and apply the exact Metropolis--Hastings correction.
A rejection leaves the complete extended state unchanged.

\begin{algorithm}[!htbp]
\caption{Implementation form of the route-conditioned transition}
\label{alg:supp-dual}
\begin{algorithmic}[1]
\Require Current extended state \(X=(\theta,u)\)
\State Evaluate \(r(X)\), \(r_-(X)\), \(p_0(\theta)\), and \(Z_R(X)\)
\State Draw \(\theta'\sim q_R(\cdot\mid X)\)
\If{\(\theta'\) is outside the parameter support}
  \State \Return \(X\)
\EndIf
\State Set \(R\gets\mathsf r(\theta,\theta')\)
\If{\(R=\mathrm{cross}\)}
  \State Draw \(u'\sim C^{\mathrm{inherit}}_{\theta,\theta'}(u,\cdot)\)
\Else
  \State Draw \(u'\sim M_{\theta'}\)
\EndIf
\If{\(\widehat L(\theta',u')=0\)}
  \State \Return \(X\)
\EndIf
\State Set \(X'=(\theta',u')\) and evaluate \(q_R(\theta\mid X')\)
\State Set
\[
A=\frac{p(\theta')\widehat L(\theta',u')q_R(\theta\mid X')}
{p(\theta)\widehat L(\theta,u)q_R(\theta'\mid X)}
\]
\State With probability \(1\wedge A\), return \(X'\); otherwise return \(X\)
\end{algorithmic}
\end{algorithm}

A valid zero likelihood estimate is part of the mathematical estimator
law and therefore produces a self-transition. Raised numerical and
consistency exceptions terminate execution. The numerical tolerances
and zero-rate sampling branch in the original network implementation
are specified in Section~\ref{supp:r4-coupled-pf-composition}.

\subsection{Prospective centering surface and retained residual}
\label{supp:m0}
The gene-network prospective experiments use a centering rule calibrated
before method outcomes. The uniform prior region \([-2,2]^2\) is divided into
a \(16\times16\) equal-area grid and sampled twice per cell, giving 512
\(N=256\) particle-filter evaluations. A quadratic surface is fitted by
posterior-mass-weighted least squares over positive likelihood estimates and
the raw residual is centered by its posterior-mass-weighted median.

The resulting effective surface enters the transition. Write
\[
\bar m(z)=\bar\beta_0+\bar\beta_1z_1+\bar\beta_2z_2
+\bar\beta_3z_1^2+\bar\beta_4z_1z_2+\bar\beta_5z_2^2,
\]
with coefficients rounded to six decimal places,
\[
\begin{aligned}
\bar\beta_0&=-65.751455,\\
\bar\beta_1&=-0.962216,\\
\bar\beta_2&=-0.187107,\\
\bar\beta_3&=-0.464049,\\
\bar\beta_4&=\phantom{-}0.210266,\\
\bar\beta_5&=-0.386847.
\end{aligned}
\]
The exact coefficients and centering constant are stored in
\texttt{calibration\_freeze.json}. The retained residual is
\[
r(X)=\log\widehat L(\theta,u)-\bar m(z),
\qquad r_-(X)=\min\{r(X),0\}.
\]
\subsection{Analytic cross-route mass}
\label{supp:pcross}
Let \(d=z_2-z_1\), \(a=\log2\), and suppose
\(D'\mid d\sim\mathcal N(d,s_d^2)\) under the baseline random walk. If
the current state lies in the central band, \(|d|<a\), a cross route moves
to \(|D'|\ge a\), giving
\[
p_0(d)=\Phi\!\left(\frac{-a-d}{s_d}\right)
+1-\Phi\!\left(\frac{a-d}{s_d}\right).
\]
If \(|d|\ge a\), a cross route moves to \(|D'|<a\), giving
\[
p_0(d)=\Phi\!\left(\frac{a-d}{s_d}\right)
-\Phi\!\left(\frac{-a-d}{s_d}\right).
\]
The base proposal standard deviation is 0.5. For independent log
multipliers, \(s_d=\sqrt2\times0.5\); for the center--contrast
parameterization, the contrast-coordinate proposal has
\(s_d=0.5\). The normalizer
\(Z_R=1-p_0+p_0e^{-r_-}\) is consequently available in closed form.

\subsection{Split-hazard inheritance}
For each reaction channel \(j\), current and proposed conditioned
hazards \(q_j\) and \(q'_j\) are split into a shared component
\(\min(q_j,q'_j)\) and two residual components. The shared event stream
couples compatible reaction events; each residual stream supplies the
excess hazard of one marginal. This is the split-coupling construction
for stochastic population processes \citep{anderson2018split}. The
conditioned-hazard particle proposal carries the corresponding path
importance correction, following the informative-observation MJP
filtering construction of \citet{golightly2014mjp}.

\subsection{Coupled ancestry resampling}
At an observation time, normalized particle weights \(w\) and \(v\) are
coupled through the matrix derived in Supplementary Section~\ref{supp:r4-index-exchange}. Its
rows and columns recover the two multinomial marginals, and exchanging
the weight vectors transposes the matrix. The coupling is therefore
compatible with the swap identity proved in Supplementary Sections~\ref{supp:r4-proof-validity}--\ref{supp:r4-coupled-pf-composition} and with
the coupled-particle-filter framework of \citet{jacob2016coupling}.

\subsection{Computational accounting}
Section~\ref{supp:network-evaluation} defines the network clocks and
separates preparation, development and online computation. Section~\ref{supp:paired-risk}
specifies the transcription budget, which includes initialization.
The two model studies retain their respective timing conventions.

\section{Reproducibility specifications}
\label{supp:reproducibility}

This section specifies the models, numerical settings, and computational
accounting. The default reproducibility workflow recomputes the reported
summaries and figures from frozen trajectory-level results. The released
coupled kernel implements the conditional split/index law specified below.
Its propagation stream uses a declared reconstruction key; the historical
per-particle propagation sub-key is unavailable. This key gives a new
random stream for the same specified simulation scheme. Numerical
conventions are described with the auxiliary-law proof.
Numerical computations use NumPy \citep{harris2020numpy} and
SciPy \citep{virtanen2020scipy}.

\subsection{Controlled transcription comparison}

The latent state is
\[
X(t)=\{G_1(t),M_1(t),G_2(t),M_2(t)\},
\qquad
X(0)=(0,0,0,0).
\]
For allele \(j\), promoter activation, promoter deactivation, mRNA
synthesis, and mRNA degradation have hazards
\[
k_{\mathrm{on},j}(1-G_j),\quad
k_{\mathrm{off},j}G_j,\quad
s_jG_j,\quad
d_mM_j.
\]
The baseline rates are
\[
k_{\mathrm{on},1}=k_{\mathrm{on},2}=0.2,\qquad
k_{\mathrm{off},1}=k_{\mathrm{off},2}=0.5,\qquad
s_1=s_2=5,\qquad d_m=1.
\]
Inference uses
\[
\theta=
\bigl(\log m_{\mathrm{on},2},\log m_{s,2}\bigr),
\]
where the two components multiply \(k_{\mathrm{on},2}\) and \(s_2\).
The prior coordinates are independent truncated Gaussian variables with
mean zero, standard deviation \(0.75\), and support \([-4,4]\). The
frequency-reduction truth used in the matched inversion is
\[
\theta^\star=(\log0.3,0)
=(-1.2039728043259361,0).
\]
The fixed alternative is $\theta^{\mathrm{alt}}=(0,\log0.3)$.
The paired posterior-risk study uses the first saved record. All model
parameters and observations are held fixed across its computational methods.

The matched longitudinal schedules are
\[
t_i=i,\qquad i=1,\ldots,20,
\]
with binomial capture probability \(\rho=0.6\). For each matched
replicate, allele-specific and total-count observations share the same
latent trajectory and the same allele-level capture draws:
\[
Y_i^{\mathrm{AS}}=(Y_{1i},Y_{2i}),
\qquad
Y_i^{\mathrm{TC}}=Y_{1i}+Y_{2i}.
\]
Thus the designs have the same 20 observation events, one biological
trajectory, time horizon, and assay-event count; the allele-specific
record has 40 scalar entries and the total-count record has 20 because
the former retains the allele allocation.

The finite-state posterior quadrature, transcript cutoff checks and
reference perturbations are specified in Section~\ref{supp:paired-risk}.

\subsection{Feed-forward gene-network experiments}

The network state is
\[
X=(G_A,M_A,P_A,G_{B1},M_{B1},G_{B2},M_{B2}),
\]
with all components initialized at zero. Let
\[
R_A=P_A/50,\qquad H(R_A)=\frac{R_A}{1+R_A}.
\]
The downstream promoter activation rates are
\[
k_{\mathrm{on},Bj}(t)=0.05+k_{1,Bj}H\{R_A(t)\}.
\]
Table~\ref{tab:supp-network-reactions} gives the fourteen hazards used
in the experiment.

\begin{table}[H]
\centering\spacingset{1}
\caption{Reaction system for the feed-forward gene-network experiment.}
\label{tab:supp-network-reactions}
\begin{tabular}{lll}
\hline
Reaction & State change & Hazard\\
\hline
A promoter on & \(G_A:0\to1\) & \(0.2(1-G_A)\)\\
A promoter off & \(G_A:1\to0\) & \(0.5G_A\)\\
A mRNA birth & \(M_A\to M_A+1\) & \(5G_A\)\\
A mRNA death & \(M_A\to M_A-1\) & \(M_A\)\\
A protein birth & \(P_A\to P_A+1\) & \(10M_A\)\\
A protein death & \(P_A\to P_A-1\) & \(P_A\)\\
B1 promoter on & \(G_{B1}:0\to1\) &
\(\{0.05+k_{1,B1}H(R_A)\}(1-G_{B1})\)\\
B1 promoter off & \(G_{B1}:1\to0\) & \(0.5G_{B1}\)\\
B1 mRNA birth & \(M_{B1}\to M_{B1}+1\) & \(5G_{B1}\)\\
B1 mRNA death & \(M_{B1}\to M_{B1}-1\) & \(M_{B1}\)\\
B2 promoter on & \(G_{B2}:0\to1\) &
\(\{0.05+k_{1,B2}H(R_A)\}(1-G_{B2})\)\\
B2 promoter off & \(G_{B2}:1\to0\) & \(0.5G_{B2}\)\\
B2 mRNA birth & \(M_{B2}\to M_{B2}+1\) & \(5G_{B2}\)\\
B2 mRNA death & \(M_{B2}\to M_{B2}-1\) & \(M_{B2}\)\\
\hline
\end{tabular}
\end{table}

The single connected trajectory is observed at
\(t=1,\ldots,20\). The measured components are
\((M_A,M_{B1},M_{B2})\), and capture probability is one. The active
rates are parameterized by
\[
k_{1,B1}=0.4e^{z_1},\qquad
k_{1,B2}=0.4e^{z_2}.
\]
The generating point is
\[
(k_{1,B1},k_{1,B2})=(0.2,0.8),\qquad
(z_1,z_2)=(-\log2,\log2).
\]
The physical target uses a uniform prior in the independent log
multipliers on \([-2,2]^2\). The center--contrast representation
\[
c=(z_1+z_2)/2,\qquad d=z_2-z_1
\]
is its unit-Jacobian push-forward, with support
\(|c-d/2|\leq2\) and \(|c+d/2|\leq2\). The baseline Gaussian
random walk has coordinate standard deviation \(0.5\) in either
representation; out-of-support proposals are rejected before the
particle filter.

For the original observation record, the likelihood estimator is a
conditioned-hazard importance particle filter with \(N=256\),
multinomial resampling after every nonfinal observation, guide floor
\(0.1\) times the target hazard, and guide cap \(5\) times the target
hazard. The implementation safeguards are 500 event proposals per
particle interval and $3000N$ proposals per likelihood estimate.
Crossing a safeguard produces a fail-stop execution error; it is not
mapped to a likelihood value or an MH rejection.

\subsection{Gene-network reference and fixed-bank comparison}
\label{supp:prospective-cost-aware}
The primary posterior event probability is evaluated by an independent
pseudo-marginal importance-integration calculation on the prospective
gene dataset. The $N=256$ reference uses 11,264 draws and gives
\[
\widehat\Psi_{\rm ref}=0.6524107281355993,
\qquad \tau=\operatorname{MCSE}=0.009105433218058626.
\]
A separate $N=512$ calculation uses 2,048 draws. The reported MCSE
quantifies simulation precision of the importance-integration reference.
Its use below is a first-order normal approximation; self-normalization
and tail approximation errors are distinct from this MCSE.

The comparison fixes an empirical initial distribution with equal mass
on 96 distinct extended states. They were chosen before method outcomes
by weighted Gumbel-top-$k$ sampling without replacement from 6,144
$q(\theta)M_\theta(du)$ candidate evaluations. The candidate positive-weight
fraction is 0.99935, normalized-weight ESS is 931.65, and maximum
normalized weight is 0.009812. Forty-eight selected states use each
coordinate representation. The inferential calculation conditions
on the selected bank and its representations. The comparison, extension
and separate equal-budget experiment are summarized in
Section~\ref{supp:allocation-experiments}. The new development and
confirmation studies are specified in Section~\ref{supp:network-evaluation}.
Their saved result files retain all CPU and fixed-step checkpoints.

\section{Network development, confirmation and scope}
\label{supp:network-evaluation}

\subsection{Development of the frozen comparison}
The five-method pilot uses the 96-state bank described in
Section~\ref{supp:prospective-cost-aware}, with two paired repetitions
per state and method. It contains 960 trajectories. The random mixture
uses probability $0.5$ for the coupled action on each proposal. A separate
tuning comparison uses 48 states, with 12 in each representation--event
stratum and weights $(0.15625,0.34375,0.1979167,0.3020833)$.
One trajectory per state and each of seven candidates gives 336 paths.
The candidates are the original baseline, selective allocation, and
constant mixtures with probabilities $0,0.25,0.5,0.75,1$.
Figure~\ref{fig:network-development} displays the 25-second comparisons.
The selected mixture probability $0.25$ is fixed before confirmation.

\begin{figure}[htbp]
\centering\includegraphics[width=\textwidth]{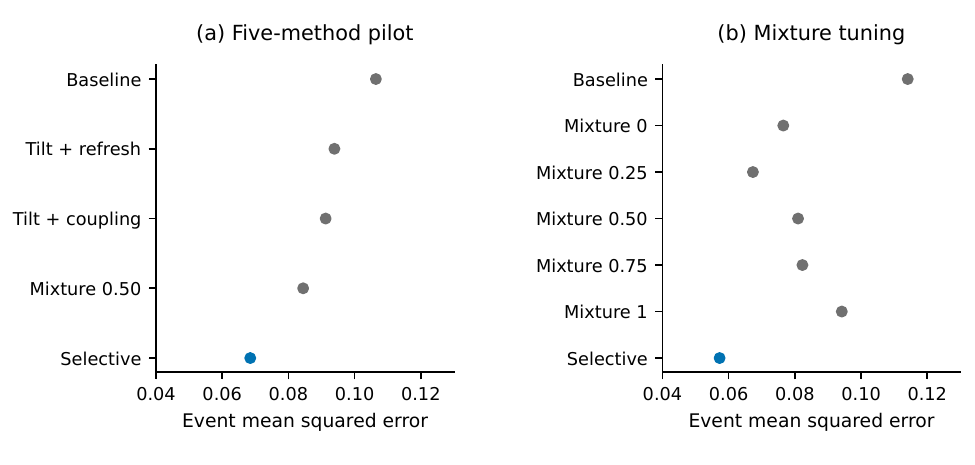}
\caption{Network development at 25 post-initialization CPU seconds.
Panel (a) uses 96 states and two repetitions per method; panel (b) uses
48 stratified states and one repetition. Points are descriptive MSEs
under the corresponding state weights. These development samples are
separate from the confirmation sample. All three CPU checkpoints and
four fixed-step checkpoints are retained in the frozen data files.}
\label{fig:network-development}
\end{figure}

\subsection{Initialization and primary inference}
The initialization proposal is $q_{\rm init}=(g_0+g_1)/2$, where
$g_0$ and $g_1$ are the two separately frozen importance proposals.
Each is a uniform/Gaussian mixture specified in the reproducibility
files. This is an independent proposal for preparing initial states;
transitions use $q_R$. The weighted-pool confirmation generates 2,048
candidates from $q_{\rm init}(\theta)M_\theta(du)$. Because the prior is
uniform, their importance weights are proportional to
$\widehat L(\theta,u)/q_{\rm init}(\theta)$. It samples 64 times with
replacement using the normalized weights. These draws select
58 distinct candidates; subsequent streams are independent across
draws. Inference conditions on this realized finite pool. The direct
panel samples candidate extended states from the same initialization
proposal and keeps the first positive likelihood estimate; all 64 initializations succeed on the first attempt.
Each initialization supplies two representations and two repetitions.
Both methods start from the same extended state in a paired repetition.
The two panels contain 1,024 complete trajectories and 22,220 transitions.

Let $X_{ir}$ and $Y_{ir}$ be the selective and mixture event averages
within initialization cluster $i$, with $r$ indexing representation and
repetition. Define
\[
 A_i=\frac14\sum_{r=1}^4(X_{ir}^2-Y_{ir}^2),\qquad
 B_i=\frac14\sum_{r=1}^4(X_{ir}-Y_{ir}).
\]
The paired cluster loss is $D_i(\psi)=A_i-2\psi B_i$.
At the reference endpoints $l=0.6289566864$ and $u=0.6758647698$,
we calculate the mean and standard error of the 64 cluster losses.
For $e\in\{l,u\}$ the endpoint interval is
\[
 I_e=\overline D(e)\ \pm\ t_{0.995,63}
       \frac{s_D(e)}{8},\qquad t_{0.995,63}=2.6561450.
\]
The reported panel interval spans both endpoint intervals. Affinity in
$\psi$ supplies the envelope between the endpoints. The reference
uses a 99\% normal interval. Its error budget $0.01$ and the four
endpoint error budgets $0.01$ give a joint nominal error budget of
$0.05$ across both panels. Coverage uses normal and Student approximations.
Its uncertainty calculation conditions on the generated finite pool.

The weighted-pool interval is $[-0.079253,-0.014823]$.
For direct initialization the interval is $[-0.027178,0.005377]$. Both upper endpoints must be negative
for the prespecified two-panel criterion; the second includes zero.
A separate 20,000-resample cluster-bootstrap sensitivity analysis gives
97.5\% intervals $[-0.073422,-0.023252]$ and $[-0.024870,0.000834]$.
These retain their role as sensitivity analyses.

\subsection{Fixed-step comparison and computational accounting}
Figure~\ref{fig:network-fixed-steps} holds the number of completed
transitions fixed. At one transition the two risk differences are
$0.003572$ and $0.017861$; at 3, 5 and 10 transitions they are negative.
The comparison uses a constant mixture and selective allocation, whose
cross-event kernels differ. The one-step equality for selective versus
full coupling therefore concerns a different pair of methods.

\begin{figure}[htbp]
\centering\includegraphics[width=\textwidth]{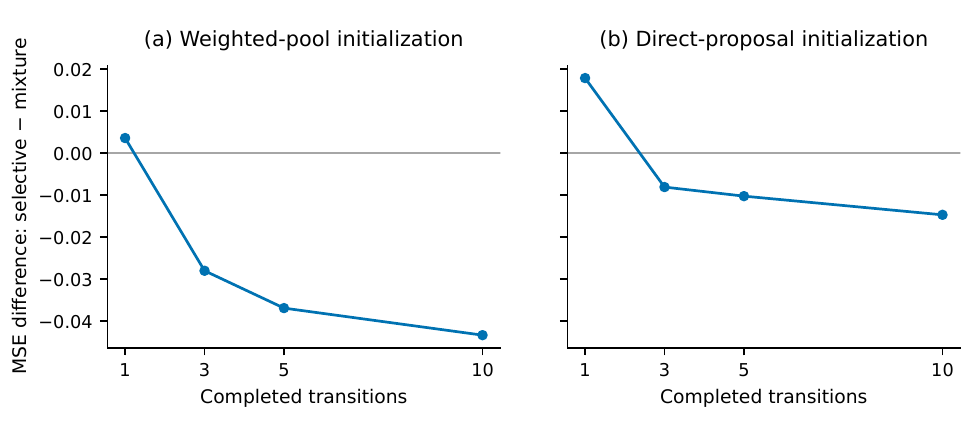}
\caption{Fixed-step risk differences in the network confirmation.
Each point is the selective-minus-mixture MSE, averaged over the same
64 initialization clusters as the main comparison. Lines connect the
four fixed horizons. These are descriptive secondary checkpoints;
primary inference concerns the 25-second CPU budget.}
\label{fig:network-fixed-steps}
\end{figure}

The charged online clock begins after state preparation and includes
input loading and transition work. An event average uses the completed
prefix at each checkpoint. A trajectory continues until it has crossed
25 CPU seconds and completed at least ten transitions. Its recorded cost
therefore includes work beyond the last budget prefix. Worker setup and
final serialization are outside this online clock.
Table~\ref{tab:network-costs} separates the measured stages. The summed
particle-filter elapsed times use a different clock from process CPU
time. Each preparation stage and reference calculation retains its
own cost record in the reproduction data.

\begin{table}[H]
\centering\spacingset{1}
\caption{Recorded network preparation, development and confirmation costs.}
\label{tab:network-costs}
\begin{tabular}{lrr}
\hline
Stage & Seconds & Clock\\
\hline
Shared calibration (512 evaluations) & 502.80 & Summed PF elapsed\\
Baseline calibration & 383.30 & Process CPU\\
Five-method development & 24,863.95 & Process CPU\\
Mixture tuning & 8,721.56 & Process CPU\\
Confirmation pool generation & 1,985.48 & Process CPU\\
Confirmation initialization & 134.47 & Process CPU\\
Confirmation method runs & 26,462.89 & Process CPU\\
\hline
\end{tabular}
\par\smallskip\begin{minipage}{\textwidth}\small Method-run costs cover all trajectories through their stopping rules. Confirmation totals 28,582.84 CPU seconds, with wall time 7,238.20 seconds. Earlier reference calculations, continuation diagnostics and initialization pilots have separate entries in the full cost ledger. These stage totals do not estimate an amortized cost advantage for a single new analysis.\end{minipage}
\end{table}

\subsection{Three additional network records}
The additional observations follow the same network and generating
parameters. Each panel has eight paired starts per record, with 32
trajectories per record across both panels and methods. The final
particle counts are 512, 1024 and 256. Each record uses its own frozen
calibration and importance reference. The mixture probability stays at
$0.25$. Pool and reference diagnostics determine the retained particle
settings before the reported comparison at those settings.

For pair $i$, write $D_i(\widehat\psi)=A_i-2\widehat\psi B_i$.
The error bars in main Figure~3 are $s_D/\sqrt8$, conditional on the
reference and the generated initial-state pool where applicable.
The reference MCSE contribution is
$2|\overline B|\operatorname{MCSE}(\widehat\psi)$.
Table~\ref{tab:cross-reference} varies the reference by 1.96 MCSE in
each direction. This separates two sources of uncertainty without
treating a reference-only range as a confidence interval. Each record's
complete three-budget paired estimates are supplied in the data files.

\begin{table}[H]
\centering\spacingset{1}
\caption{Reference sensitivity of the cross-dataset contrasts at 25 CPU seconds.}
\label{tab:cross-reference}
\begin{tabular}{lrr}
\hline
Data / initialization & Reference MCSE contribution & Reference-only range\\
\hline
1 / pool & 0.0095 & $[0.0481,0.0852]$\\
1 / direct & 0.0098 & $[0.0786,0.1168]$\\
2 / pool & 0.0100 & $[-0.1422,-0.1029]$\\
2 / direct & 0.0032 & $[0.0048,0.0171]$\\
3 / pool & 0.0047 & $[-0.0182,0.0003]$\\
3 / direct & 0.0053 & $[-0.0510,-0.0302]$\\
\hline
\end{tabular}
\par\smallskip\begin{minipage}{\textwidth}\small The MCSE contribution propagates the importance-reference uncertainty into the paired MSE difference. Ranges vary the reference by $\pm1.96$ reference MCSE while holding all saved trajectory estimates fixed. They exclude trajectory uncertainty and are not confidence intervals. Figure~3 of the main article displays the separate conditional start MCSE.\end{minipage}
\end{table}

\subsection{Earlier fixed-bank comparison}
\label{supp:allocation-experiments}
The earlier 96-state study compares selective allocation with full
coupling at the same $q_R$. Twelve paired repetitions per state and
method give an integrated selective-minus-full MSE of $0.003711$ over
horizons $1,3,5,10$, with approximate interval $[0.000603,0.006819]$.
A separate equal-budget experiment gives a 25-second difference
$-0.021362$, with interval $[-0.027497,-0.015227]$; the method MSEs
are $0.080080$ and $0.101442$. Both calculations condition on that fixed
bank. The latter comparison and the new initialization confirmation
have different comparators and initial laws. Their sample sizes and
uncertainty calculations are kept separate.

\section{Paired transcription posterior-risk experiments}
\label{supp:paired-risk}

\subsection{Target, initialization, and estimator}
The paired record consists of 20 total counts and the corresponding
allele-specific captured counts, with capture probability $0.6$.
The latent initial state is $(0,0,0,0)$. The prior coordinates are
independent $N(0,0.75^2)$ variables truncated to $[-4,4]$.
The event is $h(\theta)=\mathbf1\{\theta_1<\theta_2\}$.
The four parameter starts are $(-0.75,0)$, $(0,-0.75)$, $(0,0.75)$,
and $(0.75,0)$, with equal weights. They specify an initial distribution
for algorithm evaluation. The proposal covariance is
\[
 \Sigma=\begin{pmatrix}
 0.3792006236571427&-0.09789489409573962\\
 -0.09789489409573963&0.23672738321031908
 \end{pmatrix}.
\]
These parameters are fixed across channels and particle counts.
The bootstrap filter propagates the reaction process and retains
normalized particle weights when no resampling is performed.
Systematic resampling is selected when effective sample size is below
half the particle count. The likelihood estimate is the product of
weighted observation increments. The adaptive-resampling argument in
the observation-theory section verifies its unbiasedness for each channel.

The baseline and initial intervention use CPU budgets of 30, 60, 120,
and 240 seconds; 120 seconds is primary.
Charged intervals include initialization attempts, parameter mapping,
and transition operations. Imports, validation, scheduling, checkpoint
output, reference integration, and serialization are excluded. The event
average includes every completed post-transition state within the budget,
including repeated states after rejection. The transition crossing a
budget is recorded and excluded from that prefix. A prefix with no
completed transition returns the initial event value. Eight zero-estimate
initialization attempts trigger the same fallback, with failure flagged.
No burn-in is removed. Fixed-30-transition estimates are separate endpoints.

The formal baseline uses $N=600$ and 32 new paired repetitions per
start. The initial intervention uses 20 new groups per start, each
containing both channels at $N=600$ and $N=1200$.
Within an intervention group, proposal innovations and acceptance
uniforms are shared by transition index; particle-filter streams are
independent across counts, channels, and operations. Twenty complete
replicate blocks are randomly ordered, with random start and particle-count
order inside blocks. The two channels run concurrently with two workers.
The formal samples exclude their respective pilots.

The initial intervention continues until both 240 CPU seconds and 30 transitions
are reached, with guards of 10,000 transitions and 1,800 measured algorithm
wall seconds. All 320 trajectories finish; no initialization failure,
guard termination, or interruption occurs. There are 13,729 transitions
and 320 initialization operations. One proposal is outside prior support,
so the recorded PF count is 14,048. The 154 zero-likelihood estimates
produce ordinary rejections and remain in the record. The baseline
contains one interrupted pair after both channels had passed the
120-second endpoint; all 128 primary pairs are complete. Its 240-second
results retain the recorded interruption sensitivity.

\subsection{Independent precision confirmation}
A separate confirmation uses 310 four-cell groups and 1,240 trajectories.
The four starts receive $(n_1,n_2,n_3,n_4)=(40,110,70,90)$ groups.
The initial 320 trajectories provide the variance estimates used to
choose this allocation and are excluded from the new estimate.
All starts retain analysis weight $1/4$. Ten complete allocation
blocks each contain $(4,11,7,9)$ groups, randomly ordered within a block.
The four cells of each group are submitted in randomized order to four
workers. The next group starts after all four cells finish.
Scientific random streams are separate from the scheduling randomization.

Each trajectory stops after the first complete operation reaching
120 CPU seconds. The 30- and 60-second prefixes are retained.
The initialization fallback and the transition and algorithm-wall guards
are the same as in the initial intervention. This confirmation has no
240-second or fixed-transition extension. The sample size is fixed
before execution, with no stopping or extension based on effect signs
or interval widths.

All trajectories complete in one uninterrupted session.
There are 23,169 transitions, 1,240 initialization operations,
and 24,407 recorded PF calls. Two proposals fall outside prior support.
The 308 zero-likelihood candidates produce rejections; every
initialization succeeds on its first attempt. Total charged CPU is
153,103.234 seconds and elapsed session time is 39,591.693 seconds.
Four worker processes execute the cells, with numerical-library
thread counts set to one. The execution uses Python 3.14.7,
NumPy 2.5.2, and SciPy 1.18.0 on Windows.
CPU budgets use measured process time for each algorithmic operation.

A preceding fixed-parameter calibration replays 48 PF inputs with
two and four workers. Likelihoods, effective-sample-size diagnostics,
and resampling counts agree for every matched input. The wall-time
speed ratio is 2.031 and the total algorithm-CPU ratio is 0.982.
These are execution measurements. The new confirmation uses its
actual four-worker CPU costs; the historical two-worker samples
retain their separate role in planning and reporting.
Group boundaries record elapsed time, power status, and memory.
Every algorithmic operation records its worker, CPU time, and wall time.

\subsection{Fresh-estimator diagnostic}
A separate fixed-parameter bank contains 32 independent likelihood
estimates for each channel, particle count, and parameter start,
giving 512 evaluations. For the estimates $L_1,\ldots,L_{32}$ in
one cell, let $\overline L=32^{-1}\sum_{i=1}^{32}L_i$. We report
the sample relative variance
\[
 \widehat{\operatorname{RV}}
 =\frac{\sum_{i=1}^{32}(L_i-\overline L)^2}
 {31\overline L^2}.
\]
All 16 cells have positive sample means. The calculation uses a
common rescaling of the likelihood estimates within each cell.
Table~\ref{tab:particle-mechanism} gives the eight particle-count
comparisons. These evaluations are separate from the trajectory
runs and their charged CPU budgets.

\begin{table}[H]
\centering\spacingset{1}
\caption{Fresh-estimator relative variance at the four parameter starts.}
\label{tab:particle-mechanism}
\begin{tabular}{lrr}
\hline
Channel and parameter & $N=600$ & $N=1200$\\
\hline
Coarse: $(-0.75,0)$ & 0.080 & 0.065\\
Coarse: $(0,-0.75)$ & 0.143 & 0.098\\
Coarse: $(0,0.75)$ & 0.124 & 0.030\\
Coarse: $(0.75,0)$ & 0.056 & 0.034\\
Fine: $(-0.75,0)$ & 0.794 & 0.294\\
Fine: $(0,-0.75)$ & 0.855 & 0.513\\
Fine: $(0,0.75)$ & 0.291 & 0.106\\
Fine: $(0.75,0)$ & 0.201 & 0.070\\
\hline
\end{tabular}
\par\smallskip\begin{minipage}{\textwidth}\small Each channel--parameter--particle-count cell contains 32 independent likelihood evaluations. Entries are sample relative-variance estimates. These evaluations are separate from the risk trajectories and their charged budgets.\end{minipage}
\end{table}

\subsection{Paired uncertainty and risk definitions}
For a scalar paired contrast $D_{si}$ at start $s$, with $n_s$ replicates
at start $s$, define
\[
 \widehat D=\frac14\sum_{s=1}^4\overline D_s,\qquad
 \widehat{\operatorname{Var}}(\widehat D)
 =\frac1{16}\sum_{s=1}^4\frac{S_s^2}{n_s}.
\]
Let $v_s=S_s^2/(16n_s)$. The approximate interval uses
$t_{0.975,\widehat d}$, with
\[
 \widehat d=\frac{(\sum_s v_s)^2}{\sum_s v_s^2/(n_s-1)}.
\]
For the intervention, the four cells are first combined within each
start--replicate group. This preserves their pairing in the variance.
Intervals concern algorithmic repetition at the fixed data and starts.
Auxiliary intervals are unadjusted. Each complete allocation block first
averages the four start-specific means with equal weights. The ten
confirmation block means give a supplementary $t_9$ interval.

The baseline primary own-observation contrast is $j_F(f)-j_C(c)$.
For the intervention, the primary contrast is
\[
 I=\{e_{F,1200}-e_{F,600}\}-\{e_{C,1200}-e_{C,600}\}.
\]
Posterior variance cancels within each channel. The planned primary
interval half-width is $0.02$. The baseline achieves $0.017366$.
The initial intervention yields $0.037681$, with
$I=-0.014773$ and Monte Carlo standard error $0.018798$.
The independent confirmation gives
$I=0.011613$, standard error $0.008641$, and interval
$[-0.005405,0.028632]$. Its half-width, $0.017019$, is below the
specified $0.02$ target. The ten-block interval is
$[-0.008245,0.031471]$. Both intervals include zero. In the initial intervention, the fixed-30-transition interaction is $0.009331$, with interval
$[-0.014732,0.033395]$.

The common-full-observation comparison was added in the subsequent
conditioning audit. It uses
\[
 \overline{(A_F-p_F)^2-(A_C-p_F)^2}.
\]
Its evaluation reference is $p_F$; the coarse algorithm receives $c$.
The own-observation definition supplies main Table~2; the common-$F$ evaluation is reported here. The confirmation
common-$F$ contrasts are $-0.092139$ at $N=600$ and $-0.074753$
at $N=1200$. Their intervals are $[-0.109558,-0.074721]$ and
$[-0.093069,-0.056437]$, respectively.

The files \texttt{paired\_cells.csv} and \texttt{paired\_contrasts.csv}
supply every cell and contrast at 30, 60 and 120 seconds. They
retain the initial and confirmation samples separately. Main Table~2
reports the five principal contrasts at 120 seconds.

\subsection{Finite-state posterior reference}
The original quadrature integrates over the event and complement
triangles in parameter coordinates, using orders 24 and 28 for the
coarse and fine channels. The transcript cutoff is 64.
The probability references are $0.4877788868410538$ and
$0.7911916099887355$. A separate parameterization uses the prior CDF
identity in Section~\ref{supp:prior-cdf}. It smooths each unit-interval
coordinate by $H(s)=s^2(3-2s)$, with derivative $6s(1-s)$.
The transformed weight is $(1-u)H'(s)H'(t)w_sw_t$,
where $u=H(s)$ and $v=H(t)$. It absorbs both prior densities.
The integration sum therefore uses the likelihood times this weight.

The 20-order transformed rule has 800 nodes per channel, covers both
prior triangles, and gives probabilities $0.4877950044209498$ and
$0.7911930054990873$. Probability differences are $1.61\times10^{-5}$
and $1.40\times10^{-6}$. The absolute log-normalizer differences are
$3.90\times10^{-5}$ and $5.41\times10^{-6}$; both coordinate-mean
differences are below $2.36\times10^{-5}$. All pass the fixed thresholds
of $10^{-4}$ for event probability and $10^{-3}$ for the other quantities.

Earlier complete 8/12-order grids compare cutoffs 40 and 64, with maximum
probability difference $1.35\times10^{-10}$. At 14 selected locations
per channel, the 64/96 log-likelihood differences are at most
$1.42\times10^{-14}$. The new grid has 135 boundary-flagged nodes per
channel at threshold $10^{-8}$; their normalized quadrature masses
are $0.00039651$ and $0.00080469$. All nodes are retained.
The blocked finite-state generator conserves mass. These comparisons
are numerical diagnostics; they do not provide a certified error bound
for the infinite-state posterior. In particular, the boundary-flagged
mass is not such a bound.

\subsection{Prior-CDF quadrature}\label{supp:prior-cdf}
Let $F_0$ be the continuous marginal prior CDF on $[-4,4]$ and $Q_0=F_0^{-1}$.
For independent prior coordinates and event $x<y$, write
$x=Q_0(u)$ and $y=Q_0(u+(1-u)v)$, with $0<u,v<1$.
For any nonnegative likelihood $L$,
\[
 \int_{x<y}L(x,y)\,dF_0(x)dF_0(y)
 =\int_0^1\!\int_0^1 (1-u)
 L\{Q_0(u),Q_0(u+(1-u)v)\}\,dv\,du.
\]
The change of variables $a=F_0(x)$, $b=F_0(y)$ first removes both prior
densities. The triangular transformation $a=u$, $b=u+(1-u)v$ then has
Jacobian $1-u$, proving the identity. Swapping the coordinates gives the
complement integral. For $L=1$, both integrals equal $1/2$.
This transformation supplies the independent quadrature parameterization
used in the preceding reference check. Its agreement assesses numerical
stability under a change of integration coordinates.

\subsection{Reference perturbations and saved-trajectory sensitivity}
For an event estimator $A\in[0,1]$, let $\mu=\E A$ and
$e(p)=\E(A-p)^2$, with $p,p+\delta\in[0,1]$. Squaring gives
\[
 e(p+\delta)-e(p)=2(p-\mu)\delta+\delta^2.
\]
Adding posterior variance yields
\[
 j(p)=\E A^2+p(1-2\mu),\qquad
 j(p+\delta)-j(p)=\delta(1-2\mu).
\]
For two algorithms at a common posterior, the quadratic terms cancel,
so their MSE contrast changes by $2\delta(\mu_0-\mu_1)$.
These are exact identities for fixed output laws and also hold for
the saved empirical averages. A numerical difference between two
reference rules gives the corresponding sensitivity of each risk estimate.

Substituting the transformed-rule references gives a baseline
own-observation contrast of $-0.1003672085$, compared with the original
$-0.1003659461$. The initial intervention interaction changes from
$-0.0147730881$ to $-0.0147726410$.
The corresponding new intervals are
\[
 [-0.1177327303,-0.0830016867],\qquad
 [-0.0524531343,0.0229078523].
\]
The earlier $\pm0.001$ and $\pm0.01$ reference scenarios are retained
as additional sensitivity calculations. Reference integration and all PF diagnostic calculations
are outside the reported trajectory CPU budgets.

For the independent confirmation, the transformed references give
$I=0.0116138029$ and interval $[-0.0054045430,0.0286321488]$.
The interaction changes by $4.82\times10^{-7}$.
At the four reference corners with independent channel shifts of
$\pm0.001$, its estimates range from $0.011469$ to $0.011758$;
the largest interval half-width is $0.017069$.
With shifts of $\pm0.01$, the range is $[0.010170,0.013057]$ and
the largest half-width is $0.017541$. All checked interaction
intervals include zero. The coarse and fine particle-count intervention
intervals remain positive at these corners.

\section{Allocation for the paired record}
\label{supp:paired-rccr}

This study uses the paired record, prior, event $h(\theta)=
\mathbf1\{\theta_1<\theta_2\}$, and four equally weighted starts of
the preceding section. It connects the observation comparison to
allocation within each observation channel. The original independent
filter is retained as a baseline. The new auxiliary coupling preserves
the adaptive systematic resampling used for this record.

\subsection{The complete auxiliary law}

For channel $j$ and parameter $\theta$, let $M_{j,\theta}$ be the law
of a fresh bootstrap filter with $N=600$ particles. Its record contains
the reaction times and channels for each particle on every observation
interval, the resulting states, normalized weights, effective sample
sizes, ancestor vectors, and systematic offsets. Reaction times are
relative to the start of their observation interval. Removing the
additional recorded variables recovers the original fresh filter.

Resampling occurs when the effective sample size is below $N/2$.
When resampling is skipped, the normalized weights are retained.
For pre-observation weights $w_i$ and observation potential $g$, the
likelihood increment is $\sum_i w_i g(X_i)$. At every resampling
decision, include an offset $S\sim\operatorname{Unif}[0,1/N)$,
independent of the preceding record. An unused offset is generated
from a separate random stream. For weights $w$, define
\[
 A_w(S)_i=\min\left\{k:\sum_{l=1}^k w_l\geq S+(i-1)/N\right\},
 \qquad i=1,\ldots,N.
\]
Boundary equalities have probability zero. Each systematic offspring
count has conditional expectation $Nw_i$.

Let $\eta_t^N$ be the weighted empirical measure after observation $t$
and let $Z_t^N$ be the product of likelihood increments. Conditional
on the record through that weight update, the resampling decision is
known. Unbiased offspring counts, or the retained weights when no
resampling occurs, give
\[
 \mathbb E[Z_{t+1}^N\eta_{t+1}^N(\varphi)\mid\mathcal F_t]
 =Z_t^N\eta_t^N\{M_{t+1}(g_{t+1}\varphi)\}.
\]
Here $M_{t+1}$ denotes the latent propagation kernel, and the expectation
integrates the possible resampling after time $t$ and the next
propagation. Induction from the fixed initial state, with $\varphi=1$
at the last observation, proves
$\mathbb E_{M_{j,\theta}}\widehat L_j(\theta,U)=L_j(\theta)$.

\subsection{A joint law for the two auxiliary records}

For paired particle states $x,y$, write $a_k(x)$ and $b_k(y)$ for
their reaction rates under $\theta$ and $\theta'$. With reaction
increments $\nu_k$, use the joint generator
\begin{align*}
 \mathcal GF(x,y)={}&\sum_k\min(a_k,b_k)
    \{F(x+\nu_k,y+\nu_k)-F(x,y)\}\\
 &+\sum_k(a_k-b_k)_+\{F(x+\nu_k,y)-F(x,y)\}\\
 &+\sum_k(b_k-a_k)_+\{F(x,y+\nu_k)-F(x,y)\}.
\end{align*}
For a function of $x$ alone, the first two rates add to $a_k$.
For a function of $y$ alone, the first and third add to $b_k$.
The marginal processes therefore have the required reaction laws.
The generator is invariant under exchanging the two parameter--state
pairs. Promoter counts are bounded, synthesis rates are bounded at
fixed parameters, and degradation rates grow linearly in molecule
counts. These properties give nonexplosion on finite time intervals
for the marginal and joint processes.

The conditional implementation takes the current complete path as
given. Between its events, it simulates candidate-only reactions with
rates $(b_k-a_k)_+$. At a current event of type $k$, it synchronizes
the candidate with probability $\min(a_k,b_k)/a_k$, using pre-event
rates, and then updates the current state. The current path intensity
$a_k$ is independent of the candidate history. Thus its path-density
factor does not involve the candidate, and conditioning leaves exactly
the preceding candidate-only rates and synchronization probabilities.
Saved current event rates are positive, and replay must recover the
saved endpoint.

At resampling, the two sides share one offset $S$. Each side applies
its own ESS decision and uses either $A_w(S)$ or the identity ancestor
map. This defines a joint law for the complete ancestor vectors in
all four combinations of resampling decisions. Each marginal has
the correct systematic resampling distribution, and exchanging the
two sides leaves the construction unchanged.

\begin{proposition}[Exchange identity for the paired filter]
\label{prop:paired-record-exchange}
Alternating the joint propagation, each side's observation-weight
update, and the shared-offset resampling defines a joint record law
$J_{j,\theta,\theta'}$ with marginals $M_{j,\theta}$ and
$M_{j,\theta'}$. Its conditional kernel $C^{\rm split}$ satisfies
\[
 M_{j,\theta}(dU)C^{\rm split}_{j,\theta,\theta'}(U,dU')
 =M_{j,\theta'}(dU')C^{\rm split}_{j,\theta',\theta}(U',dU).
\]
\end{proposition}

\begin{proof}
At each propagation and resampling stage, a side's marginal law
depends on its own past and equals its fresh-filter stage law.
Induction over the finitely many stages gives the two claimed
marginals. Every joint stage is invariant under exchanging the
parameter labels and the records. Their composition has the same
symmetry. Finite nonexplosive paths, finite index vectors, and real
offsets form a standard Borel record space, so regular conditional
distributions exist. The explicit stagewise simulation gives a version
jointly measurable in the parameter pair and current record.
Disintegrating the symmetric joint law gives the identity. The path-conditional construction and shared offsets
above implement this disintegration at each stage.

To include a zero-weight record, first define a full simulation with
zero likelihood thereafter and a fixed legal continuation rule for
that side, such as uniform weights at subsequent stages. Use the same
rule on either side. The joint-stage construction still has the
required marginals and symmetry. Mapping each record to its prefix
through the first zero increment preserves that symmetry. A positive
target record is complete; conditioning on it and truncating a
zero-weight candidate therefore gives the implemented conditional
law. Such a candidate has zero extended-target mass and is rejected.
\end{proof}

The coupling used in the experiment is
\[
 C^\epsilon=0.9C^{\rm split}+0.1M_{j,\theta'}.
\]
The fresh component updates the whole auxiliary record, including
its offsets. The independent product of the two fresh laws is
exchange symmetric, so the mixture satisfies the same identity.
Selective allocation uses this mixture on cross-event proposals and
the fresh law on same-event proposals. Full coupling uses the mixture
on both routes. For a common parameter proposal, these methods have
the same cross-event subkernel and satisfy
Proposition~\ref{prop:bridge-mixture}.

These statements concern exact simulation and arithmetic. The numerical
implementation checks path validity, likelihood recursion, and the
acceptance ratio. A damaged record, nonfinite computation, or an
exceeded path guard is a failed run. These outcomes do not count as
ordinary likelihood-zero rejections.

\subsection{Frozen proposal and comparison design}

The original proposal $q_0$ is the existing Gaussian random walk.
The residual proposal $q_R$ multiplies its cross-event density by
$\exp(-r)$, where
$r=\min\{\log\widehat L_j(\theta,U)-m_j(\theta),0\}$.
Writing $p_0(\theta)$ for the original crossing probability, its
normalizing constant is
$1-p_0(\theta)+e^{-r}p_0(\theta)$.
For $d=\theta_2-\theta_1$, the proposal difference has variance
$s^2=0.8117177950589409$, so $p_0(\theta)=\Phi(-|d|/s)$.
The full acceptance ratio includes the nonuniform prior, the particle
likelihood, and both proposal densities. The reverse residual is
computed from the candidate record.

The fixed quadratic surfaces $m_j$ use the existing likelihood
quadrature, with 1,152 coarse-channel nodes and 1,568 fine-channel
nodes. Weighted least squares uses prior, likelihood, and quadrature
weights. The surfaces center the proposal residual; the particle
likelihood remains in the acceptance ratio. Both surfaces were fixed
before the risk pilot.

There are seven cells. Each channel includes the original independent
filter with $q_0$ and selective allocation with $q_R$. The fine channel
also includes all four combinations of $q_0$ or $q_R$ with full or
selective coupling. Initial parameters and proposal/acceptance seed
streams are paired within each start--replicate block. Filter streams
use the same seed assignments across methods within a channel;
realized trajectories can diverge after different decisions. Execution
uses two workers and randomized cell order within blocks.

The initial pilot has two replicates per start and cell, or 56
trajectories. One interrupted seven-cell block was replaced in full
with new seeds. The replacement rule and the new block were fixed
before comparing risks. The interrupted block remains in the archived
auxiliary results and is excluded from the primary pilot summaries.
The extension adds four replicates per start and cell, or 112
trajectories. All 112 completed their budget without an initialization
failure or a triggered guard.

The charged clock is process CPU time. It includes initialization of
the retained filter, transition computation, saved-record construction
and replay within those operations, and evaluation of the event.
Engine construction, compression, checkpoint writing, and separate
path audits are outside this clock. Initialization permits up to eight
attempts. The estimator averages completed post-transition event values
at 30, 60, and 120 seconds, including their preceding initialization
cost. If none completes, it uses the initial event value. The first
transition exceeding a checkpoint contributes cost information but no
event value to that prefix. The 120-second run stops once that budget
is reached; no operation in the extension had a zero recorded cost
or hit one of the checkpoints exactly. All summaries use this recorded
completed-prefix convention.

\subsection{Risk estimates and uncertainty}

Let $d_{hr}$ be a paired squared-error contrast at start $h$ and
replicate $r$. The reported contrast is
$\widehat\Delta=\frac14\sum_{h=1}^4\overline d_h$.
With $n$ replicates per start and within-start sample variance $s_h^2$,
the variance estimate and Welch degrees of freedom are
\[
 \widehat v=\frac{\sum_h s_h^2}{16n},\qquad
 \widehat\nu=(n-1)\frac{(\sum_h s_h^2)^2}{\sum_h s_h^4}.
\]
The intervals are approximate pointwise 95\% $t$ intervals conditional
on the fixed data and starts. They have no simultaneous-coverage
adjustment. The initial, extension, and combined summaries use
$n=2,4,6$, respectively. The quadrature event probabilities are the
same fixed references as in the observation study.
In the initial pilot, all eight observed selection contrasts under
$q_R$ at 30 seconds are zero. Their empirical variance is zero, so
the saved contrast file gives no $t$ interval for that cell.

Table~\ref{tab:rccr-extension} gives all seven planned contrasts
at the primary 120-second checkpoint in the extension. The full
cell means and contrasts are supplied as \texttt{rccr\_cells.csv} and
\texttt{rccr\_contrasts.csv}. They include every checkpoint in the
initial, extension and combined batches. The interaction is the
selective-minus-full contrast under $q_R$ minus that under $q_0$.

\begin{table}[H]
\centering\spacingset{1}
\caption{Allocation contrasts for the transcription extension at 120 CPU seconds.}
\label{tab:rccr-extension}
\begin{tabular}{lrr}
\hline
Contrast & Estimate & 95\% interval\\
\hline
Coarse: selective $q_R$ minus baseline & $-0.0148$ & $[-0.0758,0.0463]$\\
Fine: selective $q_R$ minus baseline & $-0.0005$ & $[-0.0784,0.0773]$\\
Fine: selective minus full, $q_0$ & $0.0209$ & $[-0.0183,0.0601]$\\
Fine: selective minus full, $q_R$ & $-0.0314$ & $[-0.0768,0.0141]$\\
Fine: $q_R-q_0$, full & $0.0354$ & $[-0.0270,0.0977]$\\
Fine: $q_R-q_0$, selective & $-0.0169$ & $[-0.1134,0.0795]$\\
Fine: selection-by-proposal interaction & $-0.0523$ & $[-0.1212,0.0165]$\\
\hline
\end{tabular}
\par\smallskip\begin{minipage}{\textwidth}\small Four repetitions per start and cell; 112 trajectories. Each contrast weights the four starts equally. Intervals are pointwise and conditional on the data, starts and numerical reference. All 30-, 60- and 120-second results from the initial, extension and combined batches are retained in the accompanying CSV files.\end{minipage}
\end{table}

At 120 seconds all seven extension intervals include zero. The fine
channel's RCCR-minus-baseline estimate is $-0.000531$ in the extension
and $-0.006218$ in the combined sample. At 30 seconds the extension
estimate is $0.157306$, with interval $[0.017252,0.297360]$; at 60
seconds it is $0.069620$, with interval $[-0.086410,0.225651]$.
These prefixes describe the time dependence of the completed event
average. The pilot supplies component comparisons; it does not
establish a channel-wide ordering of risk.

Initialization times also differ between batches. For matched
start--channel--method cells, the extension's mean initialization CPU
time is between 0.352 and 0.962 times the initial batch's mean.
These are different random replicates, so the ratio does not isolate
a cause. The combined sample describes a mixture with weights $1/3$
and $2/3$ for the two batches. A sensitivity calculation retains
separate start-by-batch variances with these fixed weights. At 120
seconds its intervals are $[-0.059059,0.067107]$ for the coarse baseline
contrast, $[-0.063092,0.050655]$ for the fine baseline contrast,
$[-0.063993,0.009797]$ for selection under $q_R$, and
$[-0.088392,0.009894]$ for the interaction.

For a contrast $\sum_k a_k(A_k-\Psi)^2$ within one channel, where
$\sum_k a_k=0$, shifting the reference to $\Psi+e$ changes its mean
by $-2e\sum_k a_k\mathbb E A_k$. The archived calculations evaluate
$e=\pm0.001$ and $\pm0.01$. These are sensitivity values, not certified
quadrature error bounds. The uncertainty reported here is simulation
uncertainty conditional on the numerical reference.

The extension recorded 2,905 filter calls, including 62 zero-likelihood
outcomes, and retained 1,996 complete filter records for path checks.
Their reaction paths, observation weights, likelihood recursion, and
acceptance ratios were checked against the frozen implementation.
The original independent baseline does not retain complete reaction
paths; it is checked from its saved transition and likelihood records.
The exchange identity above supplies the auxiliary-law argument.
The network experiment and this paired-record pilot retain their own
initialization laws, computing conventions, and inferential scopes.

\spacingset{1}
\bibliographystyle{agsm}
\bibliography{references_supplement}